\documentclass[12pt,reqno]{article}
\usepackage{times, bm}

\usepackage{amsmath}
\usepackage{amsthm}
\usepackage{amsfonts}
\usepackage{amssymb}
\usepackage{epic}
\usepackage{eepic}
\usepackage{epsfig}
\usepackage{mathtools}
\RequirePackage[dvipsnames]{xcolor}
\usepackage{multirow}
\usepackage{appendix}
\usepackage[colorlinks=true, allcolors=RoyalBlue]{hyperref}
\usepackage{dsfont}
\usepackage{cite}
\usepackage{mathabx}
\usepackage{enumerate}
\usepackage[scr]{rsfso}
\usepackage{verbatim}
\usepackage{lipsum}

\usepackage{macros}
\usepackage{graphics}

\usepackage{tikz}
\usetikzlibrary{decorations.markings, decorations.pathreplacing, decorations.pathmorphing, shapes.misc, knots}

\usetikzlibrary{calc}

\theoremstyle{plain}

\newtheorem{propn}{Proposition}
\newtheorem{lem}{Lemma}
\newtheorem{cor}{Corollary}

\theoremstyle{remark}
\newtheorem{rem}{Remark}

\numberwithin{equation}{section}

\newcommand{\ra}{\to}
\renewcommand{\Re}{\operatorname{Re}}
\renewcommand{\Im}{\operatorname{Im}}

\newcommand{\ga}{\gamma}

\newcommand{\de}{\delta}

\newcommand{\ep}{\epsilon}
\newcommand{\la}{\lambda}
\newcommand{\La}{\Lambda}

\newcommand{\si}{\sigma}

\newcommand{\vf}{\varphi}

\newcommand{\pa}{\partial}

\newcommand{\nco}{\newcommand}

\nco{\on}{\operatorname}

\newcommand{\CB}{{\mathcal B}}

\newcommand{\CN}{{\mathcal N}}

\newcommand{\CS}{{\mathcal S}}

\newcommand{\CW}{{\mathcal W}}

\newcommand{\CY}{{\mathcal Y}}

\newcommand{\SH}{{\mathsf H}}

\newcommand{\SQ}{{\mathsf Q}}

\newcommand{\ST}{{\mathsf T}}

\newcommand{\BR}{{\mathbb R}}

\newcommand{\BC}{{\mathbb C}}

\newcommand{\BP}{{\mathbb P}}

\newcommand{\ii}{\mathrm{i}}

\newcommand{\BZ}{{\mathbb Z}}

\DeclareMathOperator*{\Res}{Res}
\DeclareMathOperator{\Li}{Li}

\newcommand{\rf}[1]{\eqref{#1}}

\newcommand{\tp}{\tilde\varphi}

\newcommand{\hp}{\hat{\varphi}}
\usepackage{subcaption}

\DeclarePairedDelimiter{\abs}{\lvert}{\rvert}
\DeclarePairedDelimiter{\floor}{\lfloor}{\rfloor}
\DeclarePairedDelimiter{\ceil}{\lceil}{\rceil}

\DeclareMathOperator{\diag}{diag}
\newcommand{\SLNC}{\operatorname{SL}(N,\BC)}
\newcommand{\slNC}{\mathfrak{sl}(N,\BC)}

\newcommand{\Ypert}{\mathcal{Y}^\text{\tiny pert}}
\newcommand{\Yinst}{\mathcal{Y}^\text{\tiny inst}}

\begin{document}\thispagestyle{empty}
\title{Higher-Rank Mathieu Opers, Toda Chain, \\ and Analytic Langlands Correspondence}
\author{Jonah Baerman$^{1,*}$, Giovanni Ravazzini$^{2,\dagger}$, J\"org Teschner$^{2,3,\ddagger}$}
\address{$^{1}$ II. Institute for Theoretical Physics, University of Hamburg, Luruper Chaussee 149, \\
\phantom{$^{1}$} 22761 Hamburg, Germany,\\[2ex]
$^{2}$ Deutsches Elektronen-Synchrotron DESY, Notkestr. 85, 22607 Hamburg, Germany,\\[2ex]
$^{3}$ Department of Mathematics, University of Hamburg, Bundesstr. 55, 20146 Hamburg, Germany,\\[2ex]
E-mail: ${}^*$\href{mailto:jonah.baerman@desy.de}{\textit{jonah.baerman@desy.de}}, ${}^\dagger$\href{mailto:giovanni.ravazzini@desy.de}{\textit{giovanni.ravazzini@desy.de}}, ${}^\ddagger$\href{mailto:joerg.teschner@desy.de}{\textit{joerg.teschner@desy.de}}.}

\maketitle

\centerline{\bf Abstract}
\begin{quote}{\small
We study the Riemann-Hilbert problem associated to 
flat sections of oper connections
of arbitrary rank on the twice-punctured Riemann sphere with irregular singularities of the mildest type. 
We construct the solutions in terms of the solutions to a single non-linear integral equation. 
It follows from this construction 
that the generating function of the submanifold of opers coincides with the Yang-Yang function of the quantum Toda chain, proving a conjecture by Nekrasov, Rosly and Shatashvili.
In this way we may furthermore reformulate  
the quantization conditions of the Toda chain in terms of the connection problem, for which we also provide a solution.
We finally interpret our results as a variant of the Analytic Langlands Correspondence for the real version  
of the Hitchin system corresponding to the Toda chain. 
}
\end{quote}

\setcounter{tocdepth}{2}
\begingroup
\hypersetup{linkcolor=black}

\tableofcontents
\endgroup

\section{Introduction}

Our work describes connections between certain differential operators called opers
and the spectrum of the closed Toda chain, a paradigmatic example for an integrable many-particle problem. 
These connections are predicted by the relations to $\CN=2$, $d=4$ supersymmetric quantum field 
theories discussed in~\cite{Nekrasov:2009rc} and~\cite{Nekrasov:2011bc}. 
We are going to propose that one may regard these relations  from a mathematical perspective
as a variant of the analytic Langlands correspondence
\cite{Etingof:2019pni}. We will provide mathematical proofs of the conjectured relations, and use 
them to offer a construction of the solutions to the Riemann-Hilbert type problems naturally arising in this context.

\subsection{Riemann-Hilbert Correspondence for Some Higher Rank Irregular Opers}
We will, on the one hand, consider 
differential equations of the form 
\begin{equation} \label{eq:scalar-oper-intro}
    \left( \pa_z^N-u_2(z)\pa_z^{N-2}-\dots-u_N(z) \right)\chi(z)=0\,.
\end{equation}
In order to have analyticity on $C_{0,2}=\mathbb{P}^1\setminus\{0,\infty\}$ with 
irregular singularities of the mildest possible type at $0$ and $\infty$
one must have
\begin{equation}
    u_j(z) = \mathsf{u}_jz^{-j}\,,\qquad j=2,\dots,N-1\,,\qquad u_N(z) = \frac{\Lambda^N}{z^{N+1}} + \frac{\mathsf{u}_N}{z^N} + \frac{\Lambda^N}{z^{N-1}}\,.
\end{equation}
Such differential equations are known to be closely related to a class of holomorphic $\SLNC$-connections on 
$C_{0,2}$ called
opers.

A generalization of the Riemann-Hilbert correspondence associates differential equations of the 
form \rf{eq:scalar-oper-intro} to the generalized monodromy data, which include Stokes data 
in the presence of irregular singularities. In the particular case of interest
we are going to demonstrate that the Stokes data are completely determined by the eigenvalues of 
the monodromy of the simple closed curve $\ga$ separating $0$ and $\infty$. 
One may therefore formulate a variant of the Riemann-Hilbert 
problem naturally associated to the Riemann surfaces $C_{0,2}$:
\begin{quote}
(RH) {\it  For any given tuple of $N$ complex numbers $\si_1,\dots,\si_N$ satisfying 
$\sum_{k=1}^N\si_k=0$, find differential operators of the form 
$\pa_z^N-u_2(z)\pa_z^{N-2}-\dots-u_N(z)$ having coefficients $u_j(z)$ which are
holomorphic on $C_{0,2}$, for $j=2,\dots,N$,
irregular singularities of mildest possible type at 
$0$ and $\infty$, and monodromy eigenvalues $e^{2\pi\mathrm{i}\si_k}$, $k=1,\dots,N$.}
\end{quote}
We are going to show that the solutions to the Riemann-Hilbert (RH) type problem can be constructed from
the solutions to a single non-linear integral equation. 

We'd like to note that there exists 
another approach to the solution of this RH type problem using non-linear integral equations proposed 
in \cite{Gaiotto:2014bza} based on \cite{Gaiotto:2010okc,Gaiotto:2011tf}, see \cite{Hao:2025azt} for a
recent application to the case $N=2$ of the RH problem above. The 
integral equations that appear in this approach are {\it coupled systems} of non-linear integral
equations.
We hope that the new solution to this RH problem in terms of a single integral equation presented in this paper
can be useful both for numerical computations, and for more detailed investigations 
of the analytic properties of the solutions. 

\subsection{Relation to the Closed Toda Chain}

Inspiration for the solution to the Riemann-Hilbert type problem (RH) discussed in this paper
comes from the relation between 
the differential equation \rf{eq:scalar-oper-intro} and the spectrum of the closed Toda chain,
a paradigmatic quantum integrable model that has been studied intensively. It has been understood 
how to translate the quantization conditions of the Toda chain into a simple set of 
analytic conditions that solutions to the Baxter equation 
\begin{equation}\label{eq:Baxter}
   t(\lambda)q(\lambda) = \Lambda^N\left( \ii^Nq(\lambda +\ii\hbar) + \ii^{-N}q(\lambda-\ii\hbar) \right),
   \qquad t(\lambda)=\lambda^N+\sum_{k=0}^{N-1}\lambda^kE_{N-k}\,,
\end{equation}
have to satisfy in order for the coefficients $E_k$ appearing in the expansion of $t(\la)$
to be eigenvalues of the 
commuting conserved quantities of the closed Toda chain. In order to describe the 
resulting conditions more explicitly, one may construct two linearly independent 
solutions $Q^\pm$ to \rf{eq:Baxter} having $N$ free parameters $\de_1,\dots,\de_N$, 
and thereby translate the quantization 
conditions into conditions on these parameters which are as computable as the solutions 
$Q^\pm$ are. The construction of $Q^\pm$ given in~\cite{Kozlowski:2010tv} is based on 
the solution of a non-linear integral equation. It turns out that the sequences 
$(\chi_{j,n})_{n\in\BZ}$ with $\chi_{j,n}=Q^\pm(\de_j-\ii\hbar n)$
are the Laurent-coefficients of solutions to \rf{eq:scalar-oper-intro} having 
monodromy around $z=0$ represented by multiplication with $e^{-\frac{2\pi}{\hbar} \de_j}$.

\subsection{Yang-Yang Function as Generating Functions of Opers}

The relations between ${\cal N}=2$, $d=4$ supersymmetric quantum field theories and 
integrable models discussed in~\cite{Nekrasov:2009rc} suggest that 
the effective twisted
superpotentials $\mathcal{W}$ in $\mathcal{N}=2$, $d=4$ supersymmetric QFTs with partial Omega-deformation
should essentially coincide with the Yang's functions $\CY$ of the corresponding quantum integrable models, allowing   
one to describe the quantization conditions in terms of equations of the form
\begin{equation}\label{q-cond-Yang}
\frac{\pa}{\pa \de_k}\CY(\bm{\de},\Lambda)=2\pi \mathrm{i}\,n_k\in2\pi\ii\,\BZ\,, \qquad k=1,\dots,N\,.
\end{equation}
The terminology Yang's function for $\CY$ is motivated by the analogy to the potential 
for Bethe ansatz equation introduced by C.N. Yang and C.P. Yang in 1969. The instanton calculus 
for $\mathcal{N}=2$, $d=4$ supersymmetric QFTs leads in the case of pure $\mathcal{N}=2$ SYM 
to an expression for $\CW$ 
in terms of the solution to the same nonlinear integral equation as considered in our paper \cite{Nekrasov:2009rc,Meneghelli:2013tia}. 
A direct proof that the known quantization conditions of the closed Toda chain can be reformulated in 
the form \rf{q-cond-Yang} with $\CY$ directly related to $\CW$ was given in~\cite{Kozlowski:2010tv}.

The subsequent work~\cite{Nekrasov:2011bc} proposed a geometric description  
of the twisted
superpotentials  $\CW$ in terms of the  generating functions $\CS$ describing the submanifolds of opers within the 
spaces of monodromy data. This proposal has been generalized to class $S$ theories of higher rank in \cite{Hollands:2017ahy,Hollands:2019wbr,Hollands:2021itj}, emphasizing the role of abelianization \cite{Hollands:2013qza} as a natural
framework for the definition of the coordinates that are relevant in this context. 
The relations between $\mathcal{W}$ and $\CS$ conjectured in~\cite{Nekrasov:2011bc} have been established 
by quantum field theoretical methods 
in \cite{Jeong:2018qpc} in several examples for class $S$ 
theories with Lagrangian descriptions.  

By combining the conjectures of \cite{Nekrasov:2009rc} and \cite{Nekrasov:2011bc}
one arrives at a remarkable description of the Yang's functions of 
a large class of quantum integrable systems in geometric terms, simply expressed as a direct
relation between the Yang's functions $\CY$ of the integrable models, 
and the generating functions $\CS$ of the corresponding 
submanifolds of opers. 

In the particular case of pure $\mathcal{N}=2$ SYM one would, following \cite{Nekrasov:2011bc,Hollands:2017ahy,Hollands:2021itj}, consider 
a function $\CS$
satisfying 
equations of the form 
\begin{equation}
\frac{\pa}{\pa \si_j}\mathcal{S}(\bm{\si},\Lambda)=2\pi\mathrm{i}\,\eta_j(\bm{\si},\Lambda),
\end{equation}
with $(\si_1,\dots,\si_N;\eta_1,\dots,\eta_N)$ being coordinates for the  
Stokes and monodromy data of general holomorphic connections on $C_{0,2}$ having irregular
singularities of mildest type at $0$ and $\infty$, and $\eta_j(\bm{\si},\Lambda)$
being the values of the coordinates $\eta_j$ characterizing the half-dimensional 
oper locus. 
We are going to show that 
\begin{equation}\label{S-Y}
\mathcal{S}(\bm{\si},\Lambda)=\CY(\bm{\de},\Lambda)\Big|_{\bm{\de}=-\mathrm{i}\hbar\bm{\si}},
\end{equation}
with $\CY$ being the Yang's function of the closed Toda chain. This
proves the above-mentioned consequence of 
the conjectures of \cite{Nekrasov:2009rc} and \cite{Nekrasov:2011bc}
directly, without any reference to quantum field theory. 

It has furthermore been suggested in \cite{Nekrasov:2011bc} that 
the above-mentioned conjectures are related to the 
geometric Langlands correspondence between eigenvalue equations of quantized Hitchin systems and opers,
see \cite{Frenkel:2005pa} for a review. For the case of our interest one may note that 
the eigenvalue equations of the complexified version of the quantum Toda chain are related 
to opers with irregular singularities 
by a generalization of the geometric Langlands correspondence for Riemann surfaces 
with wild ramification,
as studied in \cite{FEIGIN2010873}.
Imposing, in addition, single-valuedness and normalizability of the eigenfunctions leads to the
analytic versions of the geometric Langlands correspondence which will be discussed next.

\subsection{Interpretation as a Variant of the Analytic Langlands Correspondence}\label{Intro-Langlands}

The analytic Langlands correspondence (ALC)~\cite{Etingof:2019pni} provides a geometric description 
of the spectra of 
the quantized Hitchin systems in terms of opers on Riemann surfaces.

Hitchin's integrable systems $\mathcal{M}_{\rm Hit}^{}(C,G)$ are classified 
by pairs $(C,G)$, with $C$ being a Riemann surface and $G$ being a complex semi-simple Lie group. 
Quantization of the Hitchin system yields a large family of quantum integrable models containing 
several known models such as the Toda chain as special cases. 
Of basic interest  is the spectral problem for the 
commuting Hamiltonians characterizing the integrable structures of these models. 

Oper connections are classified by pairs $(C,G')$, with $C$ being the Riemann surface they are
defined on, and $G'$ being the Lie group used to define of oper connections on $C$.
The opers that correspond to the eigenstates of the quantized Hitchin system $\mathcal{M}_{\rm Hit}^{}(C,G)$
according to the ALC are defined on the same Riemann surface $C$ as used to define the Hitchin system of interest,
but the Lie group $G'$ defining the opers corresponding to eigenstates of $\mathcal{M}_{\rm Hit}^{}(C,G)$ 
according to the ALC is the  Langlands dual 
Lie group $G'={}^{L}G$ of $G$. 

In practice one encounters two types of quantum spectral problems, differing by the choices of scalar products defining 
the underlying Hilbert spaces. For generic $C$
one may consider scalar products realizing the Hamiltonians as normal operators \cite{Etingof:2019pni}. 
If $C$ admits involutions allowing one to define real forms of the Hitchin systems, one may
alternatively consider a second type of scalar product, defined by integrating over {\it real} Lagrangian subspaces. 

So far most of the work on the analytic Langlands correspondence has considered the first type of 
spectral problems. 
The analytic Langlands correspondence for this type of problems, 
in the following referred to as (ALC)${}_\BC^{}$, predicts a one-to-one correspondence between
eigenstates and real opers, opers having monodromy conjugate to a real form ${}^{L}G_\BR$ of ${}^{L}G$ \cite{Te2018,Etingof:2019pni}.

However, it can also be interesting to study spectral problems of the second type, 
associated to real slices of the phase space
from this point of view. The closed Toda chain is arguably the 
simplest example for integrable models of this type. It is a real form of the Hitchin system
associated to the Riemann surface $C_{0,2}=\mathbb{P}^1\setminus\{0,\infty\}$ 
with irregular singularities of the mildest possible type at $0$ and $\infty$.\footnote{The Hitchin moduli space
is represented by Higgs fields of the same form as the right hand side of eq.~\rf{gen-connection} below.} 
We are going to prove:
\begin{quote}
{\rm (ALC)${}_\BR^{}$} {\it Eigenstates of the closed Toda chain are in one-to-one correspondence to opers on 
$C_{0,2}$ having 
canonical bases associated to the two irregular singularities related by 
parallel transport matrices proportional to the identity.}
\end{quote}
As the solutions to the spectral
problem are thereby classified by opers with simple conditions on the monodromy data,
one may regard the assertion above, referred to as
{\rm (ALC)${}_\BR^{}$}, as a natural variant of the 
analytic Langlands correspondence for spectral problems of the second 
type.\footnote{A general framework for variants of the analytic 
Langlands correspondence has been proposed in \cite{EFK4}.}

It seems quite remarkable that in both cases it is possible to express the quantization conditions in terms of 
the same function, the generating function $\mathcal{S}$ of the subspace of opers, 
\begin{align}
&\text{(ALC)${}_\BC^{}$ :}\quad\quad 
\mathrm{Re}\bigg(\frac{\pa}{\pa \si_j}\mathcal{S}(\bm{\si},\Lambda)\bigg)=2\pi\mathrm{i}\,n_j\,,\quad \mathrm{Re}(\si_j)=2\pi\mathrm{i}\,m_j\,, \label{ALCC-FN}\\
&\text{(ALC)${}_\BR^{}$ :}\quad\quad \frac{\pa}{\pa \si_j}\mathcal{S}(\bm{\si},\Lambda)=2\pi\mathrm{i}\,n_j\,,
\label{NS-q-cond}\end{align}
with $m_j,n_j\in\BZ$, see \cite{Te2018,Gaiotto:2024tpl} for the reformulation of (ALC)${}_\BC^{}$
in the form \rf{ALCC-FN}.
One may note, on the other hand, that the conditions implied by (ALC)${}_\BC^{}$
are equations for real-valued functions, 
whereas one finds holomorphic equations in the case of (ALC)${}_\BR^{}$.
The analytic Langlands correspondence appears to exchange reality conditions in a subtle way.

Further support for the proposal that (ALC)${}_\BR^{}$ represents a natural variant of the analytic 
Langlands correspondence is provided by the  references~\cite{Hatsuda:2018lnv, Bershtein:2021uts,Bonelli:2025bmt}, 
giving evidence that the 
quantization conditions 
of other prominent members of the family of Htichin systems, such as
the quantum Calogero-Moser models,  can also be represented in the form \rf{NS-q-cond}.

\subsection{Further Connections}

There exists an interesting variant of the relations between supersymmetric partition 
functions and quantization conditions of integrable models. It expresses the quantization conditions
of various spectral problems in terms of the isomonodromic tau-functions, which 
in turn are related to the instanton partition functions of the corresponding $\mathcal{N}=2$ SUSY QFTs by 
the correspondence discovered by Gamayun, Iorgov and Lisovyy \cite{Gamayun:2012ma,Gamayun:2013auu}. 
Background from theoretical physics for these relations 
is provided by the relations between
topological string theory and supersymmetric gauge theories
\cite{Sun:2016obh,Grassi:2016nnt,Grassi:2019coc,Gavrylenko:2020gjb}.
For the case of $N=2$ one thereby obtains 
a description of the spectrum of the Mathieu operator that is equivalent to the one discussed in this
paper, see \cite{Bershtein:2021uts} for a proof.

A more recent variant of these relations extends the scope of these methods from the quantization 
conditions to the eigenfunctions \cite{Francois:2025wwd,Francois:2025nmq}. This leads to 
exact expressions for the solutions to the Baxter equation \rf{eq:Baxter} of the closed Toda chain 
in terms of known building blocks from topological string theory. More general 
quantization conditions can be considered in this context~\cite{Grassi:2018bci,Francois:2025nmq}. 
It would be very interesting to compare these results with our approach.   

We'd furthermore like to note that the relation between the Baxter equation and opers is a 
special case of a phenomenon known as spectral duality in the literature
\cite{Mironov:2012ba,Mironov:2012uh}, relating the canonical quantizations of the spectral 
curves of (in general) different integrable systems. Providing the exact dictionary between
the conditions that have to be satisfied by the solutions to the respective equations, as done here,  
can be regarded as an analytic version of the spectral duality. It would be interesting 
to describe the quantization conditions of other pairs of quantum integrable models related 
by spectral duality in an analogous fashion. 

\subsection{Summary of Contents}

The following Section \ref{Toda-review} reviews the results on the spectrum of the closed Toda chain 
that will be relevant for us. Section \ref{oper-section} summarizes some background on opers
with irregular singularities of mildest possible type on $C=C_{0,2}$, and it formulates 
our main results on existence of solutions to the oper equation, and their monodromy data 
concisely.  Section \ref{Section:canonical basis} describes the proofs. 
Having a very special type of irregular singularity, we'll find that 
the monodromy and Stokes data can be parameterized by $N$ complex parameters $\si_1,\dots,\si_N$
satisfying $\sum_{j=1}^N\si_j=0$.

The key observation of this paper is contained in Section \ref{sec:Floquet}: There is a simple 
construction of bases for the solutions to the oper equation from a basis of two linearly
independent solutions of the Baxter equation of the closed Toda chain. The solution to 
the Riemann-Hilbert type problem (RH), and the relation 
between the Yang-Yang function $\CY$ and the generating function $\CS$ implied by  
the conjectures of \cite{Nekrasov:2009rc} and \cite{Nekrasov:2011bc} are immediate consequences. 
The connection matrix relating the two canonical bases at $0$ and $\infty$ is computed in 
Section \ref{sec:connection-langlands}, leading to the formulation of the variant 
(ALC)$_{\BR}^{}$ of the analytic Langlands correspondence given above. 
The appendices contain details of the proofs of some of the results used in the main text.

\section{Spectral Problem of the Closed Toda Chain} \label{Toda-review}

The $N$-particle, periodic Toda chain is defined by the Hamiltonian
\begin{equation} \label{eq:x-p-Hamiltonian}
    \SH = \sum_{j=1}^{N} \left(\frac{p_j^2}{2} +\Lambda^2e^{x_{j}-x_{j+1}}\right), \qquad x_{N+1} \equiv x_1\,,
\end{equation}
defining a differential operator on $L^2\left( \BR^N \right)$ by 
setting $p_j=\frac{\hbar}{\ii}\frac{\pa}{\pa x_j}$. Complete quantum integrability of the periodic Toda chain is reflected by the existence of a set $\{\SH_1,\dots,\SH_N\}$ of $N$ algebraically independent commuting operators containing $\SH$.
In order to fully establish complete integrability of the quantum Toda chain it is useful to introduce the Baxter Q-operator $\SQ(\lambda)$~\cite{Gaudin:1992ci} related to the conserved charges by the Baxter equation 
\begin{equation}\label{TQ-Transfermatrix}
   \ST(\lambda)\SQ(\lambda) = \Lambda^N\left( \ii^N\SQ(\lambda +\ii\hbar) + \ii^{-N}\SQ(\lambda-\ii\hbar) \right),
   \qquad \ST(\lambda)=\lambda^N+\sum_{k=0}^{N-1}\lambda^k\SH_{N-k}\,.
\end{equation}
The operators in  $\{\mathsf{Q}(\lambda)\,|\,\lambda\in\BR\}$ are mutually commuting. 
One may therefore refine the spectral problem of $\SH$ to finding the simultaneous 
spectral decomposition of the operators $\SQ(\lambda)$ for all $\lambda\in\BR$. Eigenstates of $\SQ(\lambda)$
are simultaneously eigenstates of $\ST(\lambda)$. Notations for the eigenvalues are defined through the expansions
 \begin{equation}\label{t-tau-def}
    t(\lambda)=\lambda^N+\sum_{k=0}^{N-1}\lambda^kE_{N-k} = \prod_{k=1}^N(\lambda-\tau_k)\,,
\end{equation}
where $E_k$ is the eigenvalue of $\SH_k$, and reality of $t(\lambda)$ for $\lambda \in \mathbb{R}$ requires $\{\tau_k\} = \{\bar\tau_{k}\}$.
It can be shown that $\SH_1 = -P$ and $\SH_2 = P^2/2-\SH$, where $P$ is the total momentum.
We will assume throughout that $E_1 = \sum_{k=1}^N\tau_k = 0$, such that $-E_2$ is the eigenvalue of the Hamiltonian $\SH$.

With the help of the Separation of Variables (SoV) it has been shown that 
\begin{quote}
{\it 
eigenstates of the quantum Toda chain are in a one-to-one correspondence with 
the solutions of the Baxter equation
\begin{equation}\label{Baxter-equation}
  t(\lambda)q(\lambda) = \Lambda^N \left( \ii^Nq(\lambda +\ii\hbar) + \ii^{-N}q(\lambda-\ii\hbar) \right),
\end{equation}
which are (i) entire, and (ii)
have asymptotic behavior 
\begin{equation}\label{q-asym}
    \abs{q^{\pm}(\lambda)} = O\Big( e^{-\tfrac{N\pi}{2\hbar}\abs{\Re\lambda}}\abs{\lambda}^{\tfrac{N}{2\hbar}(2\abs{\Im\lambda}-\hbar) } \Big)\,,
\end{equation}
uniformly in a strip of width $\hbar$ centered on the real line. }
\end{quote}
As the conditions (i) and (ii) are necessary and sufficient for existence of a square-integrable eigenfunction $\Psi_\mathbf{E}$ of $\SH_1,\dots,\SH_N$ having eigenvalues $E_1,\dots,E_{N}$, we will refer to them as quantization conditions. 

The characterization of the spectrum above is the result of a long series of works including \cite{Gutzwiller:1980yx,Gutzwiller:1981by,Gaudin:1992ci,Kharchev:1999bh,An}.
In order to show that to each solution to the conditions
(i) and (ii) there corresponds an eigenfunction one can use Sklyanin's Separation of Variables method (SoV),
see~\cite{Kharchev:1999bh} and references therein. The SoV method defines a unitary integral transformation mapping
products of solutions to the conditions (i) and (ii) to eigenstates~\cite{Kozlowski2015}.
The fact that the {\it all} eigenstates are obtained from solutions 
of conditions (i) and (ii) in this way has been proven in~\cite{An}. 

In order to describe the quantization conditions more explicitly, one may start by constructing 
families of solutions to the Baxter equation \rf{Baxter-equation} depending on $N$ parameters, 
and deriving a finite set of equations that have to be satisfied by the parameters in order to satisfy  the quantization conditions
(i), (ii). 
In the following two subsections we will briefly review the two known constructions of solutions to the quantization conditions. 

\subsection{Solutions to the Baxter Equations from Infinite Determinants}\label{sec:Baxter-det}

The first construction of two linearly independent solutions $Q^\pm_{\bm{\tau}}$ of the Baxter equation~\eqref{Baxter-equation}
was given in~\cite{Gutzwiller:1980yx,Gutzwiller:1981by}. 
In this construction, $Q^\pm_{\bm{\tau}}$ take the roots $\boldsymbol{\tau}\coloneqq (\tau_1,\dots,\tau_N)$ of the transfer matrix $t(\lambda)$~\eqref{t-tau-def} appearing in the Baxter equation as input data: explicitly, they are given by
\begin{subequations} \label{eq:Q+-}
\begin{align}
  \label{Q^+} Q^{+}_{\bm{\tau}}(\lambda) = \left( \frac{\hbar}{\Lambda} \right)^\frac{\ii N\lambda}{\hbar}\frac{e^{-N\pi \lambda/\hbar}K_{+}(\lambda)}{\prod_{k=1}^{N}\Gamma(1-\ii(\lambda-\tau_k)/\hbar)}\,,  \\ 
   \label{Q^-} Q^{-}_{\bm{\tau}}(\lambda) = \left( \frac{\Lambda}{\hbar} \right)^\frac{\ii N\lambda}{\hbar}\frac{e^{-N\pi \lambda/\hbar}K_{-}(\lambda)}{\prod_{k=1}^{N}\Gamma(1+\ii(\lambda-\tau_k)/\hbar)}\,,
\end{align}
\end{subequations}
where $K_{+}(\lambda)$ and $K_{-}(\lambda)$ are the unique solutions to the difference equations
\begin{subequations}
\begin{align}
      \label{K+eqation}K_{+}(\lambda-\ii\hbar) =  K_{+}(\lambda) - \frac{\Lambda^{2N} K_{+}(\lambda+\ii\hbar)}{t(\lambda)t(\lambda +\ii\hbar)}\,, \\ 
         \label{K-eqation}  K_{-}(\lambda+\ii\hbar) =  K_{-}(\lambda) - \frac{\Lambda^{2N} K_{-}(\lambda-\ii\hbar)}{t(\lambda)t(\lambda -\ii\hbar)}\,,
\end{align}
\end{subequations}
that tend to $1$ when $\lambda \to\infty$ uniformly away from the set of poles $\tau_k \mp\ii\hbar n, n \in \mathbb{N}$. These imply a representation in terms of the semi-infinite determinant
\begin{equation} \label{K+-det-definition}
    K_+(\lambda,\Lambda) = \det
    \begin{pmatrix}
        1 & \frac{1}{t(\lambda+\ii\hbar)} && \multicolumn{2}{c}{\qquad\;\mathbf{0}} \\
        \frac{\Lambda^{2N}}{t(\lambda+2\ii\hbar)} & 1 & \frac{1}{t(\lambda+2\ii\hbar)} \\
        & \frac{\Lambda^{2N}}{t(\lambda+3\ii\hbar)} & 1 & \frac{1}{t(\lambda+3\ii\hbar)} \\
        \mathbf{0} && \ddots & \ddots & \ddots
    \end{pmatrix},
\end{equation}
and $K_{-}(\lambda, \Lambda) = \widebar{K}_+(\widebar{\lambda}, \widebar{\Lambda})$, offering a useful starting point for analyzing convergence and  analytic properties of $K_\pm$, see~\cite[Appendix A.2]{Kozlowski:2010tv}.

While $Q^{\pm}_{\bm{\tau}}(\lambda)$ are entire, they have exponential growth in one direction along the real 
axis. In order to satisfy both quantization conditions (i) and (ii) above, one can use an ansatz of the form
\begin{equation} \label{q-ansatz}
    q(\lambda)= 
    \frac{Q^{+}_{\bm{\tau}}(\lambda)-\zeta Q^{-}_{\bm{\tau}}(\lambda)}{\prod_{j=1}^{N}e^{-\tfrac{\pi\lambda}{\hbar}}\sinh\tfrac{\pi}{\hbar}(\lambda-\delta_j)}\,, \qquad\de_j\in\BC\,,\quad j=1,\dots,N\,,
\end{equation}
ensuring validity of (ii), but sacrificing (i) for generic choice of $\de_1,\dots,\de_N$. In order to satisfy both (i) and (ii), one needs to choose 
the parameters $\de_j$, $j=1,\dots,N$, in such a way that the infinitely many 
poles coming from the denominator 
in \rf{q-ansatz} get canceled. In order to describe the resulting conditions, let us consider the infinite sequences 
\[
\mathfrak{q}_{\bm{\tau}}^\pm(\la)=\big(Q_{\bm{\tau}}^{\pm}(\la+\mathrm{i}\hbar n)\big)_{n\in\BZ}\in\BC^{\BZ}\,.
\]
It can be shown that for any given polynomial $t(\la)$ of the form \rf{t-tau-def}, and any $\Lambda\in\BC$, 
there exist $N$ values  $\de_1(\bm{\tau},\Lambda),\dots,\de_N(\bm{\tau},\Lambda)$ for $\la$, 
such that the vector $\mathfrak{q}_{\bm{\tau}}^+(\la)$
is proportional to $\mathfrak{q}_{\bm{\tau}}^-(\la)$. 
In order to see this, one may consider the 
quantum Wronskian defined by
\begin{equation}\label{Quantum-Wronskian}
    W[Q^{+}_{\bm{\tau}}, Q^{-}_{\bm{\tau}}](\lambda) \coloneqq  Q^{+}_{\bm{\tau}}(\lambda) Q^{-}_{\bm{\tau}}(\lambda+\ii\hbar) - Q^{-}_{\bm{\tau}}(\lambda) Q^{+}_{\bm{\tau}}(\lambda+\ii\hbar)\,.
\end{equation}
Vanishing of the quantum Wronskian for a given value $\la\in\BC$ implies that $\mathfrak{q}_{\bm{\tau}}^+(\la)$
is proportional to $\mathfrak{q}_{\bm{\tau}}^-(\la)$.
It can, on the other hand, be shown~\cite{Gutzwiller:1980yx, Gutzwiller:1981by} that $W[Q^{+}_{\bm{\tau}}, Q^{-}_{\bm{\tau}}](\lambda)$ has the form
\begin{equation}
    W[Q^{+}_{\bm{\tau}}, Q^{-}_{\bm{\tau}}](\lambda) = \left(\frac{ \hbar e^{-2\pi\lambda/\hbar}}{\ii\pi \Lambda}\right)^N \mathcal{H}(\lambda; \Lambda) \prod_{j=1}^{N}\sinh\left(\frac{\pi}{\hbar}(\lambda-\tau_j)\right),
\end{equation}
where $\mathcal{H}(\lambda;\Lambda)$ is known as the Hill determinant, defined by
\begin{equation}\label{Hill-det}
    \mathcal{H}(\lambda; \bm{\tau},\Lambda) = \det 
    \begin{pmatrix}
        \ddots & \ddots & \ddots && \multicolumn{2}{c}{\multirow[c]{2}{*}[3pt]{\textbf{0}}} \\
        & \frac{\Lambda^{2N}}{t(\lambda-\ii\hbar)} & 1 & \frac{1}{t(\lambda-\ii\hbar)} \\
        \multicolumn{2}{c}{\multirow[b]{2}{*}{\textbf{0}}} & \frac{\Lambda^{2N}}{t(\lambda)} & 1 & \frac{1}{t(\lambda)} & \\
        &&& \ddots & \ddots & \ddots \\
    \end{pmatrix}.
\end{equation}
As $\mathcal{H}(\lambda; \bm{\tau},\Lambda)$ is periodic in $\la$, 
$\mathcal{H}(\lambda+\mathrm{i}\hbar; \bm{\tau},\Lambda)=\mathcal{H}(\lambda; \bm{\tau},\La)$, 
and tends to $1$ when $\Re\la\ra\pm \infty$ for fixed $\Im\la$, it must be of the form 
\begin{equation} \label{Hill-formula}
\mathcal{H}(\lambda; \bm{\tau},\Lambda)=\prod_{j=1}^{N}\frac{\sinh\tfrac{\pi}{\hbar}(\lambda-\delta_j(\boldsymbol{\tau},\Lambda))}{\sinh\tfrac{\pi}{\hbar}(\lambda-\tau_j)}\,,
\end{equation}
defining the $N$ functions $\de_j(\bm{\tau},\Lambda)$, $j=1,\dots,N$. 

The quantization conditions (i) can now be satisfied with the ansatz \rf{q-ansatz} if and only if the 
constants $\xi_j(\bm{\tau},\Lambda)$ appearing in the relations 
$\mathfrak{q}_{\bm{\tau}}^+(\de_j)=\xi_j(\bm{\tau},\Lambda)\mathfrak{q}_{\bm{\tau}}^-(\de_j)$
do not depend on the value of $j\in\{1,\dots,N\}$. This set of conditions restricts $\bm{\tau}$ severely. 
For any given solution $\bm{\tau}$ to these conditions, one may for $j=1,\dots,N$ choose the parameters $\de_j$ in the ansatz \rf{q-ansatz} to be equal to $\de_j(\bm{\tau},\Lambda)$ in order to ensure validity of the quantization conditions (i), (ii) above.

Let us note that $Q^{\pm}_{\bm{\tau}}(\la)$ decay rapidly when $\mathrm{Im}(\la)\ra\pm\infty$ respectively, 
behaving as $Q^{\pm}_{\bm{\tau}}(\lambda) \sim \abs{\lambda}^{-\frac{N\abs{\lambda}}{\hbar}}$. 
As $\mathfrak{q}_{\bm{\tau}}^+(\de_j)=\xi_j(\bm{\tau},\Lambda)\mathfrak{q}_{\bm{\tau}}^-(\de_j)$ for 
all $j=1,\dots,N$, it follows that the sequences 
$\mathfrak{q}_{\bm{\tau}}^\pm(\de_j(\bm\tau,\Lambda))=\big(Q_{\bm{\tau}}^{\pm}(\de_j(\bm\tau,\Lambda)+\mathrm{i}\hbar n)\big)_{n\in\BZ}$
have rapid decay for both $n\ra\infty$ and $n\ra-\infty$.

\subsection{Solutions to the Baxter Equation From a Non-Linear Integral Equation}\label{BaxterfromNLIE}

An alternative construction of two linearly independent entire solutions to the Baxter equation \rf{Baxter-equation} was used in~\cite{Kozlowski:2010tv}. It is based on the non-linear integral equation
\begin{equation}\label{NLIE}
     \log X_{\bm{\delta}}^{}(\mu) = -\log\Theta(\mu)+\int_\BR \frac{d\mu'}{2\pi\ii}\frac{2\ii\hbar }{(\mu-\mu')^2+\hbar^2}\log( 1+ {X_{\bm{\delta}}^{}(\mu')})\,,
\end{equation}
where 
\begin{equation} \label{Theta-def}
    \Theta(\mu) = \Lambda^{-2N}\vartheta(\mu-\ii\hbar/2)\vartheta(\mu+\ii\hbar/2)\,,  \qquad
    \vartheta(\mu) = \prod_{k=1}^{N}(\mu-\delta_k)\,.
\end{equation}
Existence of a solution to \rf{NLIE} was shown in~\cite[Proposition 3]{Kozlowski:2010tv} for 
sets $\bm\de\coloneqq\{\de_1,\dots,\de_N\}$ which are invariant 
under complex conjugation, satisfy $\abs{\mathrm{Im}(\de_k)}<\hbar/2$ for all $k=1,\dots,N$, and small enough values of 
$\Lambda$. We may note that solving \rf{NLIE} by iteration yields 
a series in powers of $\Lambda^{2N}$. Whenever this series is absolutely convergent, one may conclude that
the dependence of the solutions $X_{\bm{\delta}}$ to \rf{NLIE} on $\bm\de$ is analytic for $k=1,\dots,N$.

For any given solution $X_{\bm{\delta}}^{}$ of the integral equation \rf{NLIE} 
let us define the functions $\mathrm{Q}_{\bm{\de}}^{\pm}$ as
\begin{subequations}\label{QfromY}
\begin{align}
  \label{Q^+delta} \mathrm{Q}^{+}_{\bm{\delta}}(\lambda) &= \left( \frac{\hbar}{\Lambda} \right)^\frac{\ii N\lambda}{\hbar}\frac{e^{-N\pi \lambda/\hbar}v_\uparrow(\lambda)}{\prod_{k=1}^{N}\Gamma(1-\ii(\lambda-\delta_k)/\hbar)}\,,  \\ 
   \label{Q^-delta} \mathrm{Q}^{-}_{\bm{\delta}}(\lambda) &= \left( \frac{\Lambda}{\hbar} \right)^\frac{\ii N\lambda}{\hbar}\frac{e^{-N\pi \lambda/\hbar}v_\downarrow(\lambda-\ii\hbar)}{\prod_{k=1}^{N}\Gamma(1+\ii(\lambda-\delta_k)/\hbar)}\,,
\end{align}
\end{subequations}
where 
\begin{subequations}\label{eq:v-fns}
\begin{align}
    \label{v+}\log v_\uparrow(\lambda) &= -\int_\BR\frac{d \mu}{2 \pi\ii} \frac{1}{\lambda - \mu+\ii\hbar/2}\log(1+X_{\boldsymbol\delta}(\mu))\,, \\ 
       \label{v-} \log v_\downarrow(\lambda-\ii\hbar) &= \int_\BR\frac{d \mu}{2 \pi\ii} \frac{1}{\lambda - \mu-\ii\hbar/2}\log(1+X_{\boldsymbol\delta}(\mu))\,.
\end{align}
\end{subequations}

The functions $\mathrm{Q}_{\bm{\de}}^{\pm}$ represent solutions to 
the Baxter equation \rf{Baxter-equation}, more precisely:
\begin{propn} {\rm~\cite{Kozlowski:2010tv}} The functions $\mathrm{Q}_{\bm{\de}}^{\pm}$ are entire, and there 
exists a polynomial $t_{\bm{\delta}}$ of degree $N$ such that the functions $\mathrm{Q}_{\bm{\de}}^{+}$
and $\mathrm{Q}_{\bm{\de}}^{-}$ both solve the Baxter equation \rf{Baxter-equation} with the same polynomial 
$t=t_{\bm{\delta}}$ on the left side.
\end{propn}

The function $t_{\bm{\delta}}$ can be represented in terms of $\mathrm{Q}_{\bm{\de}}^{\pm}$ by an explicit formula 
\cite[Equation (24)]{Kozlowski:2010tv}. To each tuple $\bm{\delta}=(\de_1,\dots,\de_N)$ for which there exists a solution 
$X_{\bm{\delta}}$ 
of \rf{NLIE}, we may thereby associate a polynomial 
$t_{\bm{\delta}}(\lambda)=\prod_{k=1}^{N}(\mu-\tau_k(\bm{\delta}))$, defining 
$\bm{\tau}(\bm{\delta})=(\tau_1(\bm{\delta}),\dots,\tau_N(\bm{\delta}))$. 

As observed in~\cite[Appendix B]{Kozlowski:2010tv}, 
one may construct solutions $X_{\bm{\delta}(\bm{\tau})}$ to the integral equation \rf{NLIE}
from the determinants $K_{\pm}$ considered in the previous subsection. The corresponding 
solutions to the Baxter equations are related simply as
${Q}_{\bm{\tau}}^{\pm}(\la)=\mathrm{Q}_{\bm{\de}(\bm{\tau})}^{\pm}(\la)$.
This is sufficient to prove that, for each solution to the quantization conditions, a solution $X_{\bm{\delta}}$ to~\rf{NLIE} exists.

\subsection{Quantization Conditions and Yang-Yang Function}\label{sec:Yang-Yang}

The quantization conditions (i), (ii) may now be reformulated in terms of the functions
$\mathrm{Q}_{\bm{\de}}^{\pm}(\la)$. We may, as before, consider an ansatz of the form 
\begin{equation} \label{q-ansatz-2}
    q(\lambda) = \frac{\mathrm{Q}^{+}_{\bm{\de}}(\lambda)-\zeta\mathrm{Q}^{-}_{\bm{\de}}(\lambda)}{\prod_{j=1}^{N}e^{-\frac{\pi\lambda}{\hbar}}\sinh\tfrac{\pi}{\hbar}(\lambda-\delta_j)}\,, \qquad\de_j\in\BC\,,\quad j=1,\dots,N\,.
\end{equation}
However, we may now simply assume that the parameters $\de_1,\dots,\de_N$ in \rf{q-ansatz-2} 
coincide with the parameters appearing in \rf{Theta-def}. The 
quantum Wronskian $W[\mathrm{Q}^{+}_{\bm{\de}},\mathrm{Q}^{-}_{\bm{\de}}](\lambda)$ vanishes for $\la=\de_j$
\cite[Equation (22)]{Kozlowski:2010tv}, 
implying that the infinite sequences 
$\mathfrak{q}_{\bm{\de},j}^+$ and $\mathfrak{q}_{\bm{\de},j}^-$, defined as
\[
\mathfrak{q}_{\bm{\de},j}^\pm=\big(\mathrm{Q}_{\bm{\de}}^{\pm}(\de_j-\mathrm{i}\hbar n)\big)_{n\in\BZ}\in\BC^{\BZ}\,,
\qquad j=1,\dots,N\,,
\]
are proportional to each other,
\begin{equation}\label{zetaj-def}
\mathfrak{q}_{\bm{\de},j}^+=\zeta_j^{}(\bm{\de},\Lambda)\,\mathfrak{q}_{\bm{\de},j}^-\,, \qquad
j=1,\dots,N\,.
\end{equation}
The quantization conditions can now be reformulated in terms of the functions 
$\zeta_j(\bm{\de},\Lambda)$ defined in \rf{zetaj-def}:
Conditions (i) and (ii) can be satisfied if and only if
\begin{equation}\label{q-cond-delta}
\zeta_j^{}(\bm{\de},\Lambda)=\zeta_{k}^{}(\bm{\de},\Lambda)\,,\qquad j\neq k\,, \qquad j,k\in\{1,\dots,N\}\,,
\end{equation}
allowing us to choose $\zeta=\zeta_j(\bm{\de},\Lambda)$ to cancel all poles from zeros 
of the denominator in \rf{q-ansatz-2}.
These conditions restrict the choices of 
$\de_1,\dots,\de_N$ to a discrete set once we further require that $P=\sum_{j=1}^N\tau_j=0$. The latter requirement amounts to imposing  $\sum_{j=1}^N\delta_j=0$, as it can be inferred from eq.~(\ref{delta-to-tau}) \cite[Equation (3.20)]{Kozlowski:2010tv}. Using the explicit expression for $\log\zeta_j$ (\ref{eta-explicit}) \cite[Equation (3.13)]{Kozlowski:2010tv}, together with eq.~(\ref{d_theta=P}), one can then infer that when the quantization conditions (\ref{q-cond-delta}) are imposed one has $\prod_{j=1}^N\zeta_j=\zeta^N=1$.

A result from~\cite{Kozlowski:2010tv} that will be crucial for us is the existence of a potential
$\CY(\bm{\de},\Lambda)$
for the functions $\log\zeta_j(\bm{\de},\Lambda)$, satisfying 
\begin{equation}\label{YY-eta-potential}
\log\zeta_j(\bm{\de},\Lambda)=
\frac{\partial}{\partial \de_j}\CY(\bm{\de},\Lambda)\,.
\end{equation}
The function $\CY(\bm{\de},\Lambda)$ can be represented explicitly in terms of the solution $X_{\bm\delta}(\mu)$
of the integral equation \rf{NLIE} as
\begin{subequations} \label{eq:Yang-Yang-def}
\begin{align}
   \label{Wdef} \mathcal{Y}(\boldsymbol{\delta}, \Lambda) &=  \Ypert(\boldsymbol{\delta},\Lambda)+\Yinst(\boldsymbol{\delta}, \Lambda)\,, \\
  \label{Wpert}  \Ypert(\boldsymbol{\delta}, \Lambda) &= N\frac{\ii}{\hbar}(\log\hbar-\log\Lambda)\sum_{j=1}^N\delta_j^2 
  +\sum_{j,k=1}^N\varpi(\delta_k-\delta_j)\,, \\
    \label{Yinst}\Yinst(\boldsymbol{\delta},\Lambda) &= -\int_\BR\frac{d\mu}{2\pi\ii}\left\{\frac{\log (X_{\bm\delta}(\mu)\Theta(\mu))}{2}\log(1+X_{\bm\delta}(\mu))+\Li_2(-X_{\bm\delta}(\mu))\right\}.
\end{align}
\end{subequations}
whe$\varpi(\lambda)$ is such that $\varpi'(\lambda) = \log \Gamma(1+\ii\lambda/\hbar)$. The details
of the proof of \rf{YY-eta-potential}, 
not included in \cite{Kozlowski:2010tv}, can be found in Appendix \ref{W=Y appendix}.

\begin{rem} The function $\CY(\bm{\de},\Lambda)$ has first been introduced in~\cite{Nekrasov:2009rc}
in the context of the instanton counting in 
pure ${\cal N}=2$, $d=4$ supersymmetric Yang-Mills theory with gauge group $\operatorname{SU}(N)$. It has been  
argued in~\cite{Nekrasov:2009rc} and
verified in more detail in~\cite{Meneghelli:2013tia} that the function $\CY(\bm{\de},\Lambda)$
represents the limit $\epsilon_2$ of the instanton partition functions in the Omega-background~\cite{Nekrasov:2002qd}
with parameters $\ep_1,\ep_2$. The relations between quantum integrable models 
and ${\cal N}=2$, $d=4$ supersymmetric quantum field theories discussed in~\cite{Nekrasov:2009rc} 
suggest the conjecture (later proven in~\cite{Kozlowski:2010tv})
that the quantization conditions of the closed quantum Toda chain can be 
formulated in the form above.
\end{rem}

\section{Higher Rank Mathieu Opers}\label{oper-section}

In this section we review the necessary background on opers on the twice-punctured sphere $C_{0,2} = \BP^1\setminus\{0,\infty\}$, such as their definition and parameterization both in terms of connection coefficients as well as (generalized) monodromy data.
This precedes a discussion of the relevant Riemann-Hilbert problem, and the definition of the generating function of opers, which will be central in the comparison with the quantum Toda chain.

\subsection{Opers From Solutions to the Baxter Equation}

We may easily observe that the Fourier-transform of the Baxter equation~(\ref{Baxter-equation}) yields 
an ordinary differential equation of the form
\begin{equation} \label{oper-x}
       t(-\ii\hbar \partial_x) \psi(x) = \Lambda^N\left( \ii^{N}e^x +  \ii^{-N}e^{-x} \right)\psi(x)\,.
\end{equation}
In order to exhibit the analytic properties of $\psi(x)$, it will be convenient to introduce the variable 
$z=e^x$, and the function $\chi(z)=\psi(\log(z))$, satisfying 
\begin{equation} \label{oper-z}
       t(-\ii\hbar z\partial_z) \chi(z) = \Lambda^N\left( \ii^{N}z +  \ii^{-N}z^{-1} \right)\chi(z)\,.
\end{equation}
The ordinary differential equation of $N$-th order  can be recast 
as a first order matrix equation $(-\ii\hbar z \partial_z\psi(z) - A(z))\Psi(z) = 0$
of the form
\begin{equation}\label{oper-connection}
   {A}(z)= \begin{pmatrix}
        0&E_2& E_3 & \cdots & E_{N-1} & a_N(z; \Lambda) \\
        1&0 & 0 &\cdots & 0 & 0 \\
        0&1& 0 &\cdots & 0 & 0 \\
        \vdots & \vdots & \vdots && \vdots & \vdots \\
        0 & 0 & 0 & \cdots&1&0
    \end{pmatrix},\qquad 
    \Psi(z) = \begin{pmatrix}
        \chi_1^{(N-1)} & \cdots & \chi_N^{(N-1)} \\
        \vdots && \vdots \\
          \chi^{(1)}_1 & \cdots & \chi^{(1)}_N \\
        \chi_1^{}& \cdots &\chi_N^{}  \end{pmatrix},
\end{equation}
where $a_N(z; \Lambda) \coloneqq E_N + \Lambda^N \left( \ii^{N}z +  \ii^{-N}z^{-1} \right)$.
We may observe that $\nabla=dz\,(\pa_z-\frac{\ii}{\hbar z}A(z))$ is a connection  that 
is holomorphic on $C_{0,n}=\mathbb{P}^1\setminus\{0,\infty\}$ and has irregular 
singularities at $0$ and $\infty$. However, $\nabla$ is not the most general connection of this kind. 
In fact, one may first note that the singular behavior at $0$ and $\infty$ is of the mildest possible type. 
As discussed in Section \ref{Section:canonical basis}, this can be inferred from the Newton polygon of the associated $N$-th order differential equation. 
Secondly, the connection $\nabla$ is easily seen to be gauge equivalent to a connection 
of the form $\nabla'=dz\,(\pa_z-B(z))$, with
\begin{equation}\label{oper-form}
   B(z)= \begin{pmatrix}
        0&u_2(z)& u_3(z) & \cdots & u_{N-1}(z) & u_N(z) \\
        1&0 & 0 &\cdots & 0 & 0 \\
        0&1& 0 &\cdots & 0 & 0 \\
        \vdots & \vdots & \vdots && \vdots & \vdots \\
        0 & 0 & 0 & \cdots&1&0
    \end{pmatrix},
\end{equation}
where 
\begin{equation}
    u_j(z) = \mathsf{u}_jz^{-j}\,,\qquad j=2,\dots,N-1\,,\qquad u_N(z) = \frac{\Lambda^N}{z^{N+1}} + \frac{\mathsf{u}_N}{z^N} + \frac{\Lambda^N}{z^{N-1}}\,.
\end{equation}
Connections of the form \rf{oper-form} are often called opers in the literature. The condition to be locally 
gauge-equivalent to a connection of the form \rf{oper-form} is known, in general, to define half-dimensional subspaces of the moduli spaces of all holomorphic $\SLNC$-connections on punctured Riemann surfaces $C$ having fixed singularity types at the punctures. In our case one may observe that a more general family of connections having  irregular singularities of the mildest possible type at $0$ and $\infty$ 
can be defined as follows,\footnote{\label{Lax-mat-footnote}In order to see that \rf{gen-connection} has 
a singularity at $0$ of the same type as \rf{oper-connection}, 
one may note that up to terms of higher order in $z$, it can be brought to the form 
\rf{oper-connection} by a diagonal gauge transformation. Then one can study the associated scalar, $N$-th order oper equation (see \cite[Theorem XIII-5-4]{hsieh_sibuya_1999}) by drawing the Newton polygon at $z=0$ and finding that it has a unique positive slope of $1/N$. A similar argument can be used for the 
singularity at $z=\infty$.}
\begin{equation}\label{gen-connection}
   {A}(z)= \begin{pmatrix}
        p_1&\Lambda a_1 & 0 & \cdots & 0 & z^{-1}\Lambda a_N \\
        \Lambda a_1 & p_2 & \Lambda a_2 &\cdots & 0 & 0 \\
        0&\Lambda a_2 & p_3 &\cdots & 0 & 0 \\
        \vdots & \vdots & \vdots && \vdots & \vdots \\
                0 & 0  & 0 && p_{N-1} & \Lambda a_{N-1} \\
        z\Lambda a_N & 0 & 0 & \cdots&\Lambda a_{N-1}&p_N
    \end{pmatrix}.    
\end{equation}
Connections of the form \rf{gen-connection} are related to connections of the form \rf{oper-form}
by mildly singular gauge transformations if one allows for certain singularities called apparent singularities 
in connection matrix elements $u_j(z)$ (see e.g. \cite[Lemma XIII-5-1]{hsieh_sibuya_1999}). The condition that the resulting functions
$u_j(z)$ are holomorphic defines a half-dimensional subspace of the space of 
connections of the form \rf{gen-connection}.
    
\subsection{Monodromy and Stokes Data} \label{Monodromy-Stokes}

It is well-known that flat connections on Riemann surfaces are to a large extent determined by their
monodromy data. In the cases of closed or punctured Riemann surfaces $C$ one may associate to a holomorphic 
connection $\mathbf{A}$ the collection of its holonomies
\begin{equation}
    M_{\gamma}= \textup{Hol}_{\gamma}(\mathbf{A}) =  \exp \oint_\gamma \mathbf{A}(z)\,dz\,,
\end{equation}
defining a representation $\rho:\pi_1(C)\rightarrow \operatorname{GL}(N,\mathbb{C})$ of the fundamental group of $C$. 
The correspondence between representations $\rho$ of $\pi_1(C)$ and gauge equivalence classes of holomorphic connection $\mathbf{A}$ is known as the Riemann-Hilbert correspondence. 

If the connections $\mathbf{A}$ have irregular singularities, one needs to supplement 
the representation $\rho$ by Stokes data describing the relations between canonical bases
defined in sectors around the irregular singularities. 
If $C$ has more than one irregular singularity, 
it may also be useful to consider the matrices relating canonical bases at two different irregular 
singularities. 

In the present case we are interested in the twice-punctured Riemann sphere $C_{0,2}=\mathbb{P}^1\setminus\{0,\infty\}$  
having irregular singularities of the mildest possible type at both $0$ and $\infty$.
We are going to demonstrate that
there exist canonical bases $\mathbf{Y}^{(0,\infty)}_{k}(z)=({Y}^{(0,\infty)}_{k,0},\dots,{Y}^{(0,\infty)}_{k,N-1})$, with $k=0,\dots,2N-1$, characterized by having asymptotic behavior near $0$ and $\infty$ of the following form.  
Introducing the coordinates $y=\tfrac{N\Lambda}{\hbar}z^{1/N}$ and $y'=\tfrac{N\Lambda}{\hbar}z^{-1/N}$ on 
$N$-fold covers of $C_{0,2}$, they are defined by holomorphic functions $Y^{(\infty)}_{k,n}(y)$
and $Y^{(0)}_{k,n}(y')$
admitting uniform asymptotic expansions of the form
\begin{subequations} \label{Y-asymptotics}
\begin{align}
    Y^{(\infty)}_{k,n}(y) &\sim \ii^{N+1}e^{-\alpha_n y}(\alpha_n y)^{-\frac{N-1}{2}}\tilde{f}^{(\infty)}_n\,, &y\to\infty\,, \\
    Y^{(0)}_{k,n}(y') &\sim (-\ii)^{N+1}e^{-\alpha_{-n}y'}(\alpha_{-n}y')^{-\frac{N-1}{2}}\tilde{f}^{(0)}_n\,, &y'\to\infty\,,
\end{align}
\end{subequations}
where $\alpha_n\coloneqq \exp(2\pi\ii n/N)$ and $\tilde{f}^{(0,\infty)}_n$ are formal power series in $y^{-1}$ and ${y'}^{-1}$ with leading coefficients normalized to $1$.\footnote{We have included the factors $(\pm\ii)^{N+1}\alpha_{\pm n}^{-\frac{N-1}{2}}$ in order to simplify the results of Sections~\ref{Section:canonical basis} and~\ref{sec:connection-langlands}.} 
The functions $Y^{(\infty)}_{k,n}(y)$ are holomorphic and admit the uniform asymptotic expansions above 
in closed subsectors\footnote{Sectors defined in a similar way as~(\ref{Sokal-discs}), but with a closed angular range $\theta \in \big[\frac{\pi}{2}+\frac{(k-1)\pi}{N}+\varepsilon, \frac{\pi}{2}+\frac{k\pi}{N}-\varepsilon \big]$ with $\varepsilon>0$ small enough, and similarly for $\theta'$.} of the domains
\begin{subequations} \label{Sokal-discs}
\begin{align} 
    D^{(\infty)}_k &\coloneqq \bigcup_{ \theta \in \left(\frac{\pi}{2}+\frac{(k-1)\pi}{N}, \frac{\pi}{2}+\frac{k\pi}{N} \right)}\left\{y \in \mathbb{C} \,\middle|\;\Re(ye^{\ii\theta}) > R\,\right\}, \\
    D^{(0)}_k &\coloneqq \bigcup_{\theta' \in \left(-\frac{\pi}{2}-\frac{k\pi}{N}, -\frac{\pi}{2}-\frac{(k-1)\pi}{N} \right)}\left\{y' \in \mathbb{C} \,\middle|\;\Re(y'e^{\ii\theta'}) > R' \,\right\},
\end{align}
\end{subequations}
respectively.\footnote{The formulas~(\ref{Sokal-discs}) hold for $N\geq 3$. For $N=2$, the opening angle of the sectors in which uniform asymptotic expansions hold is $2\pi$, rather than $3\pi/2$; moreover the index $k$ only takes values $0,1$. This will be justified in Section~\ref{Section:canonical basis} by analysing the structure of the Borel plane for $f^{(0,\infty)}_0$.} \label{footnote:N=2}
The Stokes matrices $\mathcal{S}_k$ describe the relations between the canonical bases associated to different sectors 
$D^{(0,\infty)}_k$,
\begin{equation}\label{Stokes-matrices-def-1}
\mathbf{Y}^{(0,\infty)}_{k+1}(z)=\mathbf{Y}^{(0,\infty)}_k(z)\mathcal{S}_{k}\,.
\end{equation}

We are furthermore going to show in Section~\ref{Section:canonical basis}
that the Stokes matrices are completely determined by the eigenvalues of the holonomy along the simple closed 
curve $\ga$ separating $0$ and $\infty$ on $C_{0,2}$ in this case. This is an important 
simplifying feature that does not hold in general
for other types of irregular singularities.\footnote{For a discussion of Stokes phenomena arising from flat sections of connections with higher Poincaré rank, see e.g. \cite{Witten:2007td} and references therein.}

As both monodromy and Stokes data are completely determined by the eigenvalues of $M_{\gamma}$, we may formulate a natural variant of the Riemann-Hilbert problem in our case. 
It aims to find the connections $\nabla$ of oper form~\rf{oper-form} having specified eigenvalues of $M_{\gamma}$.
A similar formulation of the relevant problem (RH) of Riemann-Hilbert type 
has been given in the introduction.
In order to construct solutions of this problem,  it will be useful to 
consider Floquet-bases characterized by having diagonal monodromy 
\begin{equation}
    M_{\ga} = \diag(\Sigma_1,\dots,\Sigma_N) =  \mathrm{diag}\big( e^{2\pi\ii\sigma_1},\dots,e^{2\pi\ii\sigma_N} \big)\in\SLNC\,.
\end{equation}
This condition clearly leaves considerable freedom in the definition of the Floquet-bases, 
represented by multiplication of individual basis elements with functions of $\bm{\si}$ and $\Lambda$.
We shall partially fix this ambiguity by 
considering two types of Floquet-bases
$\mathbf{F}^{(0,\infty)}(z)=({F}^{(0,\infty)}_{1},\dots,{F}^{(0,\infty)}_{N})$, 
admitting  expansions of the form $\mathbf{F}^{(0,\infty)}(z)=\bm{\vf}^{(0,\infty)}(w_{0,\infty})+O(\Lambda)$, where
\begin{equation}\label{Floquet-gen}
\begin{aligned}
w_0 &= \left( \frac{\Lambda}{\hbar} \right)^Nz\,,\qquad\quad\;\, \vf^{(0)}_j(w_0)=
N_j^{(0)}(\bm{\si})\sum_{n=0}^{\infty}
\vf_{n,j}^{(0)}(\bm{\si})w_0^{n+\si_j}\,,\\
w_\infty &= \left( \frac{\Lambda}{\hbar} \right)^Nz^{-1}\,,\qquad \vf^{(\infty)}_j(w_\infty)=
N_j^{(\infty)}(\bm{\si})\sum_{n=0}^{\infty}
\vf_{n,j}^{(\infty)}(\bm{\si})w_\infty^{n-\si_j}\,,
\end{aligned} \qquad j=1,\dots,N\,,
\end{equation}
and requiring that the bases $\mathbf{F}^{(0,\infty)}$
are related to the canonical bases by relations of the form
\begin{equation} \label{eq:F-to-Y-def}
\mathbf{Y}^{(0,\infty)}_0(z)=\mathbf{F}^{(0,\infty)}(z)\,C^{(0,\infty)}(\mathbf{\bm\sigma})\,,
\end{equation}
with $C^{(0,\infty)}(\mathbf{\sigma})$ being $N\times N$-matrices that are $\Lambda$-independent.
A prescription for fixing the remaining freedom, represented by the functions
$N_j^{(0,\infty)}(\bm{\si})$ in \rf{Floquet-gen}, will be given below. 

\subsection{Definition of the Generating Function of Opers}

The relation between the Floquet bases associated to the singularities at $0$ and $\infty$, respectively,
must be of the form
\begin{equation}\label{etaj-def}
    {F}^{(0)}_j(z)=e^{2\pi\ii\eta_j }{F}^{(\infty)}_j(z)\,,\qquad j=1,\dots,N\,.
\end{equation}
We may, more canonically, define the parameters $\eta_j$ in terms of the eigenvalues of the 
connection matrix $E$ describing the relation between the bases as
\begin{equation}\label{transpzeroinfty}
\mathbf{Y}^{(0)}_{0}(z)=\mathbf{Y}^{(\infty)}_0(z)E\,.
\end{equation}
Considering generic holomorphic connections $\mathbf{A}$ on $C_{0,2}$ 
with irregular singularities of mildest possible type
at $0$ and $\infty$, one will find that the parameters $\bm{\sigma}=(\si_1,\dots,\si_N)$
and $\bm{\eta}=(\eta_1,\dots,\eta_N)$ are independent in general, apart from the 
relations $\sum_{j=1}^N\si_j=0$ and $\sum_{j=1}^N\eta_j=0$. Taken together, it follows that
$(\bm{\sigma},\bm{\eta})$ can serve as coordinates for the space of monodromy and Stokes 
data of generic connections on $C_{0,2}$ having irregular singularities of the type 
specified above.\footnote{Similar parameterizations for the monodromies of certain families of $\SLNC$-connections on $C_{0,4}$ were used in \cite{Hollands:2017ahy,Gavrylenko:2018ckn,Gavrylenko:2018,Jeong:2018qpc}. We expect that a confluence limit should relate the coordinates defined in these references to the coordinates used in our paper.} In the case $N=2$ it is known that 
$\sigma\equiv\sigma_1=-\sigma_2$ and $\eta\equiv\eta_1=-\eta_2$ are Darboux coordinates for 
$\mathcal{M}_{\textrm{flat}}$~\cite{Its:2014lga}. We conjecture that this holds also for $N>2$.\footnote{A proof should be possible using abelianization \cite{Hollands:2017ahy,Hollands:2021itj}, and by the techniques described in \cite{Jeong:2018qpc}. Cluster coordinates have been studied for the case of our interest in \cite{williams2016toda}.}

Connections $\nabla$ of oper form have $N-1$ free parameters $E_2,\dots,E_N$ apart from $\Lambda$. It follows that
the parameters $\bm{\eta}$ introduced through~(\ref{transpzeroinfty}) must be functions of 
$\bm{\si}$ and $\Lambda$ in this case.

We will see that the relation between $\bm{\eta}$ and $\bm{\sigma}$ characterizing the subspace of oper connections
within the moduli space of generic connections on $C_{0,2}$ 
with irregular singularities of mildest possible type can be conveniently encoded in a generating 
function $\mathcal{S}(\boldsymbol{\sigma}, \Lambda)$, satisfying
\begin{equation} \label{dY-dsigma}
\frac{\partial}{\partial \sigma_j}\mathcal{S}(\boldsymbol{\sigma}, \Lambda) = 2\pi \hbar\,\eta_j(\boldsymbol{\sigma}, \Lambda)\,.
\end{equation}
While the existence of a function $\mathcal{S}(\boldsymbol{\sigma}, \Lambda)$ serving as a potential 
for the functions $\eta_j(\boldsymbol{\sigma}, \Lambda)$ may not yet be obvious at this stage,  
we will show below that such a function $\mathcal{S}$  indeed exists, relating it to the 
Yang-Yang-function $\CY$ of the closed Toda chain.

\section{Canonical Bases of Solutions to the Oper Equation} \label{Section:canonical basis}

In this section we provide a construction of the canonical bases $\mathbf{Y}^{(0,\infty)}_k$ introduced in Section~\ref{Monodromy-Stokes} via Borel-Laplace resummation (see~\cite{Sauzin:2014qzt} for an introduction). 
Specifically we will prove:
\begin{propn}\label{prop-canonical-bases}
  There exist families of canonical bases of solutions $\mathbf{Y}^{(0,\infty)}_k$ to the oper equation~\eqref{oper-z} , with $k=0,\dots,2N-1$\footnote{$k=0,1$ when $N=2$.}  that admit uniform asymptotic expansions~(\ref{Y-asymptotics}) on closed subsectors of the domains~(\ref{Sokal-discs}).
\end{propn}
 In Section~\ref{section:Stokes-matrices} we then compute the Stokes matrices $\mathcal{S}_k$ defined by~\eqref{Stokes-matrices-def-1}; in Section~\ref{Monodromy-matrix-section} we use Stokes matrices $\mathcal{S}_k, \mathcal{S}_{k+1}$ to construct the monodromy matrix $M_k$ for the canonical basis $\mathbf{Y}^{(\infty)}_k$. Thereby, we relate the Stokes constants, namely the non-trivial entries of the Stokes matrices, to the eigenvalues $\boldsymbol{\Sigma}$ of the monodromy matrix. Specifically, we prove the following
\begin{propn}\label{propn:Stokes-constants} \hfill
    \begin{enumerate}
        \item The Stokes matrices $\mathcal{S}_k$ defined by~(\ref{Stokes-matrices-def-1}) are entirely specified by $N-1$ Stokes constants $s_i$, $i=1,\dots,N-1$. 
        \item The Stokes constants $s_i$ are related to the eigenvalues $\boldsymbol{\Sigma}=(\Sigma_1,\dots,\Sigma_N)$ of the monodromy matrix by
    \end{enumerate}
    \begin{equation}
        s_i(\boldsymbol{\Sigma})=(-1)^{i+1}e_i(\boldsymbol{\Sigma})\,,
    \end{equation}
    where $e_i(\boldsymbol{\Sigma})$ is the elementary symmetric polynomial of degree $i$ in $\Sigma_1,\dots,\Sigma_N$.
\end{propn}
We will focus initially on the bases $\mathbf{Y}^{(\infty)}_k$.
The treatment of $\mathbf{Y}^{(0)}_k$ is largely analogous, and is discussed in Section~\ref{sec:Y^0}.
In the next two sub-sections we establish the preparatory results that will be needed to prove Propositions~\ref{prop-canonical-bases} and~\ref{propn:Stokes-constants}.

\subsection{Formal Solutions of the Oper Equation}
The exponential behavior of the solutions of the oper equation~(\ref{oper-z})  at $z \rightarrow \infty$ is encoded in the Newton polygon.  This can be easily drawn by rescaling the oper equation~(\ref{oper-z}) so that, once written in terms of differentials $z \partial_z$, it only involves polynomials in $z^{-1}$ holomorphic at infinity. Then one simply plots, for every $k=0,\dots,N$, the degree of the smallest non-vanishing monomial in the polynomial multiplying $(z\partial_z)^k$.  Similarly one can find the Newton polygon at $0$; the result turns out to be the same at both $0$ and $\infty$ and is shown in Figure~\ref{Newton-poly} (see e.g. ~\cite{hsieh_sibuya_1999} for more details).
\begin{figure}[ht!]
    \centering
\begin{tikzpicture}[scale=1.5]
\draw[->] (-0.5,0) -- (5,0);
\foreach \x in {1,2,3,4} {
    \draw (\x,0.1) -- (\x,-0.1) node[below] {\x};
}

\draw[->] (0,-0.5) -- (0,2);
\foreach \y in {1} {
    \draw (0.1,\y) -- (-0.1,\y) node[left] {\y};
}

\draw[RoyalBlue, thick] (0,2) -- (0,0) -- (4,1) -- (4,2);

\foreach \x/\y in {0/0, 1/1, 2/1, 4/1} {
    \fill (\x,\y) circle (2pt);
}

\end{tikzpicture}
\caption{Newton polygon for $N=4$. The point corresponding to $3=N-1$ is missing since there is no $(z\partial_z)^{N-1}$ term in the oper equation (\ref{oper-z}). There is only one positive slope $1/N$, which is the Poincar\'e rank of the singularity \cite[Theorem XIII-7-6(a)]{hsieh_sibuya_1999}. The Newton polygon is the same for both punctures at $z=0, \infty$. }
\label{Newton-poly}
\end{figure}
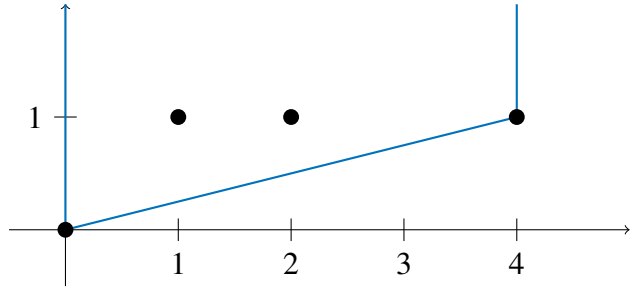
In particular, it can be observed that the only positive slope is $1/N$: this will then prescribe the exponential behavior of the solutions of the differential equations at both $0$ and $\infty$. Referring to~\cite{hsieh_sibuya_1999} for proofs, we conclude that we can write an ansatz for formal solutions to the oper equation as
\begin{equation}
   \tilde{ Y}^{(\infty)}_n(z) \,\propto\, e^{-A_n z^{1/N}}z^{B_n}\tilde{f}_n^{(\infty)}\,,
\end{equation}
where $\tilde{f}_n^{(\infty)}$ is a priori a polynomial in $\log z$ with coefficients being formal power series. 

In this section it will be convenient to work with the variable  $y = Nw^{1/N}= N\tfrac{\Lambda}{\hbar}z^{1/N}$. A simple but important observation is that, in this variable, the monodromy $z\rightarrow ze^{2\pi\ii}$ is represented by $y\to ye^{\frac{2\pi\ii}{N}}$. In this variable, the oper equation~\eqref{oper-z} reads
\begin{equation}\label{oper-y}
            \left[ \partial_y^N + \mathcal{E}_2\partial_y^{N-2}+\dots+\mathcal{E}_N - (-1)^Ny^N -\nu y^{-N} \right]\tilde{Y}^{(\infty)}(y)=0\,,
\end{equation}
where $\nu = (N\Lambda/\hbar)^{2N}$ and $\mathcal{E}_j = E_j (\ii N/\hbar)^j$. It is then easy to argue by simply plugging in the formal solution $\tilde{Y}^{(\infty)}_n(y) = e^{-\alpha y} y^{\beta}\tilde{f}_n^{(\infty)}$  and cancelling the two highest powers of $y$ that one has
\begin{equation}\label{alpha-beta}
    \alpha^N =1\,, \qquad \beta = -\frac{N-1}{2}\,,
\end{equation}
with the convention that $\tilde{f}_n^{(\infty)}$ has a non-vanishing constant term. Thus, the $N$ elements of the formal basis can be associated with the $N$ roots of unity $\alpha_n \coloneqq \exp(2\pi\ii n/N)$. Namely, we define a formal basis of solutions $\mathbf{\tilde{Y}}^{(\infty)}$ by
\begin{equation}\label{formal-basis}
 \tilde{Y}^{(\infty)}_n(y) = \ii^{N+1}e^{-\alpha_n y}(\alpha_ny)^{\beta}\tilde{f}_n^{(\infty)}(y) \,,
\end{equation}
where we indicate $\tilde{f}^{(\infty)}_n(y) = \tilde{f}^{(\infty)}_0(\alpha_ny)$ and $n=0,\dots,N-1$.  We will prove the following
\begin{lem}\label{lemma-gevrey}
For every $n=0,\dots, N-1$, $\tilde{f}_n^{(\infty)}$ is a formal power series in powers of $y^{-1}$ of Gevrey class 1. Namely, if we let $\tilde{f}_n^{(\infty)} = \sum_{k=0}^{\infty}a_ky^{-k}$, there exist constants $A, B$  such that $|a_k|\leq A B^k k! $ for every $k$.
\end{lem}
The proof is in Appendix~\ref{canonical-basis-appendix}.

\subsection{The Structure of the Borel Plane}\label{Borel-plane-section}
For convenience, let us redefine $y^{-1}\tilde{f}_0 ^{(\infty)} \eqqcolon \tp_0$,\footnote{In order to reduce clutter we omit the superscript $(\infty)$ for quantities defined in this section.} and similarly $(y\alpha_n)^{-1}\tilde{f}^{(\infty)}_n \eqqcolon \tp_n$. The Borel transform $\mathcal{B}[\tp_n]$ of $\tp_n$ is 
\begin{equation}\label{def-Borel-transform}
    \mathcal{B}[\tp_n](\zeta) \coloneqq \alpha_{-n}\sum_{k=0}^{\infty}\frac{a_k}{k!}(\alpha_{-n}\zeta)^k\,,\qquad \tp_n(y) = \sum_{k=0}^{\infty}a_k(\alpha_{n}y)^{-k-1} \,.
\end{equation}
A straightforward corollary of Lemma~\ref{lemma-gevrey} is
\begin{cor}
    The Borel transform $\mathcal{B}[\tp_n]$ defines an analytic function $\hp_n(\zeta)$ in a neighbourhood of the origin $\zeta=0$.
\end{cor}
Next, we find the following

\begin{lem}\label{lemma-borel-plane}
The Borel transform $\hp_0(\zeta) = \mathcal{B}[\tp_0]$ has singularities at $\zeta = \omega_{0,j}$, where $(\omega_{0,j} +1)^N=1$ and $\omega_{0,j} \neq 0$. Moreover, let $\hp_n(\zeta) = \mathcal{B}[\tp_n]$. Then one has $\hp_n(\zeta) = \alpha_{-n}\hp_0(\alpha_{-n}\zeta)$.
\end{lem}
The proof is in Appendix~\ref{canonical-basis-appendix}.

The singularities of $\hat\varphi_0(\zeta)$ are conveniently parameterized as $\omega_{0,j} = \alpha_j-1$, $j=1,\dots,N-1$.
Furthermore, as a consequence of Lemma~\ref{lemma-borel-plane}, it is easy to see that the poles of $\hp_n(\zeta)$ lie at $\zeta = \omega_{n,j} \coloneqq \alpha_n\omega_{0,j}$.

The angle of the ray intersecting $\omega_{n,j}$ is given by
\begin{equation}
    \arg\omega_{n,j} = \frac{\pi}{2} + \frac{\pi}{N}(2n+j) \eqqcolon \phi_{2n+j}\,,
\end{equation}
so that we can label the Stokes lines by $\phi_k$, $k=0,\dots,2N-1$.
The Borel plane is therefore divided into $2N$ equal sectors,\footnote{This is only true for $N \geq 3$, as anticipated in~\ref{footnote:N=2}. For $N=2$, there are only $2$ sectors in the Borel plane, as the only singularities are  at $\arg \omega_{0,1}=\pi$, $\arg \omega_{1,1}=0$.} as depicted in figure~\ref{fig:Borel plane2}.

\begin{figure}[ht!]
\centering
\begin{subfigure}[b]{0.48\textwidth}
\begin{tikzpicture}[scale=0.7]
    \def\n{5} 
    \def\r{2} 
    \def\shiftangle{1} 
    
    \draw[->] (-4,0) -- (4,0) ;
    \draw[->] (0,-4) -- (0,4)  ;

    \coordinate (S0) at ({0*360/\n}:\r);
    \coordinate (S1) at ({1*360/\n}:\r);
    \coordinate (S2) at ({2*360/\n}:\r);
    \coordinate (S3) at ({3*360/\n}:\r);
    \coordinate (S4) at ({4*360/\n}:\r);

    \coordinate (A0) at ({0*360/\n + \shiftangle}:\r);
    \coordinate (A1) at ({1*360/\n + \shiftangle}:\r);
    \coordinate (A2) at ({2*360/\n + \shiftangle}:\r);
    \coordinate (A3) at ({3*360/\n + \shiftangle}:\r);
    \coordinate (A4) at ({4*360/\n + \shiftangle}:\r);

    \foreach \i in {0,...,4} {
        \filldraw[black] (S\i) circle (1pt);
    }

    \draw[->, black] (S0) -- (A1);
     \draw[->, black] (S0) -- (A2);
     \draw[->, black] (S0) -- (A3);
    \draw[->, black] (S0) -- (A4);

    \draw[->, red] (S1) -- (A0);
    \draw[->, red] (S1) -- (A2);
    \draw[->, red] (S1) -- (A3);
    \draw[->, red] (S1) -- (A4);

    \draw[->, blue] (S2) -- (A0);
    \draw[->, blue] (S2) -- (A1);
    \draw[->, blue] (S2) -- (A3);
    \draw[->, blue] (S2) -- (A4);

    \draw[->, brown] (S3) -- (A0);
   \draw[->, brown] (S3) -- (A1);
   \draw[->, brown] (S3) -- (A2);
  \draw[->, brown] (S3) -- (A4);

    \draw[->, green] (S4) -- (A0);
    \draw[->, green] (S4) -- (A1);
    \draw[->, green] (S4) -- (A2);
    \draw[->, green] (S4) -- (A3);

    \node[above right, text=black] at (S0) { $1 = \alpha_0$};
    \node[above , text=red] at (S1) {$\alpha_{1}$};
    \node[left, text=blue] at (S2) {$\alpha_{2}$};
    \node[below, text=brown] at (S3) {$\alpha_{3}$};
    \node[right, text=green] at (S4) {$\alpha_{4}$};

\end{tikzpicture}
    \subcaption{Pictorial solution to the solutions to $\alpha_n + \omega_{n,j} = \alpha_{n+j}$ for $N=5$. These are all the possible oriented diagonals of the $N$-agon. The diagonals issuing from a vertex have the same color of the vertex label.}
    \label{fig:Pentagon}
\end{subfigure}%
\hfill
\begin{subfigure}[b]{0.48\textwidth}
         \centering
       \begin{tikzpicture}[scale=0.7]
    \def\n{5}  
    \def\theta{72} 

    \foreach \k in {0,1,2,3,4} {
        \pgfmathsetmacro{\angle}{360*\k/\n} 
        \pgfmathsetmacro{\x}{2*(cos(\angle) - 1)} 
        \pgfmathsetmacro{\y}{2*sin(\angle)}     

        \fill (\x,\y) circle (1pt); 

        \ifnum\k>0
            \draw[->, thick] (0,0) -- (\x,\y);
        \fi
    }

    \foreach \k in {1,2,3,4} {
        \pgfmathsetmacro{\angle}{360*\k/\n} 
        \pgfmathsetmacro{\x}{2*(cos(\angle) - 1)}
        \pgfmathsetmacro{\y}{2*sin(\angle)}

        \pgfmathsetmacro{\xr}{\x*cos(\theta+1) - \y*sin(\theta+1)}
        \pgfmathsetmacro{\yr}{\x*sin(\theta+1) + \y*cos(\theta+1)}
         \fill (\xr,\yr) circle (1pt); 
        \draw[->, thick , red] (0,0) -- (\xr,\yr); 
    } 
     \foreach \k in {1,2,3,4} {
        \pgfmathsetmacro{\angle}{360*\k/\n} 
        \pgfmathsetmacro{\x}{2*(cos(\angle) - 1)}
        \pgfmathsetmacro{\y}{2*sin(\angle)}

        \pgfmathsetmacro{\xr}{\x*cos(2*\theta) - \y*sin(2*\theta)}
        \pgfmathsetmacro{\yr}{\x*sin(2*\theta) + \y*cos(2*\theta)}
         \fill (\xr,\yr) circle (1pt); 
        \draw[->, thick , blue] (0,0) -- (\xr,\yr); 
    }
     \foreach \k in {1,2,3,4} {
        \pgfmathsetmacro{\angle}{360*\k/\n} 
        \pgfmathsetmacro{\x}{2*(cos(\angle) - 1)}
        \pgfmathsetmacro{\y}{2*sin(\angle)}

        \pgfmathsetmacro{\xr}{\x*cos(4*\theta+1) - \y*sin(4*\theta+1)}
        \pgfmathsetmacro{\yr}{\x*sin(4*\theta+1) + \y*cos(4*\theta+1)}
         \fill (\xr,\yr) circle (1pt); 
        \draw[->, thick , green] (0,0) -- (\xr,\yr); 
    }
     \foreach \k in {1,2,3,4} {
        \pgfmathsetmacro{\angle}{360*\k/\n} 
        \pgfmathsetmacro{\x}{2*(cos(\angle) - 1)}
        \pgfmathsetmacro{\y}{2*sin(\angle)}

        \pgfmathsetmacro{\xr}{\x*cos(3*\theta+1) - \y*sin(3*\theta+1)}
        \pgfmathsetmacro{\yr}{\x*sin(3*\theta+1) + \y*cos(3*\theta+1)}
         \fill (\xr,\yr) circle (1pt); 
        \draw[->, thick , brown] (0,0) -- (\xr,\yr); 
    }

    \draw[->] (-4,0) -- (4,0) node[right] {$\Re \zeta$};
    \draw[->] (0,-4) -- (0,4) node[above] {$\Im \zeta$};

    \pgfmathsetmacro{\X}{2*cos(80)}
\pgfmathsetmacro{\Y}{2*sin(80}
\draw[thick] (\X,\Y) arc (80:150:2);
\pgfmathsetmacro{\x}{2*cos(150)}
\pgfmathsetmacro{\y}{2*sin(150)}

    \draw[thick] (0,0) -- (\x,\y);
     \draw[thick] (0,0) -- (\X,\Y);

     \node[above] at (0.9,2) {$W^{(\infty)}_0$};
        \node[above] at (-3,1.5) {$W^{(\infty)}_2$};
          \node[above] at (-0.85,2) {$W^{(\infty)}_1$};
     
\end{tikzpicture}

         \caption{
        Full structure of the Borel plane for $N=5$. It is obtained by drawing all the oriented diagonals of the pentagon and bringing them to the origin. The arrows are slightly slanted merely to show the overlap of the colours. The black arc denotes an angle of $2\pi/5$, implementing the analytic continuation of the canonical basis $\mathbf{Y}^{(\infty)}_0$. Note that it crosses $2$ Stokes lines which contain in total $4= N-1$ singularities which cannot be obtained by rotations of  $2\pi/5$.
         }
         \label{fig:Borel plane2}
         \end{subfigure}
         \caption{}
     \end{figure}

\subsection{\texorpdfstring{Canonical Bases $\mathbf{Y}_k^{(\infty)}$: Proof of Proposition~\ref{prop-canonical-bases}}{Canonical Bases Yk∞: Proof of Proposition 2}}

Having understood the location of the singularities in the Borel plane, we are now able to provide a definition of a canonical basis. In general, given a formal power series $\tp(y)$ and its Borel transform $\mathcal{B}[\tp](\zeta)\equiv\hp(\zeta)$ defined in eq.~\eqref{def-Borel-transform}, one defines the Borel resummation of $\tp$ as the Laplace transform of $\hp(\zeta)$ along the direction $\theta$ 
\begin{equation}
    \mathscr{L}^{\theta}[\tp](y) \coloneqq \int_0^{e^{\ii\theta}\infty }d\zeta\, e^{-y\zeta}\hp(\zeta)\,,
\end{equation}
where by a little abuse of notation we write again $\hp(\zeta)$ for the analytic continuation along the line $\mathbb{R}^+e^{\ii\theta}$ of the Borel transform $\hp(\zeta)$. Referring to~\cite{Sauzin:2014qzt} for further details, we now prove Proposition~\ref{prop-canonical-bases}.
\begin{proof}
 In the previous section, the structure of the Borel plane for $\tp_n$ was analyzed: the rays issuing from the origin $\zeta=0$ and intersecting $\omega_{n,j}$ (where $n=0,\dots,N-1$, $j=1,\dots,N-1$) will slice the Borel plane in
 $2N$ sectors\footnote{For $N=2$ there are only $2$: $W^{(\infty)}_0 = \{ \zeta \in \mathbb{C}^\times \,|\, \Im \zeta >0\}$, $W^{(\infty)}_1 = \{ \zeta \in \mathbb{C}^\times \,|\, \Im \zeta <0\}$.} $W^{(\infty)}_k$ given by
 \begin{equation}\label{W-sector-Borel-plane}
     W^{(\infty)}_k = \left\{ \zeta \in \mathbb{C}^{\times} \,\middle|\, \phi_{k-1} <\arg \zeta <  \phi_k \right\},\qquad\phi_k = \frac{\pi}{2}+\frac{\pi}{N}k\,,
 \end{equation}
where $k=0,\dots,2N-1$ (see figure~\ref{fig:Borel plane2}). We then define a canonical basis  by resumming the formal series by Borel-Laplace resummation. The \textit{whole set} of Borel transforms $\hp_n(\zeta)$, for all $n=0,\dots,N-1$, can be analytically continued along any  straight line $\mathbb{R}^+e^{\ii\theta}$ issuing from the origin $\zeta=0$ unless $\theta= \phi_k$ for some $k=0,\dots,2N-1$.\footnote{But $k=0,1$ if $N=2$.} Thus, we are led to define the family of canonical bases $\mathbf{Y}^{(\infty)}_{k}(y)= (Y^{(\infty)}_{k,0},\dots,Y^{(\infty)}_{k,N-1})(y)$ by
\begin{equation} \label{canonical-basis-def}
    Y^{(\infty)}_{k,n}(y) \coloneqq \mathscr{L}^{\theta_k}[\tilde{Y}^{(\infty)}_n](y)= \ii^{N+1}e^{-\alpha_ny}(\alpha_ny)^{\beta+1}\mathscr{L}^{\theta_k}[\tp_n](y)\,,
\end{equation}

where $e^{\ii\theta_k}\in W_k^{(\infty)}$. Each resummation $\mathscr{L}^{\theta_k}[\tp_n](y)$ will converge to a holomorphic function on a half-plane 
\begin{equation} \label{half-plane}
    \mathbb{H}(\theta_k)\coloneqq\left\{y \in \mathbb{C} \,\middle|\,\textup{Re}\big(ye^{\ii\theta_k}\big)>R\right\},
\end{equation}
under the assumption that, for $\zeta$ in a strip surrounding a line $\arg \zeta = \theta_k \neq \phi_k$, in the Borel plane, one can find a bound $|\hp_n(\zeta)| < Ce^{R|\zeta|}$ for some constants $R,C \in \mathbb{R}_+$.\footnote{It should be possible to argue that from the integral equation~(\ref{Borel-ODE1}) but we do not have an explicit proof of this.}
However, the integration contour of the Laplace transform appearing in (\ref{canonical-basis-def}) can be deformed continuously to any straight line $\mathbb{R}^+e^{\ii\theta}$ for $\theta \in W_k^{(\infty)}$. This defines an analytic continuation  $\varphi_{k,n}(y)$ of  $\mathscr{L}^{\theta_k}[\tp_n](y)$ which admits $\tp_n$ as asymptotic power series, uniformly on closed subsectors of
\begin{equation}
{D}^{(\infty)}_k = \bigcup_{e^{\ii\theta} \in W_k^{(\infty)}}\mathbb{H}(\theta)\,,
\end{equation}
which are the domains given in eq.~(\ref{Sokal-discs}). It follows that, after reinserting the  prefactors $\ii^{N+1}e^{-\alpha_ny}(\alpha_ny)^{\beta+1}$, the asymptotic expansions claimed in~(\ref{Y-asymptotics}) hold. 
\end{proof}

\subsection{\texorpdfstring{Stokes Matrices and Monodromy Matrices: Proof of Proposition~\ref{propn:Stokes-constants}}{Stokes Matrices and Monodromy Matrices: Proof of Proposition 3}}\label{section:Stokes-matrices}
Borel sums $\varphi_{k+1,n}(y)$ and $\varphi_{k,n}(y)$ obtained by Borel resummation of the same formal power series $\tp_n$ in adjacent sectors $W_{k+1}^{(\infty)}, W_k^{(\infty)}$ can be related by the Stokes automorphism $\mathfrak{S}_k$. In general, the Stokes automorphism $\mathfrak{S}_\phi$ can be defined to act on a formal series $\tp$ so that the following holds for $\epsilon >0$ sufficiently small:
\begin{equation}
    \mathscr{L}^{\phi+\epsilon}[\tp](y)=    \mathscr{L}^{\phi-\epsilon}[\mathfrak{S}_{\phi}\tp](y)\,.
\end{equation}
The action of $\mathfrak{S}_\phi$ is non-trivial only when the Borel transform $\hp(\zeta)$ of $\tp$ is singular at some $\zeta=\omega$, with $\arg \omega = \phi$. In this case, the action of $\mathfrak{S}_\phi$ on $\tp$ will read
\begin{equation}
    \mathfrak{S}_{\phi}\tp= \tp + \sum_{i} e^{-\omega_i y}\tp_i\,,
\end{equation}
where the sum runs over all singularities $\omega_i$  of $\hp(\zeta)$ such that $\arg \omega_i = \phi$, and $\tp_i$ are formal power series (we refer to e.g. \cite{Sauzin:2014qzt,Dorigoni_2019} for further details).

Thus, in light of the analysis carried out in the previous sections, the matrices $\mathcal{S}_k$ defined by~(\ref{Stokes-matrices-def-1}), namely
\begin{equation} \label{Stokes-matrix-def}
    \mathbf{Y}^{(\infty)}_{k+1}(y) = \mathbf{Y}^{(\infty)}_k(y)\mathcal{S}_k\,, \qquad y \in D^{(\infty)}_{k+1}\cap  D^{(\infty)}_{k}
\end{equation}
can be viewed as the matrix representation of the Stokes automorphism $\mathfrak{S}_k$ acting on the formal solutions $\mathbf{\tilde{Y}}^{(\infty)}$~\eqref{formal-basis}, or equivalently on the formal series $\tp_n = (\alpha_ny)^{-1}f_n^{(\infty)}$. 

In the following subsection we will show how this observation enables us to determine the shape of the Stokes matrices $\mathcal{S}_k$.

\subsubsection{The Action of the Stokes Automorphism}

The key observation to make is that the $\mathbb{Z}_{N}$ invariance $y \rightarrow \alpha_ny$ of the oper equation~(\ref{oper-y}) constrains the action of the Stokes automorphism. 
Consider a formal solution $\tilde\varphi_n$, and the Stokes line at angle $\phi_{2n+j} = \arg\omega_{n,j}$, which separates the sectors $W^{(\infty)}_{2n+j}$ and $W^{(\infty)}_{2n+j+1}$ in the Borel plane. Denoting for simplicity $\mathfrak{S}_k \equiv \mathfrak{S}_{\phi_k}$, the resummations in these sectors will be related by the Stokes automorphism: with $e^{\ii\theta_k} \in W^{(\infty)}_k$, we have
\begin{equation}\label{Stokes-automorphism}
    \mathscr{L}^{\theta_{2n+j+1}}[\tp_n] = \mathscr{L}^{\theta_{2n+j}}[\mathfrak{S}_{2n+j}\tp_n]\,, \qquad \mathfrak{S}_{2n+j}\tp_n = \tp_n + e^{-\omega_{n,j}y}\tp_{n,j}\,,
\end{equation}
where $\tp_{n,j}$ is an as-yet undetermined formal power series.
By reinserting the prefactors $e^{-\alpha_ny}(\alpha_ny)^{\beta+1}$, we observe that the exponential behavior of the second term is $e^{-(\alpha_n+\omega_{n,j})y} = e^{-\alpha_{n+j}y}$. Therefore, since both terms should individually solve the oper equation (formally), we conclude that the unspecified $\tp_{n,j}$ is in fact proportional to $\tp_{n+j}$. Specifically, we may write
\begin{equation}\label{Stokes-aut-2}
     \mathfrak{S}_{2n+j}\tp_n = \tp_n + s_{n+j,n}\alpha_j^{\beta+1}e^{-\omega_{n,j}y}\tp_{n+j}\,,
\end{equation}
where $s_{m,n}$ are Stokes constants and $n+j$ is to be understood modulo $N$.
We may phrase this as saying that the action of the Stokes automorphism $\mathfrak{S}_k$ on $\tp_n$ only receives a nontrivial contribution from $\tp_{k-n}$.

We are now able to prove point (a) of Proposition~\ref{propn:Stokes-constants}.
\begin{proof} \hfill
(i)   First, let us argue that only $N-1$ Stokes constants are independent. We recall that the  action of the Stokes automorphism $\mathfrak{S}_{k}$ on $\tp_n$ only depends on the local properties of $\hp_n(\zeta)$ in neighbourhoods of the line $\arg \zeta = \phi_{k}$ (we refer to e.g.~\cite{Sauzin:2014qzt} for details).
    It follows then from Lemma~\ref{lemma-borel-plane} that, since one has as $\hp_n(\zeta)=\alpha_{-n}\hp_0(\zeta\alpha_{-n})$, the Stokes constants $s_{m,n}$ cannot be independent, namely: the $N$ singularities $\omega_{n,j}$, $n=0,1,\dots,N-1$ must be associated with the same Stokes constant. 
Recalling~\eqref{Stokes-aut-2}, we have specifically
\begin{equation}
    m-n = m'-n' \quad\implies\quad s_{m,n} = s_{m',n'}\,.
\end{equation}
Evidently a similar relation also holds when the left hand side is only satisfied mod $N$, however one must account for the monodromy of $y^{\beta+1}$, which is nontrivial when $N$ is even and implies $\tilde Y^{(\infty)}_{n-N}=(-1)^{N-1}\tilde Y^{(\infty)}_n$.
Therefore the minimal number of independent Stokes constants can be labeled by $s_i$ with $i=1,\dots, N-1$, if we define $s_{n-m}\coloneqq s_{m,n}$ and, by consistency with the previous observation, we let  $s_{i-N}=(-1)^{N-1}s_i$.

(ii) Second, let us demonstrate how the Stokes constants $s_i$, $i=1,\dots,N-1$ determine any Stokes matrix $\mathcal{S}_k$ entirely. Recalling the definition~(\ref{canonical-basis-def}) of the canonical bases $Y_{k,n}^{(\infty)}$,  we may rewrite the action of the Stokes automorphism~(\ref{Stokes-aut-2}) in terms of Stokes matrices $\mathcal{S}_k$
\begin{equation} \label{eq:nontrivial-jumps}
    Y^{(\infty)}_{k+1,n}(y) = Y^{(\infty)}_{k,n}(y) + Y^{(\infty)}_{k,k-n}(y)[{\cal S}_k]_{k-n,n}\,,
\end{equation}
where $k-n$ is to be understood modulo $N$.
It follows from the above equation that no two Stokes matrices $\mathcal{S}_k, \mathcal{S}_{k'}$ share off-diagonal entries, and we in fact have $[{\cal S}_k]_{m,n} = s_{m,n}$. 
From eq.~(\ref{eq:nontrivial-jumps}) we infer that the Stokes matrices $\mathcal{S}_k$ all have entries $1$ on the main diagonal.

Let us now describe the exact location of the non-vanishing, off-diagonal entries of the Stokes matrix $\mathcal{S}_k$. It is a consequence of Lemma~\ref{lemma-borel-plane} that, while there are in total $2N$  Stokes lines in the Borel plane,\footnote{But only $2$ Stokes lines for $N=2$.} each Borel transform $\hp_n(\zeta)$ only has singularities on the $N-1$ lines which intersect its singularities $\omega_{n,j}$.
Therefore, given $\tp_n$ and $\phi_k$, a Stokes jump for $\tp_n$ can only occur if there exists a $j\in\{1,\dots,N-1\}$ such that $k = 2n+j \mod 2N$.
This along with equation~\eqref{eq:nontrivial-jumps} constrains the Stokes matrices to be of the following form
\begin{equation} \label{eq:Stokes-matrix-structure}
    [{\cal S}_k]_{m,n} =
    \begin{cases}
        1 & m=n \\
        s_{n-m} & m = k-n \,(\operatorname{mod} N) \,\land\, k-2n\; (\operatorname{mod} 2N)\in\{1,\dots,N-1\} \, \\
        0 & \text{otherwise}
    \end{cases},
\end{equation}
where $k=0,\dots,2N-1$ and $m,n=0,\dots,N-1$.
This results in matrices with ones on the diagonal and Stokes constants on the anti-diagonal given by $m+n=k$, starting from the first available element just above the diagonal, and continuing towards the top right, wrapping around to the opposite side if necessary. 
\end{proof}

\begin{rem}
    As was noted previously, the case $N=2$ is somewhat exceptional, as there are only two singularities in the Borel plane, as opposed to $2N=4$, and therefore only two sectors and associated Stokes matrices.
    However eq.~\eqref{eq:Stokes-matrix-structure} implies that ${\cal S}_0={\cal S}_2=\mathds{1}$, such that this merging of sectors is also accounted for in the general treatment.
    In particular this means that the subsequent results, while obtained using the general formula, also apply to this case, keeping in mind that only the nontrivial Stokes matrices ${\cal S}_1,{\cal S}_3$ describe a change in sector.
\end{rem}
\begin{rem}
With the conventions above, the entries in the lower-triangle of $\mathcal{S}_k$ are multiplied by $(-1)^{N-1}$ when expressed in terms of $s_i$, $i=1,\dots,N-1$. Furthermore, the Stokes constant $s_i$ is associated with the Borel plane singularities at $\omega_{n, N-i}$, for all $n=0,\dots,N-1$. 
Moreover, it is easy to check by counting the solutions to the second inequality in~(\ref{eq:Stokes-matrix-structure}) that 
the number of nonzero Stokes constants in ${\cal S}_k$ is given by $\tfrac{N-1}{2}$ for $N$ odd, $\tfrac{N}{2}-1$ for $N$ and $k$ even, and $\tfrac{N}{2}$ for $N$ even and $k$ odd. The Stokes matrices are therefore sparse (the number of non-zero entries grows only linearly in $N$). Finally, it also follows from eq.~(\ref{eq:Stokes-matrix-structure}) that $\det\mathcal{S}_k=1$.
\end{rem}

\subsubsection{Stokes Constants From Monodromy}\label{Monodromy-matrix-section}

Having understood in eq.~(\ref{eq:Stokes-matrix-structure}) the shape of the Stokes matrices, we may now construct the monodromy matrix $M_k$ for a canonical basis $\mathbf{Y}^{(\infty)}_k(y)$. 
\begin{lem}\label{monodromy-lem}
    The monodromy matrix $M_k$ defined by $ \mathbf{Y}^{(\infty)}_k(y)M_k \coloneqq \mathbf{Y}^{(\infty)}_k(y\alpha_1)$ can be expressed in terms of Stokes matrices as
    \begin{equation}\label{eq:monodromy-mat-Scanonical basis}
        M_k = {\cal S}_{k}{\cal S}_{k+1}P_N\,, \qquad P_N = \begin{pmatrix}
            0 &&& (-1)^{N-1}\\
            1 & 0 \\
            & \ddots & \ddots \\
            && 1 & 0 
        \end{pmatrix}.
    \end{equation}
\end{lem}
\begin{proof}
In the $y\,\propto\,z^{1/N}$ variable, the monodromy matrix amounts to expressing the analytic continuation of $\mathbf{Y}_k^{(\infty)}(y)$, namely $\mathbf{Y}_k^{(\infty)}(\alpha_1y)$ in terms of the original basis $\mathbf{Y}_k^{(\infty)}(y)$. Explicitly, we have
\begin{align}
   {Y}_{k,n}^{(\infty)}(\alpha_1y)
    =   \mathscr{L}^{\theta_k}[\tilde{Y}^{(\infty)}_{k,n}](\alpha_1y)
    = \ii^{N+1}e^{-\alpha_{n+1}y}(\alpha_{n+1}y)^{\beta+1}\mathscr{L}^{\theta_k}[\tp_n](y\alpha_1)\,.
\end{align}
We may then notice that the resummation in the last term can be related to the resummation of $\tp_{n+1}$ in a different sector by
\begin{align}
    \mathscr{L}^{\theta_k}[\tp_n](y\alpha_1)
    &= \int_{0}^{e^{\ii\theta_k}\infty }d\zeta e^{-\zeta \alpha_1 y}\hp_{n}(\zeta) = \int_{0}^{e^{\ii\theta_k+\frac{2\pi}{N}}\infty }d\zeta e^{-\zeta y}\hp_{n}(\alpha_{-1}\zeta)\alpha_{-1} \\ \nonumber
   & =     \mathscr{L}^{\theta_k+\frac{2\pi}{N}}[\tp_{n+1}](y)\,,
\end{align}
where to pass to the last equality we use Lemma~\ref{lemma-borel-plane}. It now suffices to observe that
\begin{equation}
    \mathscr{L}^{\theta_k + \frac{2\pi}{N}}[\tp_{n+1}](y) =  \mathscr{L}^{\theta_k}[\mathfrak{S}_{k+1} \circ\mathfrak{S}_{k}\tp_{n+1}](y)\,.
\end{equation}
This can be explained as follows. Recall that the sectors $W_k^{(\infty)}$ defined in (\ref{W-sector-Borel-plane}) have an opening angle of $\pi/N$. Therefore, if we fix $\theta_k \in W_k^{(\infty)}$ we will have
\begin{equation}
    \theta_k < \phi_k< \phi_{k+1}<\theta_k +\frac{2\pi}{N}\,.
\end{equation}
It follows that, in order to relate the Laplace transforms along $\theta_k$ and $\theta_k + 2\pi/N$, the action of the two Stokes automorphisms $\mathfrak{S}_{k+1}, \mathfrak{S}_{k}$ associated with the angles $\phi_{k+1}$, $\phi_k$ is required. An example for $N=5$ is shown in figure~\ref{fig:Borel plane2} by the black arc. To reach~(\ref{eq:monodromy-mat-Scanonical basis}) one has to first represent the Stokes automorphism by Stokes matrices as in~(\ref{Stokes-aut-2},~\ref{eq:nontrivial-jumps}); secondly one has to represent the relabelling $\tp_n \rightarrow\tp_{n+1}$ by a permutation matrix $P_N$ which sends $Y^{(\infty)}_{k,n}\to Y^{(\infty)}_{k,n+1}$ for $n=0,\dots,N-2$ and $Y^{(\infty)}_{k,N-1}\to(-1)^{N-1}Y^{(\infty)}_{k,0}$. 
Observe that $\det P_N=1$ so that $\det M_k=1$, which is in agreement with $M_k$ being (conjugated to) the holonomy matrix of a flat $\slNC$ connection.
\end{proof}
We are now ready to prove point (b) of Proposition~\ref{propn:Stokes-constants}.
\begin{proof}
   Without loss of generality, we will focus on the case $k=0$ and compute the corresponding monodromy matrix $M_0$.
   
First, we compute $\mathcal{S}_0$. From eq.~\eqref{eq:Stokes-matrix-structure} we may infer from the first condition that the Stokes constants all lie on the anti-diagonal with $m+n = N$ (the one below the main one). From the second condition $2N-2n \in \{1,\dots,N-1\}$,  we find $n= \ceil*{\tfrac{N+1}{2}},\ceil*{\tfrac{N+1}{2}}+1,\dots,N-1$, which implies that  the Stokes constants $s_i$ are further constrained to lie on the upper triangle of $\mathcal{S}_0$. Note that by setting $k=0$ we only find constants $s_{N-2}, s_{N-4},\dots$ having a label of the same parity of $N$.

Next, we compute ${\cal S}_1$. Similarly to the previous case, we find that the Stokes constants are constrained to lie on the anti-diagonal $m+n=N+1$. Moreover, from the second condition $1-2n\; (\operatorname{mod} 2N)\in\{1,\dots,N-1\}$ we find a non-trivial entry on the second row and first column ($k=1, n=0$) corresponding to $s_{1,0} = (-1)^{N-1}s_{N-1}$, together with entries in the upper triangle on the columns $n= \ceil*{\tfrac{N+2}{2}},\ceil*{\tfrac{N+2}{2}}+1,\dots,N-1$. This time, we only find Stokes constants $s_{N-1}, s_{N-3},\dots$ having a label of the opposite parity of $N$. 

Due to the sparsity of the Stokes matrices, the product ${\cal S}_0{\cal S}_1$ is found to be the matrix such that the diagonal elements equal to $1$ and the entry $[{\cal S}_0{\cal S}_1]_{m,n}$ is non-vanishing if and only if either $[\mathcal{S}_0]_{m,n}$ or $[\mathcal{S}_1]_{m,n}$ is non-vanishing as well, and equal thereto. Note that the product ${\cal S}_0{\cal S}_1$ contains all the independent Stokes constants $s_1,\dots,s_{N-1}$, as can be easily verified by counting the different solutions for the column index $n$ provided above. Recalling Lemma~\ref{monodromy-lem}, we then obtain the monodromy matrix by $M_0 = {\cal S}_0{\cal S}_1P_N$
\begin{equation} \label{eq:canonical-monodromy}
    M_0 = 
    \begin{pmatrix}
        0 & 0 & 0 & \dots & 0 & 0 & (-1)^{N-1} \\
        1 & 0 & 0 & \dots & 0 & s_{N-2} & s_{N-1} \\
        0 & 1 & 0 &\dots & s_{N-4} & s_{N-3} & 0 \\
        0 & 0 & 1 & \dots & s_{N-5} & 0 & 0 \\
        \vdots & \vdots & \vdots && \vdots & \vdots & \vdots \\
        0 & 0 & 0 & \dots & 1 & 0 & 0 \\
        0 & 0 & 0 & \dots & 0 & 1 & 0
    \end{pmatrix}.
\end{equation}
The ``staircase'' sequence of Stokes constants terminates at $s_1$, which lies on the main diagonal, with $s_2$ and $s_1$ aligned vertically if $N$ is even, and horizontally if $N$ is odd.

The characteristic polynomial of $M_0$ is easily evaluated by recursively expanding the determinant by the Laplace rule, first with respect to the last column and then the first row, which ultimately gives
\begin{equation}
    \det(\lambda-M_0) = \lambda^N - \sum_{i=1}^{N-1} s_i\lambda^{N-i} + (-1)^{N} \,.
\end{equation}
It follows that the Stokes constants are related to the eigenvalues $\Sigma_j = e^{2\pi\ii\sigma_j}$ of the monodromy matrix by
\begin{equation}
    s_i = (-1)^{i+1}e_i(\boldsymbol\Sigma)\,,\qquad i=1,\dots,N-1\,,
\end{equation}
where $e_i(\boldsymbol{\Sigma)}$ is the elementary symmetric polynomial of degree $i$ in the $N$ variables $\Sigma_j$.
\end{proof}

\subsection{\texorpdfstring{Canonical Bases $\mathbf{Y}^{(0)}_k$}{Canonical Bases Yk0}} \label{sec:Y^0}

We consider formal asymptotic solutions to the oper equation as $z\to0$ in the variable $y'=\frac{N\Lambda}{\hbar}z^{-1/N}$, so that the limit $z\to0$ corresponds to $y'\to\infty$. 
In this way the problem becomes largely analogous to the one for $z\to\infty$ since both punctures have the same Poincaré rank. Specifically, the formal solutions to the oper equation will take the form
\begin{equation}
    \tilde Y^{(0)}_n(y') = (-\ii)^{N+1}e^{-\alpha_{-n}y'}(\alpha_{-n}y')^\beta \tilde f^{(0)}_n(y')\,,\qquad n=0,\dots,N-1\,,
\end{equation}
where $\tilde f^{(0)}_n = 1+O\big({y'}^{-1}\big)$ is a formal power series in ${y'}^{-1}$.
\begin{lem}
    For the canonical bases of solutions $\mathbf{Y}^{(0)}_k=(Y^{(0)}_{k,0},\dots,Y^{(0)}_{k,N-1})$ defined as the Borel resummations of formal solutions $\tilde Y^{(0)}_n(y')$ the Stokes and monodromy matrices are the same as those for $\mathbf{Y}^{(\infty)}_k$, that is
    \begin{equation}
        \mathbf{Y}^{(0)}_{k+1}(y') = \mathbf{Y}^{(0)}_k(y'){\cal S}_k\,,\qquad \mathbf{Y}^{(0)}_k(y'\alpha_{-1}) = \mathbf{Y}^{(0)}_k(y')M_k\,.
    \end{equation}
\end{lem}
\begin{proof}
For the most part the method is exactly the same, however one must take care of the change in orientation coming from the inversion, which can be compensated by a change of ordering of the basis elements as well as the Stokes lines.
Here we only cover the necessary definitions and modifications and state the result, with the technical points being analogous to those for the $z\to\infty$ case.

We define $\tp^{(0)}_n(y') = (\alpha_{-n}y')^{-1}\tilde f^{(0)}_n(y')$ and compute its Borel transform
\begin{equation}
    \hp^{(0)}_n(\zeta) \coloneqq \CB\left[ \tp^{(0)}_n \right] = \alpha_n\sum_{k=0}^\infty\frac{a_k'}{k!}(\alpha_n\zeta)^k\,,\qquad \tp^{(0)}_n(y) = \sum_{k=0}^\infty a_k'(\alpha_{-n}y')^{-k-1}\,.
\end{equation}
Each function $\hp^{(0)}_n(\zeta)$ again has $N-1$ singularities lying at the non-zero solutions to the equation $(\omega+1)^N=1$, which we label as
\begin{equation}
    \omega^{(0)}_{n,j} = \alpha_{-n}(\alpha_{-j}-1)\,,\qquad j=1,\dots,N-1\,,
\end{equation}
so that the Stokes lines are given by
\begin{equation}
    \phi^{(0)}_{2n+j} = \arg\omega^{(0)}_{n,j} \quad\implies\quad \phi^{(0)}_k = -\frac{\pi}{2}-\frac{\pi}{N}k\,,\qquad k=0,\dots,2N-1\,.
\end{equation}
Defining the sectors 
\begin{equation}
    W^{(0)}_k = \left\{ \zeta\in\BC^{\times} \,\middle|\, \phi^{(0)}_k < \arg\zeta < \phi^{(0)}_{k-1} \right\}
\end{equation}
in the Borel plane, the resummed series are then obtained as the Laplace transforms
\begin{equation} \label{eq:Laplace-0}
    \varphi^{(0)}_{k,n}(y') \coloneqq \mathscr{L}^{\theta_k}\left[ \tp^{(0)}_n \right](y') = \int_0^{e^{\ii\theta_k}\infty}d\zeta\, e^{-\zeta y'}\hp^{(0)}_n(\zeta)\,,\qquad \theta_k\in W^{(0)}_k\,,
\end{equation}
which, after assembling together all the resummations for $\theta_k\in W_k^{(0)}$, yield holomorphic functions admitting uniform asymptotic expansions on closed subsectors of the sectors $D^{(0)}_k$ given by eq.~\eqref{Sokal-discs}.\footnote{We remark that, if we assume $\Lambda/\hbar\in\BR$, the sectors  $D_k^{(0, \infty)}$cover the same angular range in the $z$ variable.}  The full canonical solutions in their respective sectors are then given by
\begin{equation}
    Y^{(0)}_{k,n}(y') = (-\ii)^{N+1}e^{-\alpha_{-n}y'}(\alpha_{-n}y')^{\beta+1}\varphi^{(0)}_{k,n}(y')\,.
\end{equation}
Defining the Stokes matrices analogously as the transition matrices
\begin{equation}
    \mathbf{Y}^{(0)}_{k+1}(y') = \mathbf{Y}^{(0)}_k(y'){\cal S}'_k\,,
\end{equation}
it is straightforward to verify from the above definitions that their matrix elements $s'_{ij} \coloneqq [{\cal S}'_k]_{ij}$ are subject to the same constraints~\eqref{eq:Stokes-matrix-structure} as those of ${\cal S}_k$.
The monodromy matrix $\mathbf{Y}^{(0)}_k(\alpha_{-1}y') \eqqcolon \mathbf{Y}^{(0)}_kM'_k$ can then be obtained from the Laplace transform~\eqref{eq:Laplace-0} as
\begin{equation}
    M'_k = {\cal S}'_k{\cal S}'_{k+1}P_N\,,
\end{equation}
with $P_N$ defined as in Lemma~\ref{monodromy-lem}.\footnote{Here it is important that we have chosen to orient both the basis vectors and the Stokes lines in the opposite way from the $z\to\infty$ case. The former ensures that we get the same permutation matrix $P_N$, despite rotating $y'$ in the opposite direction, while the latter ensures that this does not modify the structure of the Stokes matrices.}
By comparing characteristic polynomials it then follows that $s'_{ij} = s_{ij}$ and therefore ${\cal S}'_k = {\cal S}_k$ and $M'_k = M_k$ for all $k=0,\dots,2N-1$.
\end{proof}

\section{Floquet Bases and Riemann-Hilbert Problem} \label{sec:Floquet}

In this section we construct explicit Floquet solutions to the oper equation, by making precise the relation between the Toda chain and opers on $C_{0,2}$.
In particular, this allows us to solve the direct monodromy problem, as well as the Riemann-Hilbert problem outlined in Section~\ref{oper-section}, in terms of a single nonlinear integral equation.

\subsection{Construction of Floquet Bases}
In order to construct Floquet solutions, one may use an ansatz of the form
    \begin{equation}\label{Floquet-ansatz}
        F_\sigma(z) = \sum_{n\in\BZ}q_nz^{\sigma+n}\,,
    \end{equation}
    parameterized by a complex number $\si$.
    The monodromy of this solution around $z=0$ is diagonal, represented by multiplication with 
    $e^{2\pi \mathrm{i}\si}$.
Inserting this ansatz into the oper equation~\eqref{oper-z} yields a recurrence relation of the form
    \begin{equation} \label{eq:Floquet-recurrence}
        t(-\ii\hbar(\sigma+n))q_n = \Lambda^N\left( \ii^Nq_{n-1} + \ii^{-N}q_{n+1} \right)\,.
    \end{equation}
The solutions $(q_n)_{n\in\BZ}$ of this recurrence relations in general will not define a convergent series
when inserted into the right side of \rf{Floquet-ansatz}. Only for special choices 
of the parameter $\si$ one will find convergent series of the form \rf{Floquet-ansatz} defining 
analytic functions of the variable $z$.

    A simple, but crucial observation is the following: The recurrence relation \rf{eq:Floquet-recurrence}
    is closely related to the 
     Baxter equation~\eqref{Baxter-equation}. It follows immediately that for any solution $q(\lambda)$
     to~\eqref{Baxter-equation}, the definition 
    $q_n \coloneqq  q(-\ii\hbar(\sigma+n))$ yields a solution to the recurrence relation~\eqref{eq:Floquet-recurrence}.
We may therefore apply the results reviewed in Section~\ref{Toda-review} for the construction of Floquet-solutions.

At the end of Section~\ref{Toda-review} we had in particular 
observed that the sequences $(Q^\pm_{\bm\tau}(-\ii\hbar(\sigma+n)))_{n\in\BZ}$
are rapidly decaying for $n\ra\pm\infty$ if 
the parameter $\si$ is proportional to one of the zeros of the Hill determinant. These observations 
lead us to the 
following result.

\begin{propn}\label{Proposition-sigma-delta}
    For given tuple $\bm{\tau}$ of roots of $t(\la)$, and given parameter $\Lambda$, let us define
    the Laurent  series
    \begin{subequations} \label{eq:Floquet-from-Baxter}
    \begin{align}
        F^{(0)}_j(z) &= \sum_{n\in\BZ} Q_{\boldsymbol\tau}^+\big(\de_j(\bm{\tau},\Lambda)-\ii\hbar n\big)z^{\sigma_j+n}\,,\qquad j=1,\dots,N\,, \\
        F^{(\infty)}_j(z) &= \sum_{n\in\BZ}Q_{\boldsymbol\tau}^-\big(\de_j(\bm{\tau},\Lambda)-\ii\hbar n\big)z^{\sigma_j+n}\,,\qquad j=1,\dots,N\,, 
    \end{align}
    \end{subequations}
    with 
    $\de_j(\bm{\tau},\Lambda)$ defined to be the zeros of the quantum Wronskian through eq.~\rf{Hill-formula}, and 
     parameters $\si_j$ being related to $\de_j(\bm{\tau},\Lambda)$ as $\si_j = \frac{\mathrm{i}}{\hbar}\de_j(\bm{\tau},\Lambda)$.
     The Laurent series on the right of eq.~\rf{eq:Floquet-from-Baxter}
    are absolutely convergent, 
    defining two  bases of Floquet solutions to the oper equation~\eqref{oper-z}.
\end{propn}

See Appendix~\ref{app:G-expansion} for
more details on the analytic properties of the Floquet solutions. 

The construction of the two linearly independent solutions $Q_{\boldsymbol\tau}^\pm$ reviewed in 
Subsection~\ref{Toda-review} can now be applied to the solution theory of the oper differential 
equation. Recall that the monodromy data of the oper equation of interest here are completely 
determined by the eigenvalues $\si_j$, $j=1,\dots,N-1$. Proposition~\ref{Proposition-sigma-delta}
relates these data to the zeros $\de_j(\bm{\tau},\Lambda)$ of the quantum Wronskian (\ref{Quantum-Wronskian}). Computing the monodromies
of the oper therefore reduces to finding said zeros. 

One could alternatively replace the 
functions $Q_{\boldsymbol\tau}^\pm$ in the construction \rf{eq:Floquet-from-Baxter}
by the closely related functions $\mathrm{Q}_{\boldsymbol\de}^\pm$
constructed 
from the solutions to the nonlinear integral equation \rf{NLIE} in Section~\ref{BaxterfromNLIE}. 
Recall that this construction produces solutions $\mathrm{Q}_{\boldsymbol\de}^\pm$ 
to the Baxter equation which have the property that the quantum 
Wronskian of $\mathrm{Q}_{\boldsymbol\de}^\pm$ vanishes at $\la=\de_j$, for $j=1,\dots,N$,
from given input data
$\bm{\de}=(\de_1,\dots,\de_N)$.
The relation $\mathrm{Q}_{\boldsymbol\de(\bm{\tau})}^\pm=Q_{\boldsymbol\tau}^\pm$ observed in Section~\ref{BaxterfromNLIE} 
immediately allows us 
to conclude that we may solve the natural analog of the Riemann-Hilbert problem in this way,
recovering the oper from the monodromy data $\bm\si= \frac{\mathrm{i}}{\hbar}\bm\de$. 

\subsection{Connection Matrix and Generating Function}\label{subsection:generating-function}

Another useful application of the solution theory of the Baxter equation to 
the study of solutions to the oper equation follows from the simple observation that  
equations \rf{zetaj-def} and \rf{eq:Floquet-from-Baxter} immediately imply that
the matrix elements of the connection matrix defined in \rf{etaj-def}
satisfy
\begin{equation}\label{eta_j-zeta_j relation}
e^{2\pi \mathrm{i}\,\eta_j(\bm{\si},\Lambda)}=\zeta_j(\bm{\de},\Lambda)\Big|_{\bm{\de} = -{\mathrm{i}}{\hbar}\bm{\si}}=
\frac{\mathrm{Q}_{\boldsymbol\de}^+(\de_j)}{\mathrm{Q}_{\boldsymbol\de}^-(\de_j)}\Bigg|_{\bm{\de} = -{\mathrm{i}}{\hbar}\bm{\si}}.
\end{equation}
Inserting the expressions for $\mathrm{Q}_{\boldsymbol\de}^\pm(\de_j)$ in terms of the solution 
to the integral equation, and using \rf{YY-eta-potential}, it becomes possible to verify that
\begin{equation}
2\pi\ii\,\eta_j=\frac{\pa}{\pa \de_j}\CY(\bm{\de},\Lambda)\,,
\end{equation}
see Appendix~\ref{W=Y appendix} for the details. 
It follows that the Yang-Yang function $\CY(\bm{\de},\Lambda)$ can serve as
a potential for the functions $\eta_j(\si,\Lambda)$, $j=1,\dots,N$, suggesting to define
\begin{equation}\label{CW-from-CY}
\mathcal{S}(\bm{\si},\Lambda)\coloneqq\CY(\bm{\de},\Lambda)\Big|_{\bm{\de} = -{\mathrm{i}}{\hbar}\bm{\si}}\,.
\end{equation}

It can furthermore be shown (see Appendix~\ref{W=Y appendix} for the proof) that eq.~\rf{CW-from-CY} implies in addition
\begin{equation}\label{dLambda-deriv} 
\ii\hbar\frac{\partial \mathcal{S}(\boldsymbol{\delta},\Lambda)}{\partial \log \Lambda^{2N}} = u(\boldsymbol{\delta}, \Lambda) \coloneqq -E_2(\bm\de,\La)\,,
\end{equation}
where $u(\boldsymbol{\delta},\Lambda)\coloneq-E_2(\boldsymbol{\tau}(\boldsymbol{\delta},\Lambda))$. This relation generalizes relations that are known to be satisfied by the generating functions 
of opers for the cases with $N=2$ and regular singularities \cite{Teschner:2010je,Litvinov:2013sxa}.
We are thereby led to the conclusion that the generating function $\mathcal{S}$ 
of the subspace of opers  for $C_{0,2}$ with minimally irregular singularities
coincides with the Yang-Yang function $\CY$,
verifying a consequence of the conjectures of~\cite{Nekrasov:2009rc} and~\cite{Nekrasov:2011bc}.

\section{Connection Matrix and Analytic Langlands Correspondence} \label{sec:connection-langlands}

The equality between the Yang-Yang function of the Toda chain and the generating function of opers observed in the previous section suggests a geometric reformulation of the Toda quantization conditions~\eqref{q-cond-delta} in the spirit of the analytic Langlands correspondence~\cite{Etingof:2019pni}.
Here we show that this is indeed the case by relating them to the condition of trivial parallel transport between the canonical bases $\mathbf{Y}^{(0)}_0$ and $\mathbf{Y}^{(\infty)}_0$.
To do so, we solve the connection problem by computing the connection matrix $E(\bm\si)$ explicitly via the change of basis between the Floquet and canonical bases.
The latter is obtained through an explicit integral representation~\eqref{eq:Stokes-max} for the maximally decaying solutions $\chi^{(0,\infty)}$ at $z=0,\infty$ respectively, from which the full canonical basis can be constructed using the knowledge of the Stokes matrices obtained in Section~\ref{section:Stokes-matrices}.

\subsection{An Explicit Construction of Maximally Decaying Solutions}

As a first step towards solving the connection problem we construct the following two solutions.
\begin{lem}
    Let
    \begin{equation}\label{eq:little q}
       q_{\bm\de}^{\pm}(\lambda)= \frac{Q^{\pm}_{\boldsymbol{\delta}}(\lambda)}{\prod_{j=1}^N e^{-\frac{\pi \lambda}{\hbar}}\sinh{\frac{\pi}{\hbar}(\lambda-\delta_j)}}\,,
    \end{equation}
    which has poles at $\sigma_j+\BZ$, $j=1,\dots,N$ with $\sigma_j=\tfrac{\ii}{\hbar}\delta_j$.
    Then the functions
    \begin{subequations} \label{eq:Stokes-max}
    \begin{align}
        \chi^{(0)}(z) =& \lim_{n\to\infty}\int_{L_n^{(0)}}\frac{ds}{2\pi\ii}\, q^+_{\bm\de}(-\ii\hbar s)z^s\,,\\
        \chi^{(\infty)}(z) =& \lim_{n\to\infty}\int_{L_n^{(\infty)}}\frac{ds}{2\pi\ii}\, q^-_{\bm\de}(-\ii\hbar s)z^s\,,
    \end{align}
    \end{subequations}
    where the sequences of contours $L_n^{(0,\infty)}$ are chosen as in figure~\ref{fig:chi-contours1}, are well-defined analytic solutions to the oper equation~(\ref{oper-z}) on $C_{0,2}$.
\end{lem}
\begin{figure}[ht!]
\centering

\begin{tikzpicture}[scale=0.7]


\draw [thick, ->] (-5,0) -- (5,0);
\draw [thick, ->] (0,-3) -- (0,4);
\draw (-5,-3) rectangle (5,4);

\draw [rounded corners=5pt, thick, RoyalBlue, decoration={markings, mark=at position 0.1 with {\arrow{<}}}, postaction={decorate}] (-4,-2.2) -- (2.1,-2.2) -- (2.1,-1.6) -- (-2.4,-1.6) -- (-2.4,0.2) -- (1,0.2) -- (1,1.7) -- (-4,1.7);
\draw [dashed, thick, RoyalBlue] (-4,-2.2) -- (-5,-2.2);
\draw [dashed, thick, RoyalBlue] (-4,1.7) -- (-5,1.7);

\begin{scope}
    \clip (-5,-3) rectangle (5,4);
    
    \foreach \i in {1,...,7}
        \draw [rounded corners=5pt, ultra thin, RoyalBlue, opacity=1/\i] (-4+\i,-2.2) -- (2.1+\i,-2.2) -- (2.1+\i,-1.6) -- (-2.4+\i,-1.6) -- (-2.4+\i,0.2) -- (1+\i,0.2) -- (1+\i,1.7) -- (-4+\i,1.7);
\end{scope}

\node at (-2,2.2) {\footnotesize $L_0^{(0)}$};

\foreach \i in {1,...,3}
    \node [scale=0.5, ultra thin, draw=RoyalBlue, draw opacity=1/\i, rounded corners, fill=white] at (1+\i,0.95) {$L_{\i}^{(0)}$};

\node at (4.75,3.75) {$s$};
\draw (4.5,4) -- (4.5,3.5) -- (5,3.5);

\foreach \i in {-2,...,7}
    \filldraw [gray] (-2.8+\i,-0.5) circle (1.5pt);
\foreach \i in {-5,...,4}
    \filldraw [gray] (0.6+\i,1.4) circle (1.5pt);
\foreach \i in {-6,...,3}
    \filldraw [gray] (1.8+\i,-1.9) circle (1.5pt);

\filldraw (-2.8,-0.5) circle (2pt) node [left] {$\sigma_1$};
\filldraw (0.6,1.4) circle (2pt) node [below] {$\sigma_2$};
\filldraw (1.8,-1.9) circle (2pt) node [left] {$\sigma_3$};

\draw [thick, ->] (6,0) -- (16,0);
\draw [thick, ->] (11,-3) -- (11,4);
\draw (6,-3) rectangle (16,4); 

\draw [rounded corners=5pt, thick, RoyalBlue, decoration={markings, mark=at position 0.95 with {\arrow{>}}},
        postaction={decorate}] (15,-2.2) -- (12.5,-2.2) -- (12.5,-0.8) -- (7.9,-0.8) -- (7.9,0.2) -- (11.3,0.2) -- (11.3,1.7) -- (15,1.7);
\draw [dashed, thick, RoyalBlue] (15,-2.2) -- (16,-2.2);
\draw [dashed, thick, RoyalBlue] (15,1.7) -- (16,1.7);

\begin{scope}
    \clip (6,-3) rectangle (16,4);

    \foreach \i in {1,...,6}
        \draw [rounded corners=5pt, ultra thin, RoyalBlue, opacity=1/\i] (15-\i,-2.2) -- (12.5-\i,-2.2) -- (12.5-\i,-0.8) -- (7.9-\i,-0.8) -- (7.9-\i,0.2) -- (11.3-\i,0.2) -- (11.3-\i,1.7) -- (15-\i,1.7);
\end{scope}

\node at (13,2.2) {\footnotesize $L_0^{(\infty)}$};

\foreach \i in {1,...,3}
    \node [scale=0.5, ultra thin, draw=RoyalBlue, draw opacity=1/\i, rounded corners, fill=white] at (11.3-\i,0.95) {$L_{\i}^{(\infty)}$};
    
\node at (15.75,3.75) {$s$};
\draw (15.5,4) -- (15.5,3.5) -- (16,3.5);

\foreach \i in {-2,...,7}
    \filldraw [gray] (8.2+\i,-0.5) circle (1.5pt);
\foreach \i in {-5,...,4}
    \filldraw [gray] (11.6+\i,1.4) circle (1.5pt);
\foreach \i in {-6,...,3}
    \filldraw [gray] (12.8+\i,-1.9) circle (1.5pt);

\filldraw (8.2,-0.5) circle (2pt) node [right] {$\sigma_1$};
\filldraw (11.6,1.4) circle (2pt) node [right] {$\sigma_2$};
\filldraw (12.8,-1.9) circle (2pt) node [right] {$\sigma_3$};

\end{tikzpicture}

\caption{Integration contours $L_n^{(0)}$ (left) and $L_n^{(\infty)}$ (right), shown here for $N=3$ in the complex $s$-plane. Each pole $\sigma_j$ is accompanied by an infinite family of poles at integer spacing, depicted in gray. The limit $n\to\infty$ amounts to encircling all poles.}
\label{fig:chi-contours1}
\end{figure}
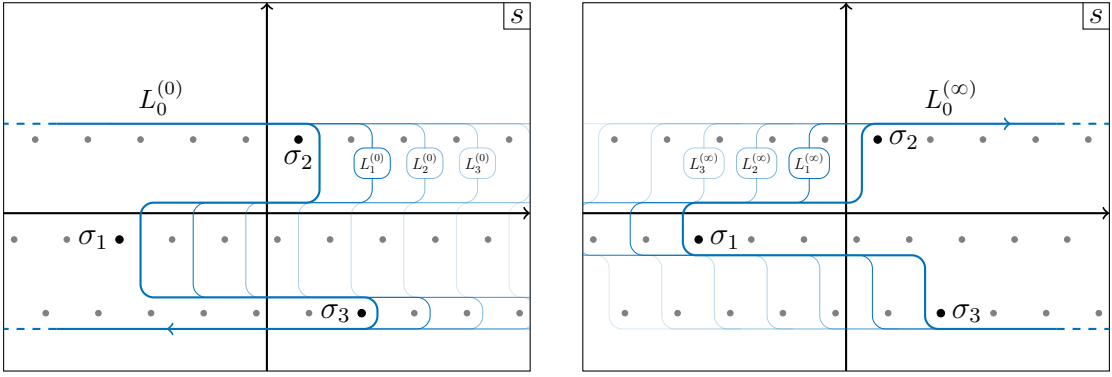
\begin{proof}
    Due to the rapid decay of the integrand along the integration contour, the integral along any $L_n^{(\infty)}$ can be evaluated as a sum over residues.
    This results in a linear combination of the $N$ Floquet solutions $\mathbf{F}^{(\infty)}$ with the principal parts truncated at order $-n$.
    Taking the limit $n\to\infty$ then amounts to summing over the infinite terms in the principal part.
    Convergence of the semi-infinite sum over residues for finite $n$, as well as the limit $n\to\infty$ is equivalent to convergence of the Floquet series as discussed in Section~\ref{sec:Floquet}, and follows from the vanishing of the quantum Wronskian.
\end{proof}
In Section~\ref{Section:canonical basis}, the possible exponential behavior of the solutions of the oper equation~\eqref{oper-z} was discussed (see equation~\eqref{alpha-beta}). 
We now prove that the oper solutions $\chi^{(\infty)}, \chi^{(0)}$ defined above exhibit the maximal possible decay (that is, associated with the root of unity $\alpha_0=1$).
\begin{propn}\label{maximal-decay-prop}
    The functions $\chi^{(0,\infty)}$ defined in eq.~\eqref{eq:Stokes-max} have the leading asymptotics
    \begin{subequations} \label{eq:chi-0-inf-asymptotics}
    \begin{align}
        \chi^{(0)}(w') &\sim \frac{(-1)^{N+1}}{(\pi\ii)^N}\sqrt{\frac{(2\pi)^{N-1}}{N}}e^{-N{w'}^{-1/N}}{w'}^\frac{N-1}{2N}\,,\; &\abs{w'}\to0\,,\quad w \in \mathscr{D}^{(0)} 
        \\
        \chi^{(\infty)}(w) &\sim \frac{1}{(\pi\ii)^N}\sqrt{\frac{(2\pi)^{N-1}}{N}}e^{-Nw^{1/N}}w^{-\frac{N-1}{2N}}\,,\; &\abs{w}\to\infty\,,\quad w \in \mathscr{D}^{(\infty)}
    \end{align}
    \end{subequations}
    with $w=(\Lambda/\hbar)^Nz$ and $w'=(\hbar/\Lambda)^Nz$ as usual, and
    \begin{subequations} \label{asympt-sectors-intermediate-basis}
    \begin{align}
    \mathscr{D}^{(0)} &\coloneqq \bigcup_{\theta\in \left(-\frac{\pi}{2}-\frac{\pi}{N},\frac{\pi}{2}+\frac{\pi}{N}\right)}
    \left\{w' \in \mathbb{C} \,\middle|\, \Re\left( N{w'}^{1/N}e^{\ii\theta} \right) < R^{-1}\right\}, \\
        \mathscr{D}^{(\infty)}&\coloneqq \bigcup_{\theta\in \left(-\frac{\pi}{2}-\frac{\pi}{N},\frac{\pi}{2}+\frac{\pi}{N}\right)}\left\{w \in \mathbb{C} \,\middle|\, \Re\left( Nw^{1/N}e^{\ii\theta} \right)> R \right\}.
    \end{align}
    \end{subequations}
\end{propn}
\begin{proof}
    The proof follows immediately from their uniform expansion in Meijer $G$-functions detailed in Appendix~\ref{app:G-expansion-maximaldecay}. The opening angle of the sector in which the asymptotics are valid follows from the Borel plane analysis explained in Appendix~\ref{range-of-asympt:section}.
\end{proof}
The explicit construction of these solutions in terms of $q$-functions is useful, as we may now apply known relations between $q^\pm_{\bm\de}$ in order to infer certain properties of $\chi^{(0,\infty)}$.
It is natural to ask when a solution may have maximal decay simultaneously at zero and infinity.
This leads us to the following corollary.

\begin{cor}\label{quantiz+decay=Corollary}
    There exists a solution to the oper equation with simultaneous maximal decay at zero and infinity if and only if the quantization conditions~(\ref{q-cond-delta}) are satisfied for $\boldsymbol\delta = -\ii\hbar\boldsymbol\sigma$.
\end{cor}
\begin{proof}
    Since maximal decay at one end uniquely specifies a solution to the oper equation, the statement of simultaneous maximal decay is equivalent to $\chi^{(0)}\,\propto\,\chi^{(\infty)}$.
    From the integral representation~\eqref{eq:Stokes-max} it follows in turn that this is equivalent to the proportionality of the residues of $q^+_{\bm\de}$ and $q^-_{\bm\de}$. Recalling~(\ref{eq:little q}), this condition is equivalent to the quantization conditions~\eqref{q-cond-delta}.
\end{proof}

\subsection{Analytic Spectral Duality}

As noted in Section~\ref{oper-section}, the Baxter equation~\eqref{Baxter-equation} and the oper equation~\eqref{oper-z} are related by a formal Fourier transform (using the variable $x=\log z$). 
Generically, this cannot be extended to an analytic Fourier transform of solutions to these equations
without severe modifications of the integration contour, as in Figure~\ref{fig:chi-contours1}.
Validity of the quantization conditions will turn out to be sufficient, and probably also necessary
for the possibility to relate solutions of the Baxter equation to solutions of the oper equation 
by a Fourier-transformation. We may prove, on the one hand, the following result.

\begin{propn}\label{Fourier-transf-prop}
    The unique (up to normalization) maximally decaying solution $\chi^{(0)}(z)$ to the oper equation can be expressed as the Fourier transform
    \begin{equation}
        \chi^{(0)}(x) = -\int_{-\infty}^\infty\frac{d\lambda}{2\pi\hbar}\,q(\lambda)e^\frac{\ii x\lambda}{\hbar}\,,
    \end{equation}
    where $x=\log z$, and $q(\lambda)=q^+_{\bm\de}(\lambda)-\zeta q^-_{\bm\de}(\lambda)$ is a solution to the Baxter equation satisfying the quantization conditions.
\end{propn}
\begin{proof}
    Consider the following pair of integrals
    \begin{equation}
        I_\pm(z) = \int_{-\ii\infty}^{\ii\infty}\frac{ds}{2\pi\ii}\,q_{\bm\de}^\pm(-\ii\hbar s)z^s\,,
    \end{equation}
which converge when $\abs{\arg(\Lambda/\hbar)^{\pm N}z}<N\pi/2$ assuming w.l.o.g. that there are no poles on the imaginary line, that is $\Re\sigma_j\notin\BZ$ for all $j=1,\dots,N$ (otherwise simply deform the contour slightly away from the poles and move it back at the end).
Due to the super-exponential decay of $q_{\bm\de}^\pm(-\ii\hbar s)$ away from their poles and a strip around $s\to\pm\infty$ the vertical contours may be deformed to Hankel contours encircling all poles in the left/right half-plane respectively.\footnote{This is analogous to the two equivalent contours for Meijer $G$-functions $G^{N,0}_{0,N}$~\cite{Luke1969}.} 
The integrals $I_\pm$ may therefore be evaluated as sums over residues, with $I_+$ only picking up the poles in the left half plane, and $I_-$ only picking up the poles in the right half plane.
Since the quantization conditions are equivalent to the relation $\Res_{s=\frac{\ii}{\hbar}\delta_j+n}q_{\bm\de}^+(-\ii\hbar s) = \zeta\Res_{s=\frac{\ii}{\hbar}\delta_j+n}q_{\bm\de}^-(-\ii\hbar s)$ for all $j=1,\dots,N$, we see that
\begin{align}
    \chi^{(0)}(z) &= -\sum_{j=1}^N\sum_{n\in\BZ} \Res_{s=\frac{\ii}{\hbar}\delta_j+n}q_{\bm\de}^+(-\ii\hbar s) z^s = \zeta I_-(z) - I_+(z) \nonumber\\
    &= \int_{-\ii\infty}^{\ii\infty}\frac{ds}{2\pi\ii}\left[ \zeta q_{\bm\de}^-(-\ii\hbar s)-q_{\bm\de}^-(-\ii\hbar s) \right]z^s = -\int_{-\ii\infty}^{\ii\infty}\frac{ds}{2\pi\ii}\,q(-\ii\hbar s)z^s\,.
\end{align}
The result then follows after the changes of variables to $\lambda=-\ii\hbar s$ and $x=\log z$.
\end{proof}

Concerning a possible converse to this statement, let us make the following observations. 
Solutions $q(\lambda)$ to the Baxter equation which are entire will generically be unbounded along the integration contour, spoiling the convergence of the Fourier transform.
For generic choice of $\de_1,\dots,\de_N$ 
there exist solutions $q(\lambda)$  to the Baxter equation \rf{Baxter-equation}
which are decaying on a strip around the real line, but which 
will generically have poles at $\lambda=\delta_j$ for some $j=1,\dots,N$. This will  imply that its Fourier transform
will not solve the oper equation.
Imposing both the condition to be entire  and to have rapid decay simultaneously 
amounts to imposing the quantization conditions~\eqref{q-asym}.
On the oper side it is furthermore clear that generic solutions can not be obtained as ordinary Fourier transforms of $L^1$ functions, as they are unbounded, contradicting the Riemann-Lebesgue lemma.\footnote{This explains the necessity of taking a  limit in the integral representation~\eqref{eq:Stokes-max}, which allows us to avoid the aforementioned obstructions.} It therefore seems likely that an analytic implementation of spectral duality 
will require that the quantisation conditions hold. 

\subsection{Connection Matrix and Quantization Conditions}

Since the connection matrix is diagonal in the Floquet basis, the solution to the connection problem amounts to computing the change of basis between the Floquet and canonical bases. 
\begin{lem}
    The change of basis matrices $C_0^{(0,\infty)}(\bm\sigma)$ defined in eq.~\eqref{eq:F-to-Y-def} take the form
    \begin{equation}
        C^{(0)}(\bm\sigma) = C^{(\infty)}(\bm\sigma) \equiv \mathsf{C}(\bm\si) = \mathsf{D}(\bm\sigma)\mathsf{V}(\bm\sigma)\mathsf{M}(\bm\sigma)^{-1}\,,
    \end{equation}
    with
    \begin{equation}
        \mathsf{D}(\bm\sigma) = \ii^{N+1}\sqrt{\frac{N(\pi\Lambda)^{N-1}}{(2\hbar)^{N-1}}}\diag\left( \frac{e^{-\pi\ii(2\floor{N/2}+N)\sigma_1}}{\prod_{j=2}^N\sin\pi(\sigma_1-\sigma_j)},\dots,\frac{e^{-\pi\ii(2\floor{N/2}+N)\sigma_N}}{\prod_{j=1}^{N-1}\sin\pi(\sigma_N-\sigma_j)} \right).
    \end{equation}
    The matrix $\mathsf{V}(\bm\sigma)$ takes a (reflected) Vandermonde form
    \begin{equation}
        \mathsf{V}(\boldsymbol\sigma) =
        \begin{pmatrix}
            \Sigma_1^{N-1} & \Sigma_1^{N-2} & \dots & \Sigma_1^0 \\ 
            \Sigma_2^{N-1} & \Sigma_2^{N-2} & \dots & \Sigma_2^0 \\
            \vdots & \vdots && \vdots \\
            \Sigma_N^{N-1} & \Sigma_N^{N-2} & \dots & \Sigma_N^0
        \end{pmatrix},
    \end{equation}
    and $\mathsf{M}(\bm\sigma)$ can be computed in components in terms of the canonical monodromy matrix $M_0$ given in eq.~\eqref{eq:canonical-monodromy} as
    \begin{equation}
        [\mathsf{M}(\bm\sigma)]_{n,k} = \left[ M_0^{\ceil{N/2}-k} \right]_{n,0}\,,\qquad n=0,\dots,N-1\,,\qquad k=1,\dots,N\,.
    \end{equation}
\end{lem}
\begin{proof}
    Let us first consider the following intermediate bases
    \begin{equation} \label{eq:Stokes-basis-inf-def}
        \chi^{(0,\infty)}_k(z) \coloneqq \chi^{(0,\infty)}\left( ze^{2\pi\ii(\ceil{N/2}-k)} \right),\qquad k=1,\dots,N\,.
    \end{equation}
    From eq.~\eqref{eq:chi-0-inf-asymptotics} it is clear that the solution $\chi^{(\infty)}_k$ is asymptotic to the canonical solution $Y^{(\infty)}_{0,n}$ with $n=\ceil*{\tfrac{N}{2}}-k\mod N$ in a sector containing $\arg w=0$ or $\arg w'=0$ respectively at infinity or zero, up to an overall normalization factor.
    The same relation evidently holds between $\chi^{(0)}_k$ and $Y^{(0)}_{0,n}$ after substituting $w$ with $w'$.
    In particular this means that the solutions are linearly independent and therefore form two bases associated to infinity and zero respectively. 
    Furthermore, since the solution with maximal decay along $w\to+\infty$ or $w'\to0^+$ is unique, it follows that
    \begin{equation}
        \chi^{(0,\infty)}(z) = \frac{(-1)^{N+1}\ii}{\pi^N}\sqrt{\frac{(2\pi\hbar)^{N-1}}{N\Lambda^{N-1}}}Y^{(0,\infty)}_{0,0}(z)\,.
    \end{equation}
    The change of basis between the canonical bases $\mathbf{Y}^{(0,\infty)}_0$ and the intermediate bases $\bm\chi^{(0,\infty)}$ is therefore given by
    \begin{equation}
        \chi^{(0,\infty)}_k(z) = \frac{(-1)^{N+1}\ii}{\pi^N}\sqrt{\frac{(2\pi\hbar)^{N-1}}{N\Lambda^{N-1}}}\sum_{n=0}^{N-1}Y^{(0,\infty)}_{0,n}\left[ M_0^{\ceil{N/2}-k} \right]_{n,0}
    \end{equation}
    where $M_0$ is the monodromy matrix in the canonical basis $\mathbf{Y}^{(0,\infty)}_0$ given in eq.~\eqref{eq:canonical-monodromy}.
    It follows that we can write
    \begin{equation}
        \mathbf{Y}^{(0,\infty)}_0(z) = \ii(-\pi)^N\sqrt{\frac{N\Lambda^{N-1}}{(2\pi\hbar)^{N-1}}}\bm\chi^{(0,\infty)}(z)\mathsf{M}(\bm\si)^{-1}\,.
    \end{equation}
    What remains is then to compute the expansion of the intermediate basis $\bm\chi^{(0,\infty)}$ in terms of the Floquet bases $\mathbf{F}^{(0,\infty)}$.
    This can be easily done by observing that the integral representation~\eqref{eq:Stokes-max} can be recast as a sum over Floquet solutions by isolating the contribution from each infinite row of poles $\sigma_j+\BZ$.
    Since the monodromy in the Floquet basis is explicit, it easily follows that
    \begin{equation} \label{eq:chi-to-F}
        \chi^{(0,\infty)}_k(z) = -\frac{\ii^N}{\pi}\sum_{l=1}^N e^{2\pi\ii(\ceil{N/2}-k)\sigma_l}\frac{F^{(0,\infty)}_l(z)}{e^{N\pi\ii\sigma_l}\prod_{\substack{j=1\\j\neq l}}^N \sin\pi(\sigma_l-\sigma_j)}\,.
    \end{equation}
    It is useful to note that this gives a Vandermonde matrix up to an overall rescaling, that is defining $\bm\chi^{(0,\infty)} = \mathbf{F}^{(0,\infty)}\mathsf{W}(\bm\sigma)$ we have\footnote{Note that both the independence on $\Lambda$ as well as the choice of puncture are non-obvious, and need not hold in general. It is essentially a special feature of our definition of $\bm\eta$.}
    \begin{equation} \label{W-V-matrix}
        \mathsf{W}(\boldsymbol\sigma) = -\frac{\ii^N}{\pi}\diag\left( \frac{e^{-\pi\ii(2\floor{N/2}+N)\sigma_1}}{\prod_{j=2}^N\sin\pi(\sigma_1-\sigma_j)},\dots,\frac{e^{-\pi\ii(2\floor{N/2}+N)\sigma_N}}{\prod_{j=1}^{N-1}\sin\pi(\sigma_N-\sigma_j)} \right) \mathsf{V}(\boldsymbol\sigma)\,.
    \end{equation}
    Finally, composing these two transitions gives the expected result. 
\end{proof}
\begin{cor}\label{E-matrix-cor-quantization}
    The connection matrix $E(\bm\sigma)$ between the canonical bases $\mathbf{Y}^{(0,\infty)}_0$ is given by
    \begin{equation}
        E(\bm\sigma) = C^{(\infty)}(\bm\sigma)^{-1}T(\bm\sigma)C^{(0)}(\bm\sigma) = \mathsf{M}(\bm\sigma)\mathsf{V}(\bm\sigma)^{-1}T(\bm\sigma)\mathsf{V}(\bm\sigma)\mathsf{M}(\bm\sigma)^{-1}\,.
    \end{equation}
    Furthermore, the quantization conditions~(\ref{q-cond-delta}) are satisfied if and only if $E(\bm\sigma)\,\propto\,\mathds{1}$.
\end{cor}
\begin{proof}
    The first point is automatic.
    As $E(\bm\si)$ is conjugate to $T(\bm\si)$ we have $E(\bm\sigma)\,\propto\,\mathds{1} \iff T(\bm\sigma)\,\propto\,\mathds{1}$,
    which is equivalent to the quantization conditions~\eqref{q-cond-delta}.
\end{proof}
The last point has been interpreted in the introduction as a variant $(\textup{ALC})_{\BR}$ of the Analytic Langlands Correspondence.

\subsection{Quantization Conditions From An Antiholomorphic Symmetry}

An interesting reformulation of this condition can be found as follows.
We may observe that the oper equation~(\ref{oper-z}) is invariant under the transformation $z \rightarrow 1/z$ together with overall complex conjugation, provided that 
$\Lambda\in\BR$, $\hbar\in\BR$ and $E_k\in\BR$ for $k=2,\dots,N$. We may next note that the formal solution
\begin{align}
    \tilde{Y}^{(\infty)}_n(z) = \ii^{N+1}e^{-N\tfrac{\Lambda}{\hbar}\alpha_nz^{1/N}}\left(N\tfrac{\Lambda}{\hbar}\alpha_nz^{1/N}\right)^{\beta+1}\tp^{(\infty)}\big(\alpha_nz^{1/N}\big)\,,
\end{align}
is mapped by this symmetry transformation  to
\begin{equation}
     \tilde{Y}^{(0)}_n(\bar z) = (-\ii)^{N+1}e^{-N\tfrac{\Lambda}{\hbar}\alpha_{-n}\bar{z}^{-1/N}}\left(N\tfrac{\Lambda}{\hbar}\alpha_{-n}\bar{z}^{-1/N}\right)^{\beta+1}\tp^{(0)}\big(\alpha_{-n}\bar z^{-1/N}\big)\,.
\end{equation}
After Borel resummation on the sectors $W_0^{(\infty)}$, $W_0^{(0)}$ respectively, these can be both evaluated on the positive real line, contained in $D^{(\infty,0)}_0$~(\ref{Sokal-discs}). 
We thereby find that
\begin{equation}
\widebar{\mathbf{Y}}^{(\infty)}\big(\bar{z}^{-1}\big)\Big|_{z\in\BR_+}=\mathbf{Y}^{(0)}(\bar{z})\Big|_{z\in\BR_+}
=\mathbf{Y}^{(\infty)}(z)\Big|_{z\in\BR_+}\!\!\!\cdot E\,.
\end{equation}
We see that ${\mathbf{Y}}^{(\infty)}$ is symmetric under the symmetry
above up to a multiplicative constant if and only if the connection matrix $E$ is proportional to the identity.\footnote{Equivalently, the action of the antiholomorphic symmetry is trivial in $\operatorname{PSL}(N,\BC)={}^L\mathrm{SL}(N,\BC)$.}

{\bf Acknowledgements.} 
JB, GR, and JT are funded by the Deutsche Forschungsgemeinschaft (DFG, German Research Foundation) -- SFB 1624 -- ``Higher structures, moduli spaces and integrability'' -- 506632645. We would like to thank  Pavlo Gavrylenko, Alba Grassi, Francesco Mangialardi, Marcos Mari\~no and Tommaso Pedroni for useful discussions, and 
Pavel Etingof, Edward Frenkel, Alba Grassi, and Lotte Hollands for helpful 
comments on a preliminary draft of this paper.

\newpage

\appendix

\section{Yang-Yang Function}\label{W=Y appendix}

\begin{propn} The Yang-Yang function defined in eq.~(\ref{eq:Yang-Yang-def}) satisfies the differential equations
\begin{subequations}
\begin{align}\label{diff-sigma}
\frac{\partial \mathcal{Y}(\boldsymbol{\sigma}, \Lambda)}{\partial \delta_j} &= 2\pi\ii\, \eta_j(\boldsymbol{\sigma}, \Lambda)\,,\\
    \label{diff-Lambda} 
\ii\hbar\frac{\partial \mathcal{Y}(\boldsymbol{\sigma},\Lambda)}{\partial \log \Lambda^{2N}} &= u(-\ii\hbar\boldsymbol{\sigma}, \Lambda)\,.
\end{align}
\end{subequations}
\end{propn}

\begin{proof}
In order to verify \rf{diff-sigma}, let us begin by noting that the relation
\begin{equation}\label{dY-ddelta-known}
     2\pi\ii \,\eta_j(\boldsymbol{\delta}, \Lambda)=\log Q_{\boldsymbol{\delta}}^+(\delta_j)-\log Q_{\boldsymbol{\delta}}^-(\delta_j)\,,
\end{equation}
allows us to express $\eta_j(\boldsymbol{\delta}, \Lambda)$ in terms of the solutions $X_{\boldsymbol{\delta}}(\mu)$ of the NLIE~(\ref{NLIE}) by using the integral representations of $v_{\downarrow}, v_{\uparrow}$ given by (\ref{eq:v-fns}),
\begin{align}\label{eta-explicit}
    2\pi\ii \eta_k(\boldsymbol{\delta}, \Lambda) &=\log\left(\frac{Q_{\boldsymbol{\delta}}^+(\delta_k)}{
    Q_{\boldsymbol{\delta}}^{-}(\delta_k)}\right) = \frac{2N\ii}{\hbar}\delta_k\left( \log\hbar-\log\Lambda \right)
    \nonumber \\
    &\quad\; + \sum_{j=1}^N\log\frac{\Gamma(1+\ii(\delta_k-\delta_j)/\hbar)}{\Gamma(1-\ii(\delta_k-\delta_j)/\hbar)} -\int_{\mathbb{R}}\frac{d \mu}{2\pi\ii}P_k(\mu)\log(1+X_{\bm\de}(\mu))\,,
\end{align}
where for brevity we have defined
\begin{equation}
    P_k(\mu) \coloneqq\frac{1}{\delta_k-\mu+\ii\hbar/2} + \frac{1}{\delta_k-\mu-\ii\hbar/2}\,.
\end{equation}
It is clear that the first terms in~(\ref{eta-explicit}) are obtained by $\tfrac{\partial \Ypert}{\partial \delta_k}$, where $\Ypert$ is defined in~(\ref{eq:Yang-Yang-def}). We now compute
\begin{align}
    \frac{\partial \Yinst}{\partial \delta_k} &= -\int_\BR \frac{d \mu}{2\pi\ii }\left[ \frac{\partial_k \log (X_{\bm\de}(\mu)\Theta(\mu))}{2}\log(1+X_{\bm\de}(\mu)) \right. \nonumber\\
    &\qquad\qquad\quad\;\, + \frac{\log(X_{\bm\de}(\mu)\Theta(\mu))}{2}\frac{X_{\bm\de}(\mu)}{1+X_{\bm\de}(\mu)}\partial_k\log X_{\bm\de}(\mu) \nonumber\\
    &\qquad\qquad\quad\;\, - \left.\vphantom{\frac{1}{2}}{\log(1+X_{\bm\de}(\mu))\partial_k\log X_{\bm\de}(\mu)} \right],
\label{partialdeltaY}\end{align}
where for brevity we denote $\partial_k \coloneqq  \tfrac{\partial}{\partial \delta_k}$. Defining for simplicity
 \begin{equation}\label{Kdef}
    K(\mu) \coloneqq\frac{ 2\ii\hbar}{\mu^2+\hbar^2}\,,
\end{equation}
the terms in the first two lines of \rf{partialdeltaY} can be written  using~(\ref{NLIE}) as 
integrals $I_1$ and $I_2$
\begin{align}
    I_1= &\int_\BR \frac{d \mu}{2\pi\ii }\frac{\partial_k \log (X_{\bm\de}(\mu)\Theta(\mu))}{2}\log(1+X_{\bm\de}(\mu)) \nonumber\\
    =&\frac{1}{2}\int_\BR\frac{d\mu}{2\pi\ii}  \int_\BR\frac{d\mu'}{2\pi\ii} K(\mu-\mu')[\partial_k\log(1+X_{\bm\de}(\mu'))]\log(1+X_{\bm\de}(\mu)) \nonumber\\ 
    = &\frac{1}{2}\int_\BR\frac{d\mu}{2\pi\ii}  \int_\BR\frac{d\mu'}{2\pi\ii} K(\mu-\mu')\frac{X_{\bm\de}(\mu')}{1+X_{\bm\de}(\mu')}[\partial_k\log(X_{\bm\de}(\mu'))]\log(1+X_{\bm\de}(\mu))\,,
\end{align}
\begin{align}
    I_2 &= \int_\BR\frac{d\mu}{2\pi\ii} \frac{\log(X_{\bm\de}(\mu)\Theta(\mu))}{2}\frac{X_{\bm\de}(\mu)}{1+X_{\bm\de}(\mu)}\partial_k\log X_{\bm\de}(\mu) \nonumber \\
    &=\frac{1}{2} \int_\BR\frac{d\mu}{2\pi\ii}  \int_\BR\frac{d\mu'}{2\pi\ii} K(\mu-\mu')\log(1+X_{\bm\de}(\mu'))\frac{X_{\bm\de}(\mu)}{1+X_{\bm\de}(\mu)}[\partial_k\log(X_{\bm\de}(\mu))]\,,
\end{align}
respectively.
So $I_1=I_2$, after exchanging $\mu'$ and $\mu$ in one of the integrals and observing that $K(\mu'-\mu)=K(\mu-\mu')$.
The third term reads
\begin{align}
    I_3 = -\int_\BR \frac{d \mu}{2\pi\ii }\log(1+X_{\bm\de}(\mu))\partial_k\log X_{\bm\de}(\mu) = &-\int_\BR \frac{d \mu}{2\pi\ii }\log(1+X_{\bm\de}(\mu))\partial_k\log (X_{\bm\de}(\mu)\Theta(\mu)) \nonumber\\ 
   & +\int_\BR \frac{d \mu}{2\pi\ii }\log(1+X_{\bm\de}(\mu))\partial_k\log \Theta(\mu)\,.
\end{align}
The first term in the equation above is clearly equal to $2I_1=I_1+I_2$, so these three contributions cancel.
It now suffices to observe that, recalling the definition~(\ref{Theta-def}), one has 
\begin{equation}\label{d_theta=P}
    \partial_k\log \Theta(\mu)  = P_k(\mu)\,,
\end{equation}
to conclude that
\begin{equation}
    \frac{\partial \Yinst}{\partial \delta_k} = -\int_{\mathbb{R}}\frac{d \mu}{2\pi\ii}P_k(\mu)\log(1+X_{\bm\de}(\mu))\,.
\end{equation}

We will next prove that  $\mathcal{Y}(\bm{\delta}, \Lambda)$ satisfies~\eqref{diff-Lambda}. We first need an explicit expression for $u(\boldsymbol{\delta},\Lambda)$. It was found in~\cite{Kozlowski:2010tv} that the zeros $\tau_j$ of $t_{\bm{\de}}(\lambda)$ defined in~(\ref{t-tau-def}) can be represented in terms of $X_{\bm\de}(\mu)$ by the formulae:
\begin{equation}\label{delta-to-tau}
    \sum_{j=1}^{N}\tau_j(\boldsymbol\delta,\Lambda)^k = \sum_{j=1}^{N}\delta_j^{k} + k\int_\BR\frac{d\mu}{2\pi\ii}\left\{ (\mu+\ii\hbar/2)^{k-1} - (\mu-\ii\hbar/2)^{k-1} \right\}\log(1+ X_{\bm\de}(\mu))\,,
\end{equation}
where $k$ is an integer. Specializing to $k=2$ and $P=\sum_{j=1}^N\tau_j=0$, we can express the eigenvalue of the Hamiltonian of the Toda chain,  
\begin{align}\label{u-hamiltonian}
   -u\coloneqq E_{2}= &= \frac{P^2}{2} -\sum_{j=1}^N\left(\frac{p_j^2}{2} +\Lambda^2e^{x_{j}-x_{j+1}} \right), \qquad x_{N+1} \coloneqq x_1\,,
\end{align}
as a function of $\boldsymbol{\delta}$. One gets  from~\eqref{delta-to-tau} 
\begin{equation}\label{u-def}
    u(\boldsymbol{\delta}, \Lambda) = \frac{1}{2}\sum_{p=1}^{N}\delta^2_p +\ii\hbar \int_\BR \frac{d\mu}{2\pi\ii}\log(1+ X_{\bm\de}(\mu))\,.
\end{equation}
We now differentiate $\mathcal{Y}(\boldsymbol{\delta},\Lambda)$ to demonstrate agreement with~(\ref{u-def}). For $\Ypert$ we simply find
\begin{equation}
  \ii\hbar\frac{\partial \Ypert(\boldsymbol{\delta}, \Lambda)}{\partial \log \Lambda^{2N}} = \frac{1}{2}\sum_{j=1}^N\delta_j^2\,,
\end{equation}
which is the first term in ~\eqref{u-def}. Before proceeding with the instanton part, let us first compute
\begin{align}\label{dX-dLambda}
    \frac{\partial}{\partial \log \Lambda^{2N}} \log(1+X_{\bm\de}(\mu)) 
    &= \frac{X_{\bm\de}(\mu)}{1+X_{\bm\de}(\mu)}\left( \frac{\partial \log X_{\bm\de}(\mu)}{\partial \log \Lambda^{2N}}\right).
\end{align}
Using~(\ref{NLIE}), we also compute\footnote{Bringing the derivative under the integral sign is allowed as $X_{\bm\de}(\mu)$ is positive and decaying as $\mu^{-2N}$ as $\mu\rightarrow\infty$ on $\mathbb{R}$~\cite[Proposition 3]{Kozlowski:2010tv}. Then the integral~\eqref{NLIE} converges absolutely for small enough values of $\Lambda^{2N} > 0$. 
}
\begin{align}\label{dlogY-dLambda}
       \frac{\partial \log (X_{\bm\de}(\mu)\Theta(\mu))}{\partial \log \Lambda^{2N}} 
       &= \int_\BR \frac{d\mu'}{2\pi\ii}\,K(\mu-\mu')  \frac{X_{\bm\de}(\mu')}{1+X_{\bm\de}(\mu')}\left( \frac{\partial \log X_{\bm\de}(\mu')}{\partial \log \Lambda^{2N}}\right),
\end{align}
where again $K(\mu)$ is given by eq.~(\ref{Kdef}). We are now ready to compute
\begin{align}
   \ii\hbar\frac{\partial  \Yinst(\boldsymbol{\delta}, \Lambda)}{\partial \log \Lambda^{2N}} &= -\ii\hbar\int_\BR\frac{d\mu}{2\pi\ii}\left[ \frac{1}{2}\log 
    (X_{\bm\de}(\mu)\Theta(\mu)) \frac{X_{\bm\de}(\mu)}{1+X_{\bm\de}(\mu)}\left( \frac{\partial \log X_{\bm\de}(\mu)}{\partial \log \Lambda^{2N}}\right) \right. \nonumber\\
    &\qquad\qquad\qquad\;\,+ \frac{1}{2} \frac{\partial \log (X_{\bm\de}(\mu)\Theta(\mu))}{\partial \log \Lambda^{2N}}\log(1+X_{\bm\de}(\mu)) \nonumber\\ 
    &\qquad\qquad\qquad\;\,-\left. \log(1+X_{\bm\de}(\mu))\left(1 +\frac{\partial \log (X_{\bm\de}(\mu)\Theta(\mu))}{\partial \log \Lambda^{2N}}\right) \right].
\end{align}
where the first two terms in the integral come from differentiating of first term in eq.~\eqref{Yinst} and we have used eq.~\eqref{dX-dLambda}; the third term comes from the derivative of the dilogarithm, defined by
\begin{equation}
    \Li_2(z) = \int_z^0dt\,\frac{\log(1-t)}{t},
\end{equation}
after recalling the explicit dependence on $\Lambda$ of (\ref{Theta-def}).
Collecting the terms, we find
\begin{align}\label{dYinst-dLambda}
   \ii\hbar\frac{\partial \Yinst(\boldsymbol{\delta}, \Lambda)}{\partial \log \Lambda^{2N}}&=\ii\hbar\int_\BR\frac{d\mu}{2\pi\ii}\log(1+X_{\bm\de}(\mu))\nonumber\\
  &\quad\,-\frac{\ii\hbar}{2}\int_\BR\frac{d\mu}{2\pi\ii}\left\{\log 
    (X_{\bm\de}(\mu)\Theta(\mu))\frac{X_{\bm\de}(\mu)}{1+X_{\bm\de}(\mu)}\left(\frac{\partial \log 
    (X_{\bm\de}(\mu)\Theta(\mu))}{\partial \log \Lambda^{2N}}\right) \right. \nonumber\\
  &\qquad\qquad\qquad\quad\, \left. -\log(1+X_{\bm\de}(\mu))\frac{\partial \log 
    (X_{\bm\de}(\mu)\Theta(\mu))}{\partial \log \Lambda^{2N}}\right\}.
\end{align}
The first integral matches the second term of~\eqref{u-def}, so all we are left to do is to prove that the second integral vanishes. This consists of the difference of two terms. In the first, we substitute for 
$\log (X_{\bm\de}(\mu)\Theta(\mu))$ the NLIE~\eqref{NLIE}; for the second we plug in~\eqref{dlogY-dLambda}. Hence we find for the first term
\begin{equation}
    J_1 =\int_{\BR^2}\frac{d\mu'}{2\pi\ii}\frac{d\mu}{2\pi\ii}\,K(\mu-\mu')\log(1+X_{\bm\de}(\mu'))\frac{X_{\bm\de}(\mu)}{1+X_{\bm\de}(\mu)}\left(\frac{\partial \log (X_{\bm\de}(\mu)\Theta(\mu))}{\partial \log \Lambda^{2N}}\right),
\end{equation}
while the second term is
\begin{equation}
    J_2= \int_{\BR^2}\frac{d\mu'}{2\pi\ii}\frac{d\mu}{2\pi\ii}\log(1+X_{\bm\de}(\mu))K(\mu-\mu')\frac{X_{\bm\de}(\mu')}{1+X_{\bm\de}(\mu')}\left( \frac{\partial \log (X_{\bm\de}(\mu')\Theta(\mu'))}{\partial \log \Lambda^{2N}}\right).
\end{equation}
In the latter integral, it is now sufficient to rename $\mu' \to \mu, \mu \to \mu'$ and again use $K(\mu'-\mu) =K(\mu-\mu')$ to conclude that $J_1=J_2$; therefore the second integral in eq.~\eqref{dYinst-dLambda} vanishes. 
\end{proof}

\section{Analytic Properties of the Oper Solutions} \label{app:G-expansion}

In this appendix we provide additional details on some analytic properties of the Floquet solutions~\eqref{eq:Floquet-from-Baxter}, as well as the maximally decaying solutions~\eqref{eq:Stokes-max}.
In Section~\ref{app:decoupling} we study the (suitably rescaled) $\Lambda\to0$ limit of the oper solutions, while in sections~\ref{app:G-expansion-maximaldecay} and~\ref{app:floquet} we study their asymptotics at finite values of $\Lambda$.
In all cases we make contact with known special functions which correspond to flat sections of $\slNC$ opers on $C_{0,2}$ with one minimally irregular puncture, and one regular one.
More precisely, they are solutions to the differential equation
\begin{equation} \label{eq:oper-w}
    t(-\ii\hbar w\partial_w)\chi(w) = (\ii\hbar)^Nw\chi(w)\,.
\end{equation}
A large class of solutions is given by Meijer $G$-functions $G^{m,0}_{0,N}\left( b_1,\dots,b_N\,\middle|\,(-1)^{N+m}w \right)$ with $b_k = \ii\tau_k/\hbar$, with $G^{m,n}_{p,q}$ defined more generally as
\begin{equation} \label{Meijer-G-def}
    G^{m,n}_{p,q}\left(
    \begin{matrix}
        a_1,\dots,a_p \\
        b_1,\dots,b_q
    \end{matrix}
    \,\middle|\,z \right) = \int_L\frac{ds}{2\pi\ii}\frac{\prod_{k=1}^m\Gamma(b_k-s)\prod_{k=1}^n\Gamma(1-a_k+s)}{\prod_{k=m+1}^q\Gamma(1-b_k+s)\prod_{k=n+1}^p\Gamma(a_k-s)}z^s\,.
\end{equation}
For our purposes we will always have $p<q$, so that the contour $L$ can be taken as a clockwise Hankel contour surrounding all the poles of $\prod_{k=1}^m\Gamma(b_k-s)$.
A basis of solutions with diagonal monodromy can be obtained by setting $m=1$ and varying the ordering of $b_1,\dots,b_N$.
Up to an irrelevant phase, these reduce to
\begin{equation}
    w^{b_j}{}_0\tilde F_{N-1}\left( \{ 1+b_j-b_k \,|\, k\neq j \} \,\middle|\, (-1)^Nw \right),
\end{equation}
where ${}_p\tilde F_q$ is the regularized hypergeometric function, related to ordinary (generalized) hypergeometric functions ${}_pF_q$ by a change of normalization
\begin{subequations} \label{eq:pFq-def}
\begin{align}
    {}_p\tilde F_q\left( 
    \begin{matrix}
        \alpha_1,\dots,\alpha_p \\ \beta_1,\dots,\beta_q
    \end{matrix}
    \,\middle|\, w\right) &= \frac{1}{\prod_{k=1}^q\Gamma(\beta_k)} {}_pF_q\left( 
    \begin{matrix}
        \alpha_1,\dots,\alpha_p \\ \beta_1,\dots,\beta_q
    \end{matrix}
    \,\middle|\, w \right), \\
    {}_pF_q\left( 
    \begin{matrix}
        \alpha_1,\dots,\alpha_p \\ \beta_1,\dots,\beta_q
    \end{matrix}
    \,\middle|\, w \right) &= \sum_{n=0}^\infty \frac{(\alpha_1)_n\dots(\alpha_p)_n}{(\beta_1)_n\dots(\beta_q)_n}\frac{w^n}{n!}\,,
\end{align}
\end{subequations}
and $(\alpha)_n=\Gamma(\alpha+n)/\Gamma(\alpha)$ is the usual Pochhammer symbol.
Another solution of note is that with $m=N$, which is well-known to have maximal decay: that is, when $\sum_{k=1}^Nb_k=0$,
\begin{equation} \label{eq:Meijer-G-asymptotics}
    G^{N,0}_{0,N}(b_1,\dots,b_N\,|\,w) \sim \sqrt{\frac{(2\pi)^{N-1}}{N}}e^{-Nw^{1/N}}w^{-\frac{N-1}{2N}}\left( 1+O\big( w^{-1/N} \big) \right)\,,\qquad \abs{w}\to\infty\,,
\end{equation}
while $\abs{\arg w}<(N+1)\pi$.
For a more complete analysis of Meijer $G$-functions, see~\cite{Luke1969}.

\subsection{Decoupling Limit and Open Toda Chain} \label{app:decoupling}

Let us consider the limit of $\mathbf{F}^{(\infty)}$ for $\Lambda\to0$ with $w = (\Lambda/\hbar)^Nz$ fixed.
This can be thought of as separating the two punctures by a long cylinder, while simultaneously ``zooming in'' on a neighborhood of $z=\infty$, see figure~\ref{fig:C02-Lambda}.
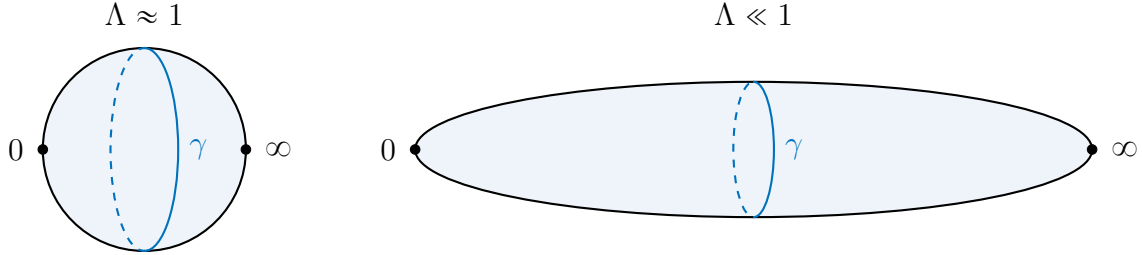
\begin{figure}[ht!]

\centering

\begin{tikzpicture}[scale=0.9]

\filldraw [thick, fill=RoyalBlue!5] (0,0) circle (1.5cm);
\filldraw [thick, fill=RoyalBlue!5] (9,0) ellipse (5cm and 1cm);
\begin{scope}
    \clip (0,-2) rectangle (2,2);
    \draw [thick, RoyalBlue] (0,0) ellipse (0.5cm and 1.5cm);
\end{scope}
\begin{scope}
    \clip (0,-2) rectangle (-2,2);
    \draw [thick, RoyalBlue, dashed] (0,0) ellipse (0.5cm and 1.5cm);
\end{scope}
\begin{scope}
    \clip (9,-2) rectangle (11,2);
    \draw [thick, RoyalBlue] (9,0) ellipse (0.3cm and 1cm);
\end{scope}
\begin{scope}
    \clip (9,-2) rectangle (7,2);
    \draw [thick, RoyalBlue, dashed] (9,0) ellipse (0.3cm and 1cm);
\end{scope}

\node at (0.8,0) {\color{RoyalBlue}{$\gamma$}};
\node at (9.6,0) {\color{RoyalBlue}{$\gamma$}};

\filldraw (-1.5,0) circle (2pt) node [left, xshift=-3pt] {$0$};
\filldraw (1.5,0) circle (2pt) node [right, xshift=3pt] {$\infty$};

\filldraw (4,0) circle (2pt) node [left, xshift=-3pt] {$0$};
\filldraw (14,0) circle (2pt) node [right, xshift=3pt] {$\infty$};

\node at (0,2) {$\Lambda\approx1$};
\node at (9,2) {$\Lambda\ll1$};

\end{tikzpicture}

\caption{The twice-punctured Riemann sphere $C_{0,2}$ and its dependence on the complex structure parameter $\Lambda$. Schematically, taking $\Lambda$ to zero stretches the sphere into an ellipse, separating the punctures and rendering their effects near the opposite end negligible.}
\label{fig:C02-Lambda}

\end{figure}
In this limit the oper equation~\eqref{oper-z} simply reduces to eq.~\eqref{eq:oper-w}, meaning that the singularity at $z=0$ is now regular.
Taking this limit for $\mathbf{F}^{(\infty)}$ and exchanging it with the sum\footnote{This is allowed since the decay of $Q^{\pm}$ along the imaginary axis, and therefore the absolute convergence of the Floquet series~\eqref{eq:Floquet-from-Baxter}, is uniform in $\Lambda$.} we have 
\begin{equation}
    \lim_{\Lambda\to0}F^{(\infty)}_j(w) = \sum_{n\in\BZ}\lim_{\Lambda\to0} \frac{K_-(\delta_j(\boldsymbol\tau,\Lambda)-\ii\hbar n)}{\prod_{k=1}^N\Gamma(1+\ii(\delta_j(\boldsymbol\tau,\Lambda)-\tau_k)/\hbar+n)}\left( e^{N\pi\ii}w \right)^{\ii\frac{\delta_j(\boldsymbol\tau,\Lambda)}{\hbar}+n}\,.
\end{equation}
From the Hill determinant~\eqref{Hill-det} it follows straightforwardly that $\lim_{\Lambda\to0}\boldsymbol\delta(\boldsymbol\tau,\Lambda) = \boldsymbol\tau$.
Furthermore, it follows from the determinant expression~\eqref{K+-det-definition} that $K_-(\lambda)$ becomes upper triangular, such that $\lim_{\Lambda\to0}K_-(\lambda) = 1$ away from its poles.
Note that when $n<0$ the argument of $K_-$ tends to its poles $\tau_j -\ii\hbar n$: however, since we have dropped the irregular singularity at zero, these terms must vanish anyway.\footnote{A careful analysis shows that these terms are suppressed by the Gamma functions in the denominator, which also become singular.}
The final result is therefore
\begin{align} \label{eq:F-infinity-limit}
    \lim_{\Lambda\to0}F^{(\infty)}_j(w) &= \left( e^{N\pi\ii}w \right)^{\ii\frac{\tau_j}{\hbar}}\sum_{n=0}^\infty \frac{\left( (-1)^Nw \right)^n}{n!\prod_{\substack{k=1 \\k\neq j}}^N\Gamma(1+\ii(\tau_j-\tau_k)/\hbar+n)} \nonumber\\
    &= \frac{\left( e^{N\pi\ii}w \right)^{\ii\frac{\tau_j}{\hbar}}}{\prod_{\substack{k=1\\k\neq j}}^N\Gamma(1+\ii(\tau_j-\tau_k)/\hbar)} \sum_{n=0}^\infty \frac{\left( (-1)^Nw \right)^n}{n!\prod_{\substack{k=1\\k\neq j}}^N(1+\ii(\tau_j-\tau_k)/\hbar)_n} \nonumber\\
    &= \left( e^{N\pi\ii}w \right)^{\ii\frac{\tau_j}{\hbar}}{}_0\tilde F_{N-1}\left( \{ 1+\ii(\tau_j-\tau_k)/\hbar\,\middle|\,k\neq j \} \,|\, (-1)^Nw \right)\,.
\end{align}
We therefore see that the basis $\mathbf{F}^{(\infty)}$ reduces to a simple expression in terms of well-known special functions in the vicinity of the puncture at $z=\infty$, as made precise by this limit.
In particular it is easy to verify that for $N=2$ these give the familiar modified Bessel functions $I_{\pm2i\tau/\hbar}(2\sqrt{w})$ up to a phase, thereby providing a natural generalization of the results of e.g.~\cite{Its:2014lga,Gavrylenko:2017lqz,Coman:2020qgf}.

A similar result can be obtained for $\mathbf{F}^{(0)}$ in the limit $\Lambda\to0$ with $w'=(\hbar/\Lambda)^Nz$ fixed, which amounts to making the singularity at $z=\infty$ regular, and zooming in on $z=0$, such that
\begin{equation} \label{eq:F-0-limit}
    \lim_{\Lambda\to0}F^{(0)}_j(w') = \left( e^{N\pi\ii}w' \right)^{\ii\frac{\tau_j}{\hbar}}{}_0\tilde F_{N-1}\left( \{1-\ii(\tau_j-\tau_k)/\hbar \,|\, k\neq j\}\,\middle|\, \frac{(-1)^N}{w'} \right).
\end{equation}

The maximally decaying solution $G^{N,0}_{0,N}$ may be expanded in the Floquet solutions ${}_0\tilde F_{N-1}$ through the following relation\footnote{A similar relation holds near $z=0$ by substituting $z\to1/z$ and $\mathbf{b}\to-\mathbf{b}$.}
\begin{align} \label{MeijerG-maximal-decay}
    G^{N,0}_{0,N}(b_1,\dots,b_N\,|\,z) &= \sum_{j=1}^Nz^{b_j}{}_0F_{N-1}\left( \{1+b_j-b_k\,|\,k\neq j\} \,\middle|\, (-1)^Nz \right)\prod_{\substack{k=1\\k\neq j}}^N \Gamma(b_k-b_j) \nonumber\\
    &=\sum_{j=1}^N z^{b_j}{}_0\tilde F_{N-1}\left( \{1+b_j-b_k\,|\,k\neq j\}\,\middle|\,(-1)^Nz \right)\prod_{\substack{k=1\\k\neq j}}^N \frac{\pi}{\sin b_k-b_j}\,,
\end{align}
which arises by splitting the Hankel contour $L$ into $N$ individual Hankel contours surrounding only the poles of a single $\Gamma(b_k-s)$.
This further supports the intuition that $\mathbf{F}^{(0)}$ should be associated to the puncture at $z=0$ and its canonical bases $\mathbf{Y}^{(0)}_k$, and $\mathbf{F}^{(\infty)}$ to $z=\infty$ and $\mathbf{Y}^{(\infty)}_k$. 

The result closely mirrors the one for the full $\Lambda$-dependent solutions, as can be read off from eq.~\eqref{eq:chi-to-F} after substituting $\mathbf{b}\to\bm\si$ and adjusting normalization factors.
This ultimately follows from the fact that the full solutions also admit Meijer $G$-like integral representations, from which it can be shown that the change of basis matrices exhibit no explicit $\Lambda$ dependence.
The crucial observation will be that the finite $\Lambda$ asymptotics, as well as the limit $\Lambda\to0$, are essentially controlled by the large order behavior of the Floquet series, or equivalently the asymptotics of the integrand in the integral representation~\eqref{eq:Stokes-max}, as will be shown in the following sections.

Lastly, let us mention that the equation~\eqref{eq:oper-w} can be interpreted as the oper equation corresponding to the open Toda chain. 
In general, one can deform the interaction between $x_1$ and $x_N$ in~\eqref{eq:x-p-Hamiltonian} to be $\kappa\Lambda^2e^{x_N-x_1}$, where $\kappa$ is a free parameter (this is sometimes referred to as a twist, and crucially does not spoil the integrability of the system). 
Clearly, $\kappa=1$ is the case of our interest, while for $\kappa=0$ one recovers the Hamiltonian of the open Toda chain. 
In this case, the Baxter equation would read~\cite{Kozlowski:2010tv}
\begin{equation}
    t(\lambda)q(\lambda) = (\ii\Lambda)^Nq(\lambda+\ii\hbar)\,,
\end{equation}
whose formal Fourier transform, after rescaling $(\Lambda/\hbar)^Nz=w$, yields eq.~\eqref{eq:oper-w}.

\subsection{Maximally Decaying Solutions} \label{app:G-expansion-maximaldecay}

A crucial step in the discussion of the quantization conditions is the identification of the maximally decaying solution.
For this it is useful to provide an expansion of $\chi^{(0,\infty)}(z)$ in terms of known special functions, along with checks of uniform convergence, which allows us to read off their asymptotics.

\begin{lem}
    The function $\chi^{(\infty)}(z)$ defined in eq.~\eqref{eq:Stokes-max} admits the following expansion valid at least for $\abs{\arg((\Lambda/\hbar)^Nz)}<N\pi/2$
    \begin{subequations} \label{eq:chi-Mittag-Leffler}
    \begin{align} 
        \chi^{(\infty)}(z) &= \frac{1}{(\pi\ii)^N} \left[ G_{0,N}^{N,0}\left( \sigma_1,\dots,\sigma_N\,\middle|\, \left( \frac{\Lambda}{\hbar} \right)^Nz \right) \right. \nonumber\\
        &\qquad\qquad\;\,- \left. \sum_{i=1}^N\sum_{m=1}^\infty c^{(\infty)}_{i,m} G_{1,N+1}^{N+1,0}\left(
        \begin{matrix}
            \sigma_i-m+1 \\ \sigma_1,\dots,\sigma_N,\sigma_i-m
        \end{matrix}
        \,\middle|\, \left( \frac{\Lambda}{\hbar} \right)^Nz \right) \right], \label{eq:chi-Mittag-Leffler-infty} \\
    \intertext{where $c^{(\infty)}_{i,m} = \Res_{s=\sigma_i-m}(v_\downarrow(-\ii\hbar(s+1)))$.
    Similarly $\chi^{(0)}(z)$ admits the expansion}
        \chi^{(0)}(z) &= \frac{(-1)^{N+1}}{(\pi\ii)^N} \left[ G_{0,N}^{N,0}\left( -\sigma_1,\dots,-\sigma_N \,\middle|\, \left( \frac{\Lambda}{\hbar} \right)^Nz^{-1} \right) \right. \label{eq:chi-Mittag-Leffler-0} \\
        &\qquad\qquad\quad\;\;\,+ \left. \sum_{i=1}^N\sum_{m=1}^\infty c^{(0)}_{i,m} G_{1,N+1}^{N+1,0}\left( 
        \begin{matrix}
            -\sigma_i - m + 1 \\ -\sigma_1,\dots,-\sigma_N,-\sigma_i-m
        \end{matrix} 
        \,\middle|\, \left( \frac{\Lambda}{\hbar} \right)^Nz^{-1} \right) \right], \nonumber
    \end{align}
    \end{subequations}
    at least when $\abs{\arg((\Lambda/\hbar)^Nz^{-1})}<N\pi/2$, where $c^{(0)}_{i,m} = \Res_{s=\sigma_i+m}(v_\uparrow(-\ii\hbar s))$.
\end{lem}

\begin{proof}
    For simplicity we focus on the case of $\chi^{(\infty)}$, as the proof for $\chi^{(0)}$ is analogous.
    The function $v_\downarrow(-\ii\hbar(s+1))$ is meromorphic with simple poles only at $s=\sigma_i-m$, $i=1,\dots,N$, $m\in\BZ_{>0}$, and satisfies $v_\downarrow\to1$ as $\abs{s}\to\infty$ uniformly away from its poles~\cite{Kozlowski:2010tv}.
    Its residues decay super-exponentially since they are related to the $Q$-functions evaluated at the Hill zeros by
    \begin{equation}
        Q^-_{\bm\de}(-\ii\hbar(\sigma_i-m)) = \left( e^{\pi\ii}\frac{\Lambda}{\hbar} \right)^{N(\sigma_i-m)}\frac{c^{(\infty)}_{i,m}(-1)^{m-1}(m-1)!}{\prod_{\substack{j=1\\j\neq i}}^N\Gamma(1-\sigma_j+\sigma_i-m)} = \zeta_iQ^+_{\bm\de}(-\ii\hbar(\sigma_i-m))\,.
    \end{equation}
    Applying Stirling's formula to both sides (using the asymptotics stated below eq.~\eqref{Hill-det} for $Q^+_{\bm\de}$) and solving for $c^{(\infty)}_{i,m}$, it is then clear that decay is dominated by $c^{(\infty)}_{i,m}\sim m^{-2Nm}$ (up to at most exponential factors) at leading order as $m\to\infty$.
    We can therefore reconstruct $v_\downarrow$ from the following Mittag-Leffler expansion
    \begin{equation} \label{eq:Mittag-Leffler}
        v_\downarrow(-\ii\hbar(s+1)) = 1 + \sum_{i=1}^N\sum_{m=1}^\infty \frac{c^{(\infty)}_{i,m}}{s-\sigma_i+m}\,,
    \end{equation}
    which converges absolutely and uniformly on compact sets.
    This then leads to an expansion for $q^-_{\bm\de}(-\ii\hbar s)$, which we can plug into the integral transform~\eqref{eq:Stokes-max}
    \begin{equation}
        \chi^{(\infty)}(w) = \lim_{n\to\infty}\frac{1}{(\pi\ii)^N}\int_{L_n^{(\infty)}}\frac{ds}{2\pi\ii}\, \sum_{i=1}^N\sum_{m=1}^\infty c^{(\infty)}_{i,m}\frac{\prod_{j=1}^N\Gamma(\sigma_j-s)}{s-\sigma_i+m} w^s\,,
    \end{equation}
    where we recall $w=(\Lambda/\hbar)^Nz$ and used the relation $\Gamma(z)\Gamma(1-z) = \pi/\sin\pi z$.
    To show that the integral commutes with the infinite sum over $m$ we use dominated convergence.
    For this it is convenient to split the computation into three parts.
    First, we split off the first $m_0$ terms of the sum over $m$, where $m_0$ is chosen such that $\abs{s-\sigma_i+m}\geq1$ for all $m>m_0$, $i=1,\dots,N$, and $s\in \ii\BR$.
    For the remaining integral over the infinite sum, we further split off the poles with negative real part.
    This set contains at most finitely many poles coming from the Gamma functions, plus an infinite number of poles coming from $v_\downarrow$.
    However, by construction the latter each receive only a single contribution from the expansion~\eqref{eq:Mittag-Leffler}, and so exchanging the sum and the integral is trivial.
    Finally, we consider the remaining integral with the sum over $m>m_0$, and poles only in the right half-plane.
    Since there are no more poles going to $-\infty$, the limit $n\to\infty$ effectively drops out.
    Furthermore, in analogy with Meijer $G$-functions, for $\abs{\arg w}<N\pi/2$, the integration contour may be deformed to the imaginary line.
    It is then easy to verify that dominant convergence holds since
    \begin{equation}
        \abs*{\sum_{i=1}^N\sum_{m=m_0}^M c^{(\infty)}_{i,m}\frac{\prod_{j=1}^N\Gamma(\sigma_j-s)}{s-\sigma_i+m}w^s} \leq e^{\ii\arg ws}\sum_{i=1}^N\sum_{m=1}^\infty\abs{c^{(\infty)}_{i,m}}\prod_{j=1}^N\abs{\Gamma(\sigma_j-s)}<\infty\,,
    \end{equation}
    for all $M\geq m_0$ and $s\in \ii\BR$, since the right hand side is clearly integrable on account of the strong decay of the Gamma functions along the integration path.
    Exchanging the infinite sum and integral then produces eq.~\eqref{eq:chi-Mittag-Leffler} since $s-\sigma_i+m = -\Gamma(\sigma_i-m-s+1)/\Gamma(\sigma_i-m-s)$.
\end{proof}

There is much literature on the asymptotic expansions of Meijer $G$-functions~\cite{Luke1969,Fields-MeijerG}, which is especially well understood in the case $G_{p,q}^{q,0}$.
The expansions of interest for our purposes are
\begin{subequations}
\begin{align}
    G_{0,N}^{N,0}(\sigma_1,\dots,\sigma_N\,|\, w) &\sim \sqrt{\frac{(2\pi)^{N-1}}{N}}e^{-Nw^{1/N}}w^{-\frac{N-1}{2N}}\left( 1+O\big( w^{-1/N} \big) \right),\\
    G_{1,N+1}^{N+1,0}\left(
        \begin{matrix}
            \sigma_i-m+1 \\ \sigma_1,\dots,\sigma_N,\sigma_i-m
        \end{matrix}
        \,\middle|\, w \right) &\sim \sqrt{\frac{(2\pi)^{N-1}}{N}}e^{-Nw^{1/N}}w^{-\frac{N+1}{2N}}\left( 1+O\big( w^{-1/N} \big) \right),
\end{align}
\end{subequations}
which are both valid for $\abs{\arg w}<(N+1)\pi$. We then realize that the terms in the sum appearing in (\ref{eq:chi-Mittag-Leffler-infty}) are subleading.
What remains to show is that the subleading behavior of the correction terms remains subleading as we take the infinite sum.
\begin{lem}
    The infinite sum~\eqref{eq:chi-Mittag-Leffler-infty} converges uniformly in $z$ for $(\Lambda/\hbar)^Nz\in(0,\infty)$, and its leading asymptotic is given by that of the leading $G$-function.
    The same holds for~\eqref{eq:chi-Mittag-Leffler-0} when $(\Lambda/\hbar)^Nz^{-1}\in(0,\infty)$.
\end{lem}
\begin{proof}
    Again, we restrict to $\chi^{(\infty)}$, as the proof for $\chi^{(0)}$ is analogous.
    Uniform convergence can be proved via the Weierstrass M-test~\cite{Whittaker_Watson_2021}.
    Consider just a single term in the expansion~\eqref{eq:chi-Mittag-Leffler-infty}.
    Since we restrict to positive argument, we can freely deform the integration contour into a vertical line and write
    \begin{equation}
        \chi^{(\infty)}_{i,m}(w) \coloneqq G_{1,N+1}^{N+1,0}\left(
        \begin{matrix}
            \sigma_i-m+1 \\ \sigma_1,\dots,\sigma_N,\sigma_i-m
        \end{matrix}
        \,\middle|\, w \right) = \int_{-\ii\infty-a_{i,m}}^{\ii\infty-a_{i,m}}\frac{ds}{2\pi\ii} \frac{\prod_{j=1}^N\Gamma(\sigma_j-s)}{s-\sigma_i+m}w^s\,,
    \end{equation}
    where again $w=(\Lambda/\hbar)^Nz$ and $a_{i,m}$ is chosen such that the integration line lies to the left of all poles.
    W.l.o.g. let us assume that $a_{i,m} = m-\sigma_i+1$ and $\Re(a_{i,m})>0$, which is valid for all but finitely many $\chi^{(\infty)}_{i,m}$ (say $m>m_0$).
    We can then shift the integral so that the contour lies on the imaginary line and obtain the bound
    \begin{align}
        \abs{\chi^{(\infty)}_{i,m}(w)} &= \frac{1}{w^{\Re(a_{i,m})}}\abs*{\int_{-\ii\infty}^{\ii\infty} \frac{ds}{2\pi\ii} \frac{\prod_{j=1}^N\Gamma(1+\sigma_j-\sigma_i+m-s)}{s-1}w^s} \nonumber\\
        &\leq \frac{1}{w^{\Re(a_{i,m})}} \int_{-\infty}^\infty \frac{d\lambda}{2\pi} \prod_{j=1}^N\abs{\Gamma(1+\sigma_j-\sigma_i+m-\ii\lambda)}\,,
    \end{align}
where we have changed variables to $\lambda = -\ii s$ to make reality manifest.
By H\"older's inequality we then have
\begin{align}
    \abs{\chi^{(\infty)}_{i,m}(w)} &\leq \frac{1}{w^{\Re(a_{i,m})}}\prod_{j=1}^N\left\{ \int_{-\infty}^\infty\frac{d\lambda}{2\pi}\, \abs{\Gamma(1+\sigma_j-\sigma_i+m-\ii\lambda)}^N \right\}^\frac{1}{N} \nonumber\\
    &\leq \frac{1}{2\pi w^{\Re(a_{i,m})}}\prod_{j=1}^N \left\{ \Gamma(\Re(\sigma_j-\sigma_i)+m+1)^{N-2} \int_{-\infty}^\infty d\lambda\, \abs{\Gamma(\sigma_j-\sigma_i+m+1-\ii\lambda)}^2 \right\}^\frac{1}{N} \nonumber\\
    &= \frac{1}{w^{\Re(a_{i,m})}}\prod_{j=1}^N \left\{ \Gamma(\Re(\sigma_j-\sigma_i)+m+1)^{N-2}\frac{\Gamma(2(\Re(\sigma_j-\sigma_i)+m+1))}{2^{2(\Re(\sigma_j-\sigma_i)+m+1)}} \right\}^\frac{1}{N}\nonumber\\
    &\eqqcolon \frac{d_{i,m}}{w^{\Re(a_{i,m})}}\,,
\end{align}
where for the second inequality we have used $\abs{\Gamma(x+\ii y)}\leq\Gamma(x)$ for $x>0$, and the remaining integral is known~\cite{NIST:DLMF}.
This expression is dominated by the super-exponential growth of the Gamma functions, giving something of order $m^{Nm}$ as $m\to\infty$ (again up to at most exponential factors).
Since this does not compensate for the decay of $c^{(\infty)}_{i,m}$ we therefore have 
\begin{equation}
    \sum_{m=1}^\infty \abs{c^{(\infty)}_{i,m}\chi^{(\infty)}_{i,m}(w)} \leq \sum_{m=1}^\infty \frac{\abs{c^{(\infty)}_{i,m}d_{i,m}}}{w_0^{\Re(a_{i,m})}}<\infty\,,\qquad w\geq w_0\,,\qquad w_0\in(0,\infty)\,.
\end{equation}
Since $w_0$ can be chosen freely in the interval $(0,\infty)$ without spoiling convergence of the sum, we conclude that it is uniformly convergent for $w\in(0,\infty)$.
The statement about the asymptotics then follows immediately.
\end{proof}

\subsection{Floquet Solutions} \label{app:floquet}

Similar to the maximally decaying solution, the Floquet solutions also admit an integral representation in the following form
\begin{subequations}\label{Floquet-integral}
\begin{align}
    F^{(0)}_j(z) &= -\lim_{n\to\infty}\pi\int_{L_n^{(0)}}\frac{ds}{2\pi\ii}\frac{Q^+(-\ii\hbar s)}{e^{\pi\ii(s-\sigma_j)}\sin\pi(s-\sigma_j)}z^s\,,\\
    F^{(\infty)}_j(z) &= -\lim_{n\to\infty}\pi\int_{L^{(\infty)}_n}\frac{ds}{2\pi\ii}\frac{Q^-(-\ii\hbar s)}{e^{\pi\ii(s-\sigma_j)}\sin\pi(s-\sigma_j)}z^s\,,
\end{align}
\end{subequations}
where the contour is the same as in Figure~\ref{fig:chi-contours1} (note however that there is now only one row of poles at $\sigma_j+\BZ$), and well-definedness is proven in the same fashion.
However, the expansion in known special functions is complicated by the fact that the contours can no longer be deformed to the imaginary line, thereby obstructing the proofs of dominated and uniform convergence.
To understand the asymptotics of the Floquet solutions, we can instead make use of the explicit relation with the intermediate basis $\boldsymbol\chi^{(0,\infty)}$ defined in eq.~\eqref{eq:Stokes-basis-inf-def}, whose leading asymptotic behavior is known in a sector containing the positive real line. Specifically, we prove the following

\begin{lem} \label{lem:Floquet-asymptotics}
    The Floquet solutions $F^{(0,\infty)}_j$ defined in eq.~\eqref{eq:Floquet-from-Baxter} have leading asymptotics in the sectors $\abs{\arg w}<\pi/N$ and $\abs{\arg w'}<\pi/N$, with $w=(\Lambda/\hbar)^Nz$ and $w'=(\hbar/\Lambda)^Nz$, given by
    \begin{subequations} \label{eq:Floquet-asymptotics-even}
    \begin{align}
        F^{(\infty)}_j(w) &\sim e^{N\pi\ii\sigma_j}\sqrt{\frac{(2\pi)^{1-N}}{N}}e^{Nw^{1/N}}w^{-\frac{N-1}{2N}}\left( 1 + O\big( w^{1/N} \big) \right), &w\to\infty\,, \\
        F^{(0)}_j(w') &\sim e^{N\pi\ii\sigma_j}\sqrt{\frac{(2\pi)^{1-N}}{N}}e^{N{w'}^{-1/N}}{w'}^\frac{N-1}{2N}\left( 1 + O\Big( {w'}^{-1/N} \Big) \right), &w'\to0\,,
    \end{align}
    \end{subequations}
    when $N$ is even, and
    \begin{subequations} \label{eq:Floquet-asymptotics-odd}
    \begin{align}
        F^{(\infty)}_j(w) &\sim e^{N\pi\ii\sigma_j}\sqrt{\frac{(2\pi)^{1-N}}{N}} \left\{ e^{-\pi\ii\sigma_j}e^{N\left( e^{\pi\ii}w \right)^{1/N}}\left( e^{\pi\ii}w \right)^{-\frac{N-1}{2N}}\left( 1 + O\big( w^{1/N} \big) \right) \right. \\
        &\quad\;+ \left. e^{\pi\ii\sigma_j}e^{N\left( e^{-\pi\ii}w \right)^{1/N}}\left( e^{-\pi\ii}w \right)^{-\frac{N-1}{2N}}\left( 1 + O\big( w^{1/N} \big) \right) \right\}, \qquad\quad\;\; w\to\infty\,, \nonumber\\
        F^{(0)}_j(w') &\sim e^{N\pi\ii\sigma_j}\sqrt{\frac{(2\pi)^{1-N}}{N}} \left\{ e^{-\pi\ii\sigma_j}e^{N\left( e^{\pi\ii}w' \right)^{-1/N}}\left( e^{\pi\ii}w' \right)^\frac{N-1}{2N}\left( 1 + O\Big( {w'}^{-1/N} \Big) \right) \right. \\
        &\quad\;+ \left. e^{\pi\ii\sigma_j}e^{N\left( e^{-\pi\ii}w' \right)^{-1/N}}\left( e^{-\pi\ii}w' \right)^\frac{N-1}{2N}\left( 1 + O\Big( {w'}^{-1/N} \Big) \right) \right\}, \qquad w'\to0\,, \nonumber
    \end{align}
    \end{subequations}
    when $N$ is odd.
\end{lem}
\begin{proof}
    We again focus on the case of $F^{(\infty)}_j$, since $F^{(0)}_j$ is proven analogously.
    The function $F^{(\infty)}_j$ can be expanded in the intermediate basis $\boldsymbol\chi^{(\infty)}$ as
    \begin{equation}
        F^{(\infty)}_j(w) = \chi^{(\infty)}_i(w)\left[ \mathsf{W}(\boldsymbol\sigma)^{-1} \right]_{ij}\,,
    \end{equation}
    with the change of basis matrix $\mathsf{W}(\boldsymbol\sigma)$ given explicitly in eq.~\eqref{W-V-matrix}.
    Since it is the product of a diagonal matrix with a (reflected) Vandermonde matrix, it is straightforward to obtain
    \begin{equation}
        \left[ \mathsf{W}(\boldsymbol\sigma)^{-1} \right]_{ij} = (-1)^i\frac{\pi}{\ii^N}e^{\pi\ii(2\floor{N/2}+N)\sigma_j}e_{i-1}(\{\Sigma_k\,|\,k\neq j\})\prod_{\substack{k=2 \\k\neq j}}^N\frac{\sin\pi(\sigma_j-\sigma_k)}{\Sigma_j-\Sigma_k}\,,
    \end{equation}
    where $e_{i-1}$ are again the elementary symmetric polynomials.
    The dominant contribution as $w\to\infty$ in the sector $\abs{\arg w}<\pi/N$ comes from $\chi^{(\infty)}_N$ when $N$ is even, and both $\chi^{(\infty)}_1$ and $\chi^{(\infty)}_N$ when $N$ is odd,\footnote{Depending on the sign of $\arg w$, only one term is dominant, with the two terms intersecting at $\arg w=0$.} as is easily seen from the asymptotics~\eqref{eq:chi-0-inf-asymptotics}.
    Using the relations $(\Sigma_j-\Sigma_k)/\Sigma_k = 2ie^{\pi\ii(\sigma_j-\sigma_k)}\sin\pi(\sigma_j-\sigma_k)$ and $\sum_{k=1}^N\sigma_k=0$, it is straightforward to compute
    \begin{equation}
        \left[ \mathsf{W}(\boldsymbol\sigma)^{-1} \right]_{Nj} = 2^{1-N}\pi\ii e^{N\pi\ii\sigma_j}\,,
    \end{equation}
    when $N$ is even and
    \begin{equation}
        \left[ \mathsf{W}(\boldsymbol\sigma)^{-1} \right]_{1j} = 2^{1-N}\pi\ii e^{(N+1)\pi\ii\sigma_j}\,,\qquad \left[ \mathsf{W}(\boldsymbol\sigma)^{-1} \right]_{Nj} = 2^{1-N}\pi\ii e^{(N-1)\pi\ii\sigma_j}\,,
    \end{equation}
    when $N$ is odd.
    Plugging in the asymptotics~\eqref{eq:chi-0-inf-asymptotics} then yields the above result.
\end{proof}
The understanding of the asymptotics of the Floquet solutions provides additional insight on the normalization chosen in eq.~\eqref{eq:Floquet-from-Baxter}.
In Section~\ref{Monodromy-Stokes} it was noted that the normalization can be partially fixed by considering the double-scaled limits studied in Appendix~\ref{app:decoupling}.
However this leaves the possibility for multiplicative corrections of the form $1+O(\Lambda)$, which may spoil the $\Lambda$-independence of the change of basis matrix $C^{(0,\infty)}(\bm\si)$.
The above result allows us to show that for the choice made in eq.~\eqref{eq:Floquet-from-Baxter} this does not happen, that is asymptotically and for fixed monodromy, the Floquet solutions are well-approximated by their (suitably scaled) $\Lambda\to0$ limits~(\ref{eq:F-infinity-limit},~\ref{eq:F-0-limit}).
\begin{cor}
    The Floquet basis is asymptotic to 
    \begin{subequations} \label{Floquet-asymptotics}
    \begin{align}
        F^{(\infty)}_j(w) &\sim \left( e^{N\pi\ii}w \right)^{\sigma_j}{}_0\tilde F_{N-1}\left( \{ 1+\sigma_j-\sigma_k \,|\, k\neq j \}\,\middle|\, (-1)^Nw \right), &w\to\infty\,, \\
        F^{(0)}_j(w') &\sim \left( e^{N\pi\ii} w' \right)^{\sigma_j}{}_0\tilde F_{N-1}\left( \{ 1-\sigma_j+\sigma_k \,|\, k\neq j\} \,\middle|\, \frac{(-1)^N}{w'} \right), &w'\to0\,,
    \end{align}
    \end{subequations}
    in the sense that the limit of their ratios converges to one, at least in the sectors $\abs{\arg w}<\pi/N$ and $\abs{\arg w'}<\pi/N$ respectively.
    This coincides precisely with the right hand side of equations~(\ref{eq:F-infinity-limit},~\ref{eq:F-0-limit}) after substituting $\bm\tau\to-\ii\hbar\bm\si$.
\end{cor}
The proof follows straightforwardly from the comparison of the leading asymptotics given in Lemma~\ref{lem:Floquet-asymptotics} with those given for ${}_pF_q$ in~\cite{Luke1969}.

\section{Further Details on the Borel Transforms} \label{canonical-basis-appendix}

In this appendix we include the proofs of some results used in Section~\ref{Section:canonical basis}. 
As a corollary of our results, we prove the range of validity of the asymptotics of $\chi^{(0,\infty)}$ defined in eq.~\eqref{eq:Stokes-max}.

\subsection{\texorpdfstring{Proof of Lemma~\ref{lemma-gevrey}}{Proof of Lemma 1}}
\begin{proof}
     For definiteness, we specialize the proof to $\tilde{f}_0^{(\infty)}$: it follows immediately from the definition of $\tilde{f}_n^{(\infty)}$ that if $\tilde{f}_0^{(\infty)}$ is Gevrey-1, then so is $\tilde{f}_n^{(\infty)}$ for every $n=0,\dots,N-1$.
     
     Let us derive the differential equation obeyed by $\tilde{f}_0^{(\infty)}$. By plugging the ansatz~(\ref{formal-basis}) into~(\ref{oper-y}), and ordering the different powers of $y$, one finds
\begin{align}\notag
    &\left[ (\partial_y-1)^N-(-1)^N + y^{-1}P_{N-1}(\partial_y) + y^{-2}P_{N-2}(\partial_y;\mathcal{E}_2)+\dots+y^{-N}P_{0}(\mathcal{E}_2,\dots,\mathcal{E}_N) \right]\tilde{f}_0 \\
    & = \nu y^{-2N}\tilde{f}_0^{(\infty)},
 \label{ODE-for-asymptotic-expansion}   \end{align}
where we have denoted by $P_j(\partial_y; \mathcal{E}_2,\dots,\mathcal{E}_{N-j})$ a polynomial of degree $j$ in $y$ which depends linearly on the rescaled Hamiltonians $\mathcal{E}_1,\dots,\mathcal{E}_N$, and $\mathcal{E}_1=0$.

The Newton polygon at infinity of the differential equation~(\ref{ODE-for-asymptotic-expansion}) now has a segment of slope $0$ and one of slope $1$. Therefore, considering that all the coefficients are rational functions, there exists a power series solution without logarithmic terms~\cite{Fauvet-0-dim-partion-function}. We can make the $\nu$-dependence of $\tilde{f}_0$ explicit by expanding
\begin{equation}
    \tilde{f}_0^{(\infty)} = \sum_{k=0}^\infty\nu^{k}\tilde{F}_k
\end{equation}
where $\tilde{F}_k$ are again formal power series in $y^{-1}$. (Note that we can choose solutions of~(\ref{oper-y}) to be analytic at $\nu =0$.) Now we may observe that $\tilde{F}_0$, up to an overall (and possibly $\nu$-dependent) prefactor is in fact the asymptotic expansion $\tilde{f}(w^{1/N})$ for the Meijer G-function $ G_{0,N}^{N,0}(b_1,\dots,b_N \,|\, w)$ appearing in~\eqref{eq:Meijer-G-asymptotics}. As expected, it is known to contain no logarithmic terms~\cite{Luke1969}; moreover, it is known to be a series in powers of $w^{-1/N}$, or equivalently in integer powers of $y^{-1}$. Observe then that one must have
\begin{equation}
    \mathcal{D}_y\tilde{F}_{k+1}= y^{-2N}\tilde{F}_k
\end{equation}
where for brevity we have indicated by $\mathcal{D}_y$ the differential operator appearing on the left hand side of eq.~(\ref{ODE-for-asymptotic-expansion}). 
Then we infer that for every $k$, $\tilde{F}_k$ is a formal power series in $y^{-1}$; then so is $\tilde{f}_0^{(\infty)}$.

Let us now prove that the coefficients $a_k$ of $\tilde{f}_0$ grow at most factorially. Consider the equation~(\ref{ODE-for-asymptotic-expansion}) in the form $y^N\mathcal{D}_y\tilde{f}_0^{(\infty)} = \nu y^{-N}\tilde{f}_0^{(\infty)}$. The left-hand side of the equation consists of terms that take the form
\begin{equation}
    y^{p}\partial_y^q\tilde{f}_0^{(\infty)} =   y^{p}\partial_y^q \sum_{k=0}^{\infty}a_ky^{-k} =  \sum_{k=0}^{\infty}a_ky^{-k-q+p}P_q(k)
\end{equation}
where $P_q(k) = (-1)^b(k)_q$ is a polynomial of degree $q$ in $k$. We can then relabel $m = k -p+q$ and write 
\begin{equation}
     y^{p}\partial_y^q\tilde{f}_0^{(\infty)} =    \sum_{m=-p+q}^{\infty}a_{m+p-q}y^{m}P_q(m)
\end{equation}
where now $P_q(m)$ is a polynomial of degree $q$ in $m$. A recurrence relation for the coefficients $a_m$ can be obtained by comparing the coefficients of $y^{-m}$. Observing that $p-q \leq N-1$, which implies that for a fixed $k$, the coefficient $a_{m+N-k}$ will only multiply polynomials of degree less than $k$, it is easy to argue that one will have for the term $y^{-m}$ an equation of the form
\begin{equation}
     \sum_{k=1}^{N} a_{m+N-k}Q_k(m) = \nu a_{m-N} 
\end{equation}
where again $Q_k(m)$ denotes a polynomial in $m$ of degree $k$. Thus, the recurrence relation is of the type
\begin{equation}
    a_{m+N-1} = \frac{1}{Q_1(m)}\left\{ -\sum_{k=2}^{N} a_{m+N-k}Q_k(m) + \nu a_{m-N}\right\}.
\end{equation}
Thus, considering the degree of the polynomials in the above formula, we can argue that the ratio between $a_{m+k}$ and $a_{m}$ is of order $O(m^k)$: this growth is compatible with $|a_{m}|\leq AB^m (m!)^r$  for $r \geq 1$ where $A,B$ are constants that depend on $\mathcal{E}_2,\dots,\mathcal{E}_N$ and $\nu$.
\end{proof}
\subsection{\texorpdfstring{Proof of Lemma~\ref{lemma-borel-plane}}{Proof of Lemma 2}}
 \begin{proof}
 We will derive an equation for $\hp_0(\zeta)$ from our knowledge of the differential equation for $\tp_0$. It is easy to see that the differential equation obeyed by $\tp_0$ will have the same shape as that for $f_0^{(\infty)}$~(\ref{ODE-for-asymptotic-expansion}), although the polynomials appearing will have different coefficients. We recall that, when $\tp, \tilde\psi$ are formal series without constant term (which is our case), one has
 \begin{equation}
     \mathcal{B}[\partial_y\tp](\zeta) = -\zeta\mathcal{B}[\tp]\,,\qquad
     \mathcal{B}[\tp \tilde\psi](\zeta) = \int_0^{\zeta}d\zeta'\,\hp(\zeta-\zeta')\hat\psi(\zeta')\,,
 \end{equation}
 where $\hat\psi(\zeta) \coloneqq \mathcal{B}[\tilde\psi]$. Then, the Borel transform of~(\ref{ODE-for-asymptotic-expansion}) reads
 \begin{align}\label{Borel-ODE1}
  (-1)^N  [(\zeta +1)^N -1]\hp_0(\zeta) &=-\sum_{k=0}^{N-1}\frac{1}{k!}\int_0^{\zeta}d\zeta'\,(\zeta -\zeta')^k P_{N-k-1}(-\zeta'; \boldsymbol{\mathcal{E}}_k)\hp_0(\zeta') \nonumber \\ 
  &\quad\;+ \nu \int_0^{\zeta}d\zeta'\,\frac{(\zeta-\zeta')^{2N+1}}{(2N+1)!}\hp_0(\zeta') 
\end{align}
where by $\boldsymbol{\mathcal{E}}_k = (\mathcal{E}_{2},\dots,\mathcal{E}_{k+1})$ for $k \geq 1$ and $\bm{\mathcal{E}}_0 = \varnothing$. 

We proceed by contradiction. Consider an open set $\Omega_j$ in $\mathbb{C}$ containing the line $[0, \omega_{0,j} ]$ with $(1+\omega_{0,j})^N -1 =0, \omega_{0,j} \neq 0$, and assume that $\hp_0(\zeta)$ is analytic on $\Omega_j$. Then the convolution integrals appearing on the right hand side above are analytic functions of $\zeta$ on $\Omega_j$~\cite[Lemma 5.3]{Sauzin:2014qzt}. 
We conclude that, under the assumption that $\hp_0(\zeta)$ is analytic at $\omega_{0,j}$, the right hand side of~(\ref{Borel-ODE1}) is also analytic at $\omega_{0,j}$. But then the same equation~(\ref{Borel-ODE1}) will imply that $\hp_0(\zeta)$ has poles at $\omega_{0,j}$, leading to a contradiction.

The second part of the lemma follows from the definition of $\tp_{n}$ as being equal to $\tp_0$ with $y$ replaced by $y\alpha_n$: upon applying the definition of the Borel transform~(\ref{def-Borel-transform}) one reaches the desired result.
 \end{proof}

\subsection{Asymptotics of Maximally Decaying Solution}\label{range-of-asympt:section}

In this section we employ our previous Borel plane analysis to justify the range of validity of the asymptotics expansion of~(\ref{eq:chi-0-inf-asymptotics}).

The maximally decaying solution $\chi^{(\infty)}(z)$ defined by~(\ref{eq:Stokes-max}) has been proved to have maximal decay at $\infty$ on the real line in Proposition~\ref{maximal-decay-prop}. Having maximal decay thereon, it must be proportional to the canonical basis element $ Y^{(\infty)}_{k,0}(y)$, obtained by resummation of $\tp_0$ in a sector $W_k^{(\infty)}$ with $k \neq 1,\dots,N$: indeed in such a way the Borel  resummation can be evaluated on a sector that either contains or has a boundary on the real line. Let us call $\varphi_0(y)$ the holomorphic function obtained by assembling the Borel resummation of $\tp_0$ in the arc of directions $\arg\omega_{0,N-1} = -\tfrac{\pi}{2}-\tfrac{\pi}{N}< \theta< \tfrac{\pi}{2}+\tfrac{\pi}{N} = \arg\omega_{0,1}$. Then a corollary of Lemma~\ref{lemma-gevrey} is
\begin{cor}
    In closed subsectors of the domain
    \begin{equation} \label{asymptotic-range}
        \mathscr{D}_y^{(\infty)}\coloneqq\bigcup_{\theta\in \left(-\frac{\pi}{2}-\frac{\pi}{N},\frac{\pi}{2}+\frac{\pi}{N}\right) }\left\{y \in \mathbb{C} \,\middle|\, \Re{(ye^{\ii\theta})}> R\right\}
    \end{equation}
    one has  $\varphi_0(y) \sim_1 \tp_0$, where by $\sim_1$ we indicate a Gevrey-1 asymptotics.
\end{cor}
\begin{proof}
That the Borel sum is asymptotic to the corresponding formal series is known (see e.g.\cite{Sauzin:2014qzt}). We only need to find the correct domain on which the asymptotic expansions are valid. To this end, we observe that
\begin{equation}
   \phi_{\textup{max}}\coloneqq  \min_{j=1,\dots,N-1}\abs{\arg \omega_{0,j}} = \frac{\pi}{2}+ \frac{\pi}{N}.
\end{equation}
It then suffices to recall that the Laplace transform at an angle $\theta$ defines a holomorphic function on a half-plane $\mathbb{H}(\theta)$ centered on the line $\mathbb{R}^+ e^{-\ii\theta}$ (as in (\ref{half-plane})). Because the rotation of the line of resummation implements an analytic continuation (in the opposite direction), we conclude that assembling together resummation at angles $\theta \in (-\phi_{\textup{max}},\phi_{\textup{max}})$ we obtain an holomorphic function $\varphi_0(y)$ on the sector
\begin{equation}
    \mathscr{D}^{(\infty)}_y \coloneqq \bigcup_{\theta \in (-\phi_{\textup{max}},\phi_{\textup{max}}) } \mathbb{H}(\theta)
\end{equation}
admitting $\tp_0$ as uniform asymptotic expansion on closed subsectors thereof. The asymptotics is of  Gevrey-type $1$ because that is the Gevrey class of $\tp_0$. 
\end{proof}
 Thus, reinstating the coordinate $w = N^{-N}y^N$ one obtains the domain $\mathscr{D}^{(\infty)}$ of~(\ref{asympt-sectors-intermediate-basis}). The validity of the asymptotics on $w' \in \mathscr{D}^{(0)}$ can be argued analogously.

\bibliographystyle{ytamsalpha}
\bibliography{biblio}

\end{document}